\pdfoutput=1
\documentclass[12pt,a4paper]{amsart}
\usepackage[a4paper,inner=2.5cm,outer=2.5cm,top=2.5cm,bottom=2.5cm]{geometry}
\usepackage{amsmath,amssymb,amsthm,enumerate,mathtools,stmaryrd,longtable
}
\usepackage{hyperref}
\hypersetup{colorlinks=true,linkcolor=blue,citecolor=teal,filecolor=magenta,urlcolor=cyan}

\usepackage{makecell}
\usepackage{graphicx}

\usepackage{float}

\usepackage{extarrows}

\usepackage{xcolor}

\newcommand{\C}{\mathbb{C}}

\usepackage[all]{xy}

\theoremstyle{plain}
\newtheorem{theorem}{Theorem}[section]
\newtheorem{proposition}[theorem]{Proposition}
\newtheorem{corollary}[theorem]{Corollary}
\newtheorem{lemma}[theorem]{Lemma}

\theoremstyle{definition}
\newtheorem{definition}[theorem]{Definition}

\makeatletter
\@addtoreset{proofpart}{theorem}
\makeatother
\theoremstyle{remark}
\newtheorem{remark}[theorem]{Remark}

\newcommand{\Z}{\mathbb{Z}}
\newcommand{\cD}{\mathcal{D}}
\newcommand{\cS}{\mathcal{S}}

\newcommand{\cE}{\mathcal{E}}

\newcommand{\cW}{\mathcal{W}}
\newcommand{\cT}{\mathcal{T}}

\newcommand{\VEV}[1]{{\big\langle 0 \big| {#1} \big| 0 \big\rangle}}

\newcommand{\res}{\mathop{\rm res}}

\newcommand{\vev}[1]{\langle 0 | {#1} | 0 \rangle}

\newcommand{\Res}{\mathop{\rm Res}}

\title{Symplectic duality for topological recursion}

\author[B.~Bychkov]{Boris~Bychkov}
\address{B.~B.: current affiliation: Department of Mathematics, University of Haifa, Mount Carmel, 3498838, Haifa, Israel; and Faculty of Mathematics, National Research University Higher School of Economics, Usacheva 6, 119048 Moscow, Russia}
\email{bbychkov@hse.ru}

\author[P.~Dunin-Barkowski]{Petr~Dunin-Barkowski}
\address{P.~D.-B.: Faculty of Mathematics, National Research University Higher School of Economics, Usacheva 6, 119048 Moscow, Russia; HSE--Skoltech International Laboratory of Representation Theory and Mathematical Physics, Skoltech, Bolshoy Boulevard 30 bld. 1, 121205, Moscow, Russia; and NRC ``Kurchatov Institute'' -- ITEP, 117218 Moscow, Russia}
\email{ptdunin@hse.ru}

\author[M.~Kazarian]{Maxim~Kazarian}
\address{M.~K.: Faculty of Mathematics, National Research University Higher School of Economics, Usacheva 6, 119048 Moscow, Russia; and Center for Advanced Studies, Skoltech, Bolshoy Boulevard 30 bld. 1, 121205, Moscow, Russia}
\email{kazarian@mccme.ru}

\author[S.~Shadrin]{Sergey~Shadrin}
\address{S.~S.: Korteweg-de Vries Institute for Mathematics, University of Amsterdam, Postbus 94248, 1090 GE Amsterdam, The Netherlands}
\email{S.Shadrin@uva.nl}

\begin{document}

\begin{abstract} We consider weighted double Hurwitz numbers, with the weight given by arbitrary rational function times an exponent of the completed cycles. Both special singularities are arbitrary, with the lengths of cycles controlled by formal parameters (up to some maximal length on both sides), and on one side there are also distinguished cycles controlled by degrees of formal variables. In these variables the weighted double Hurwitz numbers are presented as coefficients of expansions of some differentials that we prove to satisfy topological recursion.
	
Our results partly resolve a conjecture that we made in~\cite{BDKS-fullysimple} and are based on a system of new explicit functional relations for the more general $(m,n)$-correlation functions, which correspond to the case when there are distinguished cycles controlled by formal variables in both special singular fibers. These  $(m,n)$-correlation functions are the main theme of this paper and the latter explicit functional relations are of independent interest for combinatorics of weighted double Hurwitz numbers. We also put our results in the context of what we call the ``symplectic duality'', which is a generalization of the $x-y$ duality, a phenomenon known in the theory of topological recursion.
\end{abstract}
	
	\maketitle

\setcounter{tocdepth}{3}
\tableofcontents

\section{Introduction}

The main results of this paper concern new cases of topological recursion~\cite{EynardOra-TopoRec} and blobbed topological recursion~\cite{BorotSha-Blobbed} that simultaneously generalize the results on topological recursion in the theory of two-matrix model~\cite{CEO,EO-x-y-1st,EynardBook,DOPS} (bi-colored maps) and double weighted Hurwitz numbers~\cite{ACEH,bychkov2020topological}, as well as their fully simple analogs~\cite{borot2018simple,BCGF21,BDKS-fullysimple}.

Basically, we consider enumeration of double weighted Hurwitz numbers with a possibility that some cycles over $0$ and $\infty$ are not marked and are controlled by extra formal variables. This problem emerged naturally in~\cite{BDKS-fullysimple}, see Conjecture 4.4 there, as a framework for a natural generalization of the so-called $x-y$ duality in the topological recursion~\cite{EO-x-y-1st,eynard2013xy,borot2018simple,BCGFLS-freeprob,Hock-x-y}. In the present paper we generalize this $x-y$ duality to something that we call ``symplectic duality'' which make sense for general double weighted Hurwitz problems


We obtain these new cases of topological recursion and blobbed topological recursion and the respective symplectic duality as straightforward corollaries of our study of the so-called \emph{$(m,n)$-point correlators}. The main body of our paper is devoted to proving a recursion formula for these $(m,n)$-point correlators and then proving loop equations and the projection property for them.

In the rest of the introduction we recall the basic definitions related to topological recursion and the context for the $x-y$ duality which emerges from the theory of two-matrix model and enumeration of maps. Then we introduce the general $(m,n)$-point functions and describe our results on them, namely the recursion formulas and the loop equations and the projection property for them. After that we describe the implications for the new cases of topological recursion, and then we introduce the notion of symplectic duality and discuss 
certain (partly conjectural) chains of symplectic dualities. 

\subsection{Topological recursion and the \texorpdfstring{$x-y$}{x-y} duality} Let us first recall the formulation of topological recursion and some of its basic properties.

\subsubsection{Topological recursion}

Topological recursion of Chekhov, Eynard, and Orantin
\cite{ChekhovEynard,EynardOra-TopoRec,EynardBook} is a recursive computational procedure. It associates to a small input data that consists of a Riemann surface $\Sigma$, a symmetric bi-differential $B$ on $\Sigma^2$ (the Bergman kernel), and two functions $x$ and $y$ on $\Sigma$, subject to some conditions, a system of symmetric differentials $\omega^{(g)}_{n}$ on $\Sigma^n$, $g\geq 0$, $n\geq 1$, as well as some constants $\omega^{(g)}_{0}$, $g\geq 0$.

We assume that $dx$ is meromorphic and has isolated simple zeros $p_1,\dots,p_N\in \Sigma$, $y$ is holomorphic near $p_i$ and $dy|_{p_i}\not=0$, $i=1,\dots,N$.  We also assume that $B$ is meromorphic with the only pole being the order $2$ pole on the diagonal of $\Sigma^2$ with the bi-residue equal to $1$.

The symmetric differentials $\omega^{(g)}_n$, $n\geq 1$, $2g-2+n$, are produced from the initial unstable differentials $\omega^{(0)}_1$ and $\omega^{(0)}_2$ by an explicit recursive procedure:
\begin{align}\label{eq:toprec}
	\omega^{(g)}_{n+1}(z_0,z_{\llbracket n \rrbracket}) \coloneqq \frac 12 \sum_{i=1}^N \res_{w\to p_i} \frac{\int_w^{\sigma_i(w)} \omega^{(0)}_2(z_0,\cdot)}{\omega^{(0)}_{1}(\sigma_i(w))-\omega^{(0)}_{1}(w)} \Bigg( \omega^{(g-1)}_{n+2}(w, \sigma_i(w), z_{\llbracket n \rrbracket})+
	\\ \notag
	\sum_{\substack{g_1+g_2=g,\ I_1\sqcup I_2 = \llbracket n \rrbracket \\ (g_1,|I_1|),(g_2,|I_2|)\not=(0,0)
	}}
	\omega^{(g_1)}_{|I_1|+1} (w,z_{I_1})\omega^{(g_2)}_{|I_2|+1} (\sigma_i(w),z_{I_2})
	\Bigg).
\end{align}
Here $\sigma_i$ is a deck transformation of $x$ near the point $p_i$, $i=1,\dots,N$, by $\llbracket n \rrbracket$ we denote the set $\{1,\dots,n\}$, 
and $z_I$ denotes $\{z_i\}_{i\in I}$ for any $I\subseteq \llbracket n \rrbracket$. If $g=0$, then we assume that $\omega^{(g-1)}_{n+2}=0$. Here and everywhere below, if not specified otherwise, a sum of the form $\sum_{I_1\sqcup\dots\sqcup I_k = A}$ is understood as a sum over ordered collections of sets which are allowed to be empty. The initial differentials are typically chosen as $\omega^{(0)}_1=ydx$ and $\omega^{(0)}_2=B$. Note however that there is a variety of choices of $\omega^{(0)}_1$ and $\omega^{(0)}_2$ which do not affect the recursion~\eqref{eq:toprec} that we discuss below.

The constants $\omega^{(g)}_0$ are defined by the inversion of the boundary insertion operator, and in the simplest cases for $g\geq 2$ the definition is
\begin{equation}
	\omega^{(g)}_0 \coloneqq \frac{1}{2-2g}\sum_{i=1}^N \res_{w\to p_i} \omega^{(g)}_1 (w) \int_{p_i}^w \omega^{(0)}_1,
\end{equation}
but it is different in genera $0$ and $1$, see~\cite[Section 4.2]{EynardOra-TopoRec} and needs a correction in some cases, see a discussion in~\cite{eynard2013xy} for the case of algebraic curves.

\subsubsection{Possible choices for \texorpdfstring{$x$}{x}, \texorpdfstring{$y$}{y}, and unstable differentials} In general, the choice of $\omega^{(0)}_1$ and $\omega^{(0)}_2$ and even the functions $x$ and $y$ is just a matter of convention. Assume we want to let the given $\omega^{(g)}_n$, $2g-2+n>0$, satisfy the recursion~\eqref{eq:toprec}, and we allow any choice of $x$, $y$, $\omega^{(0)}_1$ and $\omega^{(0)}_2$ such that Equation~\eqref{eq:toprec} is satisfied.

The function $x$ is used only to determine the points $p_1,\dots,p_N$, and the same points can be recovered quite often as the critical points of $x^{-1}$, $\log x$, $e^x$, or any other holomorphic function in~$x$. In these cases, $y$ can be changed accordingly such that $\omega^{(0)}_1$ is preserved. For instance, $ydx = (-y)d(-x)=ye^{-x} de^x = yx d\log x = -yx^2 dx^{-1}$ gives in some situations (always, if we stick to the chosen critical points $p_1,\dots,p_N$) alternative choices for $(x,y)$ that preserve  $\omega^{(0)}_1=ydx$, the involutions $\sigma$, and thus the recursion~\eqref{eq:toprec}.

But most important for this paper is that the convention $\omega^{(0)}_1 = y(z)dx(z)$ is not necessary, and often inconvenient. We can change $y$ to $y+F(x)$ for any function $F(x)$, and the recursion~\eqref{eq:toprec} is still satisfied for $\omega^{(0)}_1 = (y(z)+F(x(z))dx(z)$. One can consider it either as a change of $y$ with a subsequent change of $\omega^{(0)}_1$ such that the convention $\omega^{(0)}_1 = ydx$ is preserved, or, alternatively, one can think that $y$ is fixed but $\omega^{(0)}_1$ is chosen to be $(y(z)+F(x(z))dx(z)$ instead of $y(z)dx(z)$.

A similar freedom exists in a choice of $\omega^{(0)}_2$. It can be deformed by any holomorphic bi-differential in the variables $x(z_1)$ and $x(z_2)$ without affecting the result of the recursion. It is convenient sometimes to chose it in the form $\omega^{(0)}_2 = B(z_1,z_2)-dx(z_1)dx(z_2)/(x(z_1)-x(z_2))^2$ instead of just $\omega^{(0)}_2 = B(z_1,z_2)$.

Final remark concerns the rescaling $y\to \alpha y$ with a nonzero constant~$\alpha$. This implies the corresponding rescaling $\omega^{(g)}_n \to \alpha^{2-2g-n} \omega^{(g)}_n$ for all $g,n\geq 0$.

\subsubsection{Loop equations and blobbed topological recursion} The differential obtained by~\eqref{eq:toprec} and considered as a differential in the first argument is a global meromorphic $1$-form on~$\Sigma$. Moreover, all its possible poles are at $p_1,\dots,p_N$. The principal parts of the poles are determined by the so called linear and quadratic loop equations which are essentially an equivalent form of~\eqref{eq:toprec}. Assume that $\Sigma=\C P^1$ is rational. In this case, the requirement that~$\omega^{(g)}_n$ has no poles other than $p_1,\dots,p_N$ is called projection property. If it is satisfied then $\omega^{(g)}_n$ being a rational $1$-form is uniquely recovered from the principal parts of its poles, and the residual expression~\eqref{eq:toprec} is intended to realize the recovery procedure in a closed form, see~\cite{BorotSha-Blobbed,BEO}, and also a short exposition in~\cite[Section 1.1.3]{bychkov2020topological}.

One can consider a situation when for a given system of differentials $\omega^{(g)}_n$ the loop equations are satisfied but the projection property fails. These conditions taken together are called the \emph{blobbed topological recursion}~\cite{BorotSha-Blobbed}. Remark that the blobbed topological recursion is not formally speaking a recursion: there is no control on the principal parts of the poles of these forms outside $p_1,\dots,p_N$, so that $\omega^{(g)}_n$ is not uniquely determined by the imposed restrictions.

\subsubsection{Classical \texorpdfstring{$x-y$}{x-y} duality}
One of the recurring topics in the literature on topological recursion is the so-called $x-y$ duality. Namely, one can replace functions $x$ and $y$ by $\tilde x = y$ and $\tilde y = -x$, and use the same $\tilde \Sigma = \Sigma$ and $\tilde B = B$ to compute via the same procedure the differentials $\tilde\omega^{(g)}_{n}$, $g\geq 0$, $n\geq 1$, as well as the constants $\tilde\omega^{(g)}_{0}$, $g\geq 0$.

Under some conditions, see~\cite{EynardOra-TopoRec,EO-x-y-1st,eynard2013xy}, the statement on the $x-y$ duality reads $\omega^{(g)}_{0}=\tilde \omega^{(g)}_{0}$. In this form, it is often presented as a statement on ``symplectic invariance'' of the constants $\omega^{(g)}_{0}$. The latter interpretation refers to the ``symplectic structure'' $\Omega = dy \wedge dx = d\tilde y \wedge d \tilde x$, for which two different ways to integrate it to a $1$-differential are chosen: $ydx$ and $\tilde y d\tilde x$. Note that
these two choices for the $1$-differentials integrating $\Omega$ determine two choices of the setup for topological recursion mentioned above.

One more often occurring addendum to the theory of topological recursion in the context of the $x-y$ duality is the so-called $(m,n)$-differentials, unifying the systems of differentials constructed above: $\omega^{(g)}_{m,n}$ are $(m+n)$-differentials on $\Sigma^{m+n}$, $g,m,n\geq 0$, with $\omega^{(g)}_{m,0} = \omega^{(g)}_{m}$,  $m\geq 1$, $\omega^{(g)}_{0,n} = \tilde \omega^{(g)}_{n}$, $n\geq 1$, and $\omega^{(g)}_{0,0} = \omega^{(g)}_{0} = \tilde \omega^{(g)}_{0}$. These $(m+n)$-differentials are subject to a system of loop equations and play a crucial role in connecting the two sides of the $x-y$ duality. See more details below.

\subsubsection{Example: bi-colored maps and the two-matrix model}\label{sec:bmaps} Let $t=(t_1,\dots,t_d,0,0,\dots)$ and $s=(s_1,$ $\dots,s_e,0,0,\dots)$ be two sets of formal parameters.  Consider the partition function of the formal Hermitian two-matrix model:
\begin{align}
	\int_{\mathcal{H}_N} dM_1 dM_2 e^{-N (Tr(M_1M_2) -Tr V_1(M_1)-Tr V_2(M_2))}
\end{align}
Here $\mathcal{H}_N$ is the product of two spaces of Hermitian $N\times N$ matrices, $dM_1$ and $dM_2$ are the properly normalized Haar measures, and $V_1(M_1)\coloneqq \sum_{i=1}^{d} \frac{t_i}i M_1^i$ and   $V_2(M_2)\coloneqq \sum_{i=1}^{e} \frac{s_i}i M_2^i$ are two polynomials. See e.g. \cite{EO-x-y-1st} for more details.

The logarithm of the partition function of this matrix model enumerates bi-colored maps on genus $g$ surfaces (the genus is controlled by the parameter $N^{2-2g}$), that is, the ways to glue a genus $g$ surface out of black and white polygons along their sides, such that each edge of the resulting embedded graph is bounding a black and a white polygon. Each black (resp., white) $p$-gon carries the variable $t_p$, $p=1,\dots,d$ (resp., $s_p$, $p=1,\dots,s$), and black (resp., white) $p$-gons with more than $d$ (resp., $e$) sides are forbidden. Each of the resulting polygonal decompositions of a genus $g$ surface is considered up to isomorphisms and is counted with the weight equal to the inverse order of its automorphism group.

Define $\omega^{\bullet}_{m,n}(\xi_1,\dots,\xi_m,\zeta_{m+1},\dots,\zeta_{m+n})$ as the cumulants of this matrix model by
\begin{align} \label{eq:cumulants-2MM}
	\omega^{\bullet}_{m,n}\coloneqq \frac{\int_{\mathcal{H}_N} \prod_{i=1}^m Tr \frac{-d\xi_i}{\xi_i-M_1} \prod_{j=1}^n Tr \frac{-d\zeta_{m+j}}{\zeta_{m+j}-M_2} dM_1 dM_2 e^{-N (Tr(M_1M_2) -Tr V_1(M_1)-Tr V_2(M_2))}}{	\int_{\mathcal{H}_N} dM_1 dM_2 e^{-N (Tr(M_1M_2) -Tr V_1(M_1)-Tr V_2(M_2))}}
\end{align}
These $\omega^{(g),\bullet}_{m,n}$ are the expansions of some $(m,n)$-differentials as above corresponding to one of the possible choices of the $x-y$ duality in this case. They have a clear combinatorial interpretation: we consider bi-colored maps on possibly disconnected surfaces (the parameter $N$ controls the Euler characteristic), with $m$ (resp., $n$) distinguished white (resp., black) polygons. Distinguished polygons are ordered, one of their sides is marked (so these are the so-called rooted polygons), and the $i$-th distinguished $p$-gon is white and labeled by $-d\xi_i/\xi_i^{p+1}$ for $i=1,\dots,m$ (resp., is black and labeled by $-d\zeta_i/\zeta_i^{p+1}$ for $i=m+1,\dots,m+n$) instead of $t_p$ (resp., $s_p$).

The connected two-matrix model cumulants (with a sign twist), $(-1)^n\omega^{(g)}_{m,n}$, are exactly the system of $(m,n)$-differen\-tials for a certain particular $x-y$ duality~\cite{EO-x-y-1st,CEO,EynardBook,DOPS}. In this case, we have
\begin{align} \label{eq:x-y-FirstExample}
	x & = \xi; &
	y & = -\zeta; &
	\omega^{(0)}_1 & = ydx + \Big(\frac 1x + \sum_{i=1}^d t_i x^{i-1}\Big)dx;
	\\ \notag
	 \tilde x & = -\zeta
	 &  \tilde y & = -\xi
	 & \tilde\omega^{(0)}_1 & = \tilde yd\tilde x -\Big(-\frac 1{\tilde x} + \sum_{i=1}^e s_i (-\tilde x)^{i-1}\Big)d\tilde x,
\end{align}
where $x$ and $y$ are some global rational functions on $\mathbb{C}\mathrm{P}^1$.

\subsection{Results of the paper}
To introduce the objects studied in this paper, we pass to the language of vacuum expectation values and KP integrability (see Sect.~\ref{sec:Fock} below for more details; see also Section~\ref{sec:mnpoint} for the definitions directly in terms of Schur polynomials, from which the relation to Hurwitz numbers is clear).

\subsubsection{Vacuum expectation values and hypergeometric type correlator functions} \label{Sec:vaccum} Consider the charge~$0$ Fock space~$\mathcal{V}_0$ whose bosonic realization is presented as $\mathcal{V}_0\cong \mathbb{C}[[q_1,q_2,q_3,\dots]]$. Define the operators $J_k$, $k\in\Z$, acting on $\mathcal{V}_0$ as $J_k=k\partial_{q_k}$, $J_{-k}=q_k$  (the operator of multiplication by $q_k$) if $k>0$, and $J_0=0$. Given a formal power series $\hat\psi(\theta,\hbar)$ in $\theta$ and $\hbar^2$ such that $\hat\psi(0,0)=0$, we introduce also the operator $\mathcal{D}_{\hat\psi}$ acting diagonally in the basis of Schur functions indexed by partitions $\lambda$, $|\lambda| \geq 0$,
\begin{equation}
\mathcal{D}_{\hat\psi}\colon s_\lambda \mapsto e^{\sum_{(i,j)\in \lambda} \hat\psi(\hbar\,(j-i),\hbar^2)} s_\lambda. 	
\end{equation}
Let $| 0 \rangle$ denote $1\in \mathcal{V}_0$ and let $ \langle 0 | \colon \mathcal{V}_0 \to \mathbb{C}$ be the extraction of the constant term.

In these terms, the disconnected $(m,n)$-point functions $H_{m,n}^\bullet$, and differentials $\omega_{m,n}^\bullet$ of our interest are defined as the following vacuum expectation values (VEVs):
\begin{align}\label{eq:Hfirstmn}
	H^{\bullet}_{m,n} &= \VEV{\prod_{i=1}^m \Big(\sum_{k=1}^\infty \tfrac{X_i^{k}}{k} J_k \Big)e^{\sum_{i=1}^\infty \frac{t_iJ_i}{i\hbar}}\cD_{\hat\psi}
	e^{\sum_{i=1}^\infty \frac{s_iJ_{-i}}{i\hbar}}
	\prod_{j=m+1}^{m+n} \Big(\sum_{k=1}^\infty \tfrac{Y_{j}^k}{k}J_{-k}  \Big)
	},\\
	\omega^{\bullet}_{m,n}
   &=d_1\dots d_{m+n}H_{m,n}^\bullet\\
   &= \VEV{\prod_{i=1}^m \Big(\sum_{k=1}^\infty {X_i^{k}} J_k \Big)e^{\sum_{i=1}^\infty \frac{t_iJ_i}{i\hbar}}\cD_{\hat\psi}
	e^{\sum_{i=1}^\infty \frac{s_iJ_{-i}}{i\hbar}}
	\prod_{j=m+1}^{m+n} \Big(\sum_{k=1}^\infty {Y_{j}^k}J_{-k}  \Big)
	} \prod_{i=1}^m \frac{dX_i}{X_i} \prod_{j=m+1}^{m+n} \frac{dY_{j}}{Y_{j}}.\notag
\end{align}

Using inclusion-exclusion formulas one can pass to the so-called connected VEVs denoted by $\vev{\cdots}^\circ$, which are related to connected correlator functions $H_{m,n}$ (resp.,  forms $\omega_{m,n}$) in exactly the same way. There are automatic restrictions on the exponent of $\hbar$ entering the connected functions so that they admit the genus expansion, $H_{m,n}=\sum_{g\ge0}\hbar^{2g-2+m+n}H_{m,n}^{(g)}$ and we have
\begin{align}\label{eq:Hfirstmncon}
	\sum_{g\ge0}\hbar^{2g-2+m+n}H_{m,n}^{(g)} &= \vev{\prod_{i=1}^m \Big(\sum_{k=1}^\infty \tfrac{X_i^{k}}{k} J_k \Big)e^{\sum_{i=1}^\infty \frac{t_iJ_i}{i\hbar}}\cD_{\hat\psi}
	e^{\sum_{i=1}^\infty \frac{s_iJ_{-i}}{i\hbar}}
	\prod_{j=m+1}^{m+n} \Big(\sum_{k=1}^\infty \tfrac{Y_{j}^k}{k}J_{-k}  \Big)
	}^\circ\\\label{eq:Hfirstmnend}
	\omega^{(g)}_{m,n}
   &=d_1\dots d_{m+n}H_{m,n}^{(g)}
\end{align}

The functions $H_{m,n}^{(g)}$ and the forms $\omega_{m,n}^{(g)}$ are the main objects of our research. These are formal power series in the variables $X_1,\dots,X_m,Y_{m+1},\dots,Y_{m+n}$, while the quantities $t_1,t_2,\dots$, $s_1,s_2,\dots$, as well as the coefficients of the series $\hat\psi$ are regarded as parameters of the problem.

The cumulants of the matrix model~\eqref{eq:cumulants-2MM} correspond to the special case $\hat\psi(\theta)=\log(1+\theta)$ and $t_i=0$, $i>d$, $s_j=0$, $j>e$ (with the substitution $\xi = X^{-1}$, $\zeta = Y^{-1}$, and $N=\hbar^{-1}$), see~\cite{goulden2008kp,guaypaquet20152d,HarnadOrlov,Kazarian_2015,ALS}. In general, the coefficients of $H_{m,n}^{(g)}$ have the combinatorial meaning, dual to that of mentioned in Section \ref{sec:bmaps}, 
of generalized weighted double Hurwitz numbers enumerating ramified coverings of the sphere. The coverings in question have two distinguished ramification points, zero and infinity.  There are $m$ marked preimages of $0$ and $n$ marked preimages of $\infty$ whose ramification orders correspond to the exponents of $X$ and $Y$ variables, respectively. Besides, there are unmarked preimages of $0$ and $\infty$ whose ramification orders correspond to the indices of $t$ and $s$ variables, respectively. The ramification type over the points different from $0$ and $\infty$ is encoded in the coefficients of the series $\hat\psi$. This relation to counting ramified coverings becomes clear when the definitions are rephrased directly in terms of Schur polynomials, see Section~\ref{sec:mnpoint}.

\subsubsection{The results on \texorpdfstring{$(m,n)$}{(m,n)}-functions and topological recursion} Let us impose certain restrictions on the parameters of the problem. Namely, we assume that the $t$ and $s$ parameters are specialized such that they have only finitely many nonzero entries, $t=(t_1,\dots,t_d,0,0,\dots)$, $s=(s_1,\dots,s_e,0,0,\dots)$. Besides, we assume that $\partial_\theta\hat\psi(\theta,0)$ and the coefficient of every positive power of $\hbar$ in $\hat\psi$ are rational. With these assumptions we introduce the notion of \emph{spectral curve} of the problem (see~Sect.~\ref{def:spectralcurve}) which is $\Sigma=\C P^1$ with an affine coordinate~$z$ and holomorphic functions $X,Y$ on it possessing the following properties: $X$ is a local coordinate at $z=0$, $Y$ is a local coordinate at $z=\infty$, and the forms $dX/X$ and $dY/Y$ extend as global meromorphic (i.e. rational) $1$-forms on $\Sigma$. We assume that the parameters are chosen in a generic way so that the forms $dX/X$ and $dY/Y$ have only simple zeroes. In fact, the spectral curve is determined by the restriction $\psi(\theta)=\hat\psi(\theta,0)$ only. We may think, therefore, of $\hat\psi$ as an $\hbar^2$-deformation of the function~$\psi$ that does not affect the equation of the spectral curve.

We prove that for any triple $(g,m,n)$ with $2g-2+m+n>0$ the form $\omega_{m,n}^{(g)}$ written at the corresponding local coordinates and regarded as an $(m+n)$-differential on $\Sigma^{m+n}$ extends as a global meromorphic $(m+n)$-differential (Theorem~\ref{th:rat}).

Moreover, we prove certain equalities relating these forms for different triples of $(g,m,n)$ and show that these equalities allow one to compute $\omega_{m,n}^{(g)}$ inductively in a closed form (Theorems~\ref{th:formalHz} and~\ref{th:mainrecursion}).
As an example, let us demonstrate here a simplified version of this relation for the case $g=0$. It relates the functions $H^{(0)}_{m+1,n}$ and $H^{(0)}_{m,n+1}$ through the change between the variables $X=X_{m+1}$ and $Y=Y_{m+1}$ entering these functions implied by the spectral curve equation. Denote
\begin{equation}
DH^{(0)}_{m,n+1}=Y\partial_Y H^{(0)}_{m,n+1},\qquad m+n>1,
\end{equation}
and in the exceptional cases $m+n=1$ we set explicitly
\begin{equation}
DH^{(0)}_{0,2}(Y,Y_2)=\tfrac{Y(z)}{Y'(z)}\tfrac{1}{z_2-z},\qquad DH^{(0)}_{1,1}(X_1,Y)=\tfrac{Y(z)}{Y'(z)}\tfrac{z_1z^{-1}}{z-z_1}.
\end{equation}
Denote $M=\{1,\dots,m\}$, $N=\{m+2,\dots,m+n+1\}$. For $J=\{j_1,\dots,j_k\}\subset M$ we set $X_J=(X_{i_1},\dots,X_{j_k})$, and we treat $Y_J$ in a similar way. Then, for $m+n\ge2$ we have
\begin{multline}
H^{(0)}_{m+1,n}+H^{(0)}_{m,n+1}=-
\sum_{r=2}^{m+n}\frac{1}{r!}\sum_{j=0}^{r-2}\bigl(X\partial_X\bigr)^j\biggl(\\
\frac{X}{dX}\frac{dY}{Y}\;
[v^j]\tfrac{e^{-v\psi(\theta)}\partial_\theta^{r-1}e^{v\psi(\theta)}}{v}\Bigm|_{\theta=\Theta}
\sum_{\substack{M\cup N=\sqcup_{i=1}^r K_i~
K_i\ne\varnothing\\~I_i=K_i\cap M,~J_i=K_i\cap N}}
\prod_{i=1}^r DH^{(0)}_{|I_i|+1,|J_i|}(X_{I_i};X;Y_{J_i})\biggr)+{\rm const}.
\end{multline}
Here $\Theta=\Theta(z)$ is a certain Laurent polynomials entering (along with $X=X(z)$ and $Y=Y(z)$) the equation of spectral curve and ${\rm const}$ is a certain function in $X_M,Y_N$ but independent of~$z$. In order to apply this formula we observe that $e^{-v\psi(\theta)}\partial_\theta^{r-1}e^{v\psi(\theta)}$ is a polynomial of degree $r-1$ in $v$ without free term. We will show that this relation actually determines both $H^{(0)}_{m+1,n}$ and $H^{(0)}_{m,n+1}$ as rational functions in~$z$.

We prove also that under certain explicitly formulated additional assumptions on $\hat\psi$ (which are satisfied in the special case~\eqref{eq:PQh}--\eqref{eq:PQ} discussed below) the forms $\omega^{(g)}_{m,n}$ can be integrated so that $H^{(g)}_{m,n}$'s are also rational for all $g\ge0$.

The form $\omega_{m,n}^{(g)}$ has poles at zeroes of $dX/X$ with respect to each of the first $m$ arguments and it has poles at zeroes of $dY/Y$ with respect to each of the last $n$ arguments. We prove that the principal parts of these poles satisfy a series of loop equations, including the linear and quadratic ones (Theorem~\ref{th:formalWz} and Corollary~\ref{cor:HLE}).

The last assertion directly implies 
that the forms $\omega_{m,0}^{(g)}$ satisfy the blobbed topological recursion, and also the forms $\omega^{(g)}_{0,n}$ satisfy (another) blobbed topological recursion.

The last step in establishing the true topological recursion is the projection property. We can reformulate the problem in the following way. Assume we are given the function $\psi(\theta)$ such that $\psi'(\theta)$ is rational. It defines the spectral curve, and thus, two topological recursions on it, one with the poles at the critical points of $X$, and another one with the poles at the critical points of $Y$. The question is whether there exists an $\hbar^2$-deformation $\hat\psi$ of $\psi$ such that the forms obtained by the topological recursions coincide with the forms $\omega^{(g)}_{m,0}$ and $\omega^{(g)}_{0,n}$ respectively, defined as the corresponding VEV's? We answer this question positively in~\cite{bychkov2020topological} for a large variety of cases in the situation when all $t$-parameters are equal to zero. In this paper, we show that the answer obtained in~\cite{bychkov2020topological} can be extended to the setting of the present paper. Namely, we show that the projection property is satisfied (and thus, both topological recursions holds true) if $\hat\psi$ is chosen in the following form
\begin{equation}\label{eq:PQh}
e^{\hat\psi(\theta,\hbar)}=R(\theta)e^{\frac{P(\theta+\hbar/2)-P(\theta-\hbar/2)}{\hbar}},
\end{equation}
where $R$ is a rational function and $P$ is a polynomial (Theorem~\ref{th:projpr}). This function is an $\hbar^2$-deformation of~$\psi(\theta)$ defined by
\begin{equation}\label{eq:PQ}
e^{\psi(\theta)}=R(\theta)e^{P'(\theta)}.
\end{equation}
The requirement that $\psi'$ is rational is satisfied since $\psi'=R'/R+P''$. A choice of~$\psi$ in this form
covers most of the known cases of enumerative problems for different kinds of special Hurwitz numbers (ordinary, monotone, strictly monotone, $r$-spin Hurwitz numbers, Bousquet-M\'elou--Schaeffer numbers, enumeration of maps, hypermaps, etc., see~\cite{bychkov2020topological}). It also demonstrates the importance of $\hbar^2$-deformations (it is crucial, for example, for the study of $r$-spin Hurwitz numbers, see~\cite{DKPS,bychkov2020topological}). The only reasonable 
case which is not covered by our considerations is the one when the polynomial $P'$ is replaced by a rational function.

Remark that the concept of topological recursion was initially introduced in order to simplify explicit computation of correlator functions. However, even if it works its practical use in the cases considered in the present paper is quite restrictive, because in order to apply it one needs the knowledge of explicit positions of all critical points of the~$x$ function, and they are given usually by complicated algebraic equations. In opposite, recursions of Theorems~\ref{th:formalHz} and~\ref{th:mainrecursion} provide an efficient way to compute $H^{(g)}_{m,n}$ explicitly in a closed form for small particular values of $(g,m,n)$ and without any restrictive assumption on the initial data of the problem.

Remark also that even for those cases when the projection property is satisfied this property along with the loop equations is not sufficient to determine uniquely $\omega^{(g)}_{m,n}$ in the case when both $m>0$ and $n>0$. This is due to the fact that this form still has additional poles on the diagonals $z_i=z_j$ where $i\in\{1,\dots,m\}$ and $j\in\{m+1,\dots,m+n\}$, and we have no control on the principal parts of these poles at the moment. That is one of the reasons why the general theory of topological recursion at its present stage does not cover the case of $(g,m,n)$ differentials with arbitrary $m$ and $n$.

\subsubsection{Formal vs analytic setups}
The $n$-point functions $H^{(g)}_{m,n}$ and differentials $\omega^{(g)}_{m,n}$ regarded as formal power series are defined by~\eqref{eq:Hfirstmn}--\eqref{eq:Hfirstmnend} with no restriction on the sets of $(t,s)$-parameters and the series $\hat\psi$. The restrictions of the previous section are needed to guarantee the analytic properties of the extension of these functions to the spectral curve. Without these analytic properties the very discussion of topological recursion and loop equations has no meaning. However, the recurrence relations of Theorem~\ref{th:mainrecursion} and even the equation of the spectral curve make sense in the formal case. In fact, we first prove these relations in the setting of formal series in Sect.~\ref{sec:formalprel}--\ref{sec:formalfinal} with no restrictions imposed on the set of parameters. Then, we analyze in Sect.~\ref{sec:rat}--\ref{sec:projprop} the analytic properties of these function, and this requires to impose certain restrictions on the set of parameters formulated at the beginning of the previous section.

\subsubsection{Specializations}
In Section~\ref{sec:Ex} we discuss various specializations and examples of the results of the present paper. In particular, Section~\ref{sec:t0} is devoted to the $t=0$ specialization. In this case we do not get any new topological recursion results, as all the respective results for the respective $H_{m,0}^{(g)}$ functions are already proved in \cite{bychkov2021explicit,bychkov2020topological} (while $H_{0,n}^{(g)}$ functions are identically zero). However, this $t=0$ case is still interesting, as in this case we can write explicit closed formulas for the $H_{m,n}^{(g)}$ functions (not just a recurrence relation).

\subsection{A variety of \texorpdfstring{$x-y$}{x-y} dualities and the symplectic duality}
It would be useful to put the results formulated in the previous section to the general context of the $x-y$ duality for topological recursion and its extension that we call the symplectic duality.

\subsubsection{A variety of \texorpdfstring{$x-y$}{x-y} dualities and setups for topological recursion}
Let the forms  $\omega^{(g)}_{n}$ be defined by topological recursion with some initial data.
Recall that adding to $y$ an arbitrary function of $x$ does not change any $\omega^{(g)}_{n}$ for $2g-2+n\geq 0$ (and $\omega^{(0)}_{1}$ is not produced by the recursion, it is merely a convenient convention that it should be equal to $ydx$).  Thus, $y\to y+F(x)$ does not change the topological recursion, but it substantially changes the other side of the $x-y$ duality, since now $\tilde x = y + F(x)$ and $\tilde y =- x$. Note that despite the shift $y\to y+F(x)$ the invariant ``symplectic form'' $\Omega$ is still the same, $dx\wedge dy$.

This leads to a natural question whether there are reasonable, i.~e. non-artificial examples of chains of several setups for topological recursion related by $x-y$ dualities with suitable shifts of $y$'s as above. Non-artificial here means that we want the resulting $\omega^{(g)}_{n}$'s to have an enumerative meaning in some expansions. The only known example of this type in the literature with two different possible choices of the dual problem comes from the theory of maps / fully simple maps / the Hermitian two-matrix model, which we briefly recall below.

Remark that the authors of many papers on topological recursion (including those of the present paper) usually introduce some artificial corrections to the unstable functions (and to speculate on the philosophic meaning of these corrections), in order to represent studied relations in a shorter form. However, for the rest of this section, in order to avoid confusion with the variables to be exchanged by duality, we shall always list explicitly $x$, $y$, and $\omega^{(0)}_1$ as well as $\tilde x$, $\tilde y$, and $\tilde \omega^{(0)}_1$, and we never assume the convention $\omega^{(0)}_1 = ydx$ by default; instead, we will always represent the forms $\omega^{(0)}_1$ and $\tilde \omega^{(0)}_1$ in their original unmodified forms.


\subsubsection{Fully simple maps}
One of the instances of $x-y$ dualities for maps is given by~\eqref{eq:x-y-FirstExample} as discussed above. Another example of an $x-y$ duality for maps, for a different choice of $y$, is given by the enumeration of the so-called fully simple maps. A fully simple map is a bi-colored map with all distinguished polygons of just one fixed color and with all polygons of the other color being $2$-gons, with an extra condition that the distinguished polygons do not have common vertices (which implies that at least some $t_i\not=0$ for $i\geq 3$). See e.g. \cite{borot2018simple} for more details. The corresponding $n$-point generating function is given by
\begin{align}
	W^{\bullet}_{n} =  \VEV{\prod_{i=1}^n \Big(\sum_{j=1}^\infty J_j w_i^{j} \Big)\cD_{-\log(1+\theta)} e^{\sum_{i=1}^d \frac{t_iJ_i}{i\hbar}}\cD_{\log(1+\theta)}
		e^{\frac{J_{-2}}{2\hbar}}
		 \Big)
	}
\end{align}
With this new definition, we obtain a new $x-y$ duality statement for enumeration of maps. The corresponding (disconnected) $(m,n)$-differentials are given by
\begin{align}
	& \omega^{\bullet}_{m,n} =  \VEV{\prod_{i=1}^n \Big(\sum_{j=1}^\infty J_j w_{m+i}^{j} \Big)\cD_{-\log(1+\theta)} \prod_{i=1}^m \Big(\sum_{j=1}^\infty J_j X_i^{j} \Big) e^{\sum_{i=1}^d \frac{t_iJ_i}{i\hbar}}\cD_{\log(1+\theta)}
	e^{\frac{J_{-2}}{2\hbar}}
}  \\ \notag
& \qquad \quad \cdot \prod_{i=1}^m \frac{dX_i}{X_i} \prod_{i=1}^n \frac{dw_{m+i}}{w_{m+i}}.
\end{align}
In this case, the data for the $x-y$ duality is as follows:
\begin{align}
	x &= \frac{1}{X}; & y &= -w;& \omega_1^{0} &= ydx + \frac{dx}{x}; \\ \notag
	\tilde x &= w; & \tilde y & = -\frac 1X& \tilde \omega_1^{0} &= \tilde y\tilde dx -\frac{d\tilde x}{\tilde x},
\end{align}
where $x$ and $y$ are some global rational functions on $\mathbb{C}\mathrm{P}^1$, and the function $x$ and the differential $\omega_1^{0}$ are exactly the same as in Equation~\eqref{eq:x-y-FirstExample} for the choice of parameters $s_i=\delta_{i,2}$.

\subsubsection{Summary of two \texorpdfstring{$x-y$}{x-y} dualities}

Let us resume the observations made above. For a non-trivial choice of the parameters $t=(t_1,\dots,t_d,0,0,\dots)$ (with at least one $t_i\not=0$ for $i\geq 3$) and the fixed choice of the parameters $s=(0,1,0,0,\dots)$, we consider the system of $\omega^{(g)}_{n}$'s given by
\begin{align} \label{eq:usual-maps}
	\omega^{\bullet}_n = \VEV{ \prod_{i=1}^m \Big(\sum_{j=1}^\infty J_j X_i^{j} \Big) e^{\sum_{i=1}^d \frac{t_iJ_i}{i\hbar}}\cD_{\log(1+\theta)}
		e^{\frac{J_{-2}}{2\hbar}}
	}  \prod_{i=1}^m \frac{dX_i}{X_i} .
\end{align}
They satisfy the topological recursion for the input data given by the curve $\mathbb{C}\mathrm{P}^1$, the Bergman kernel $B=dz_1dz_2/(z_1-z_2)^2$ in some global coordinate $z$, some global rational function $x=x(z)$ such that $x(z)^{-1}$ can serve as a local coordinate at $z=\infty$ and such that Equation~\eqref{eq:usual-maps} gives the expansion at $z=\infty$ in this local coordinate $X=x(z)^{-1}$, and some choice of function $y$.
The enumerative meaning of the coefficients of the expansions of $\omega^{(g)}_n$ in $X$ at $z_1=\cdots=z_n=\infty$ in the local coordinates, $X_i=x(z_i)^{-1}$ is given by the count of the ordinary maps, that is, the bi-colored maps with $n$ distinguished white polygons such that all black polygons are two-gons.

The choice of $y$ allows ambiguity, $y=\omega^{(0)}_1/dx+F(x)$, and for at least two different choices of $F(x)$ the $x-y$ dual topological recursion produces the differentials whose expansions have meaningful enumerative interpretation. These two choices are $F_1(x) = -x^{-1}-\sum_{i=1}^d t_i x^{i-1}$ and $F_2(x)=-x^{-1}$.

The first $x-y$ dual topological recursion produces $n$-differentials $\tilde \omega^{(g)}_n$ whose expansions in some local coordinate $Y$ at $z=0$ are given by
\begin{align} \label{eq:Maps-dual}
	\tilde \omega^{\bullet}_n = (-1)^n \VEV{e^{\sum_{i=1}^d \frac{t_iJ_i}{i\hbar}}\cD_{\log(1+\theta)}
		e^{\frac{J_{-2}}{2\hbar}}
		 \prod_{i=1}^n \Big(\sum_{j=1}^\infty J_{-j} Y_i^{j} \Big)
	}  \prod_{i=1}^n \frac{dY_i}{Y_i},
\end{align}
and the enumerative meaning of these expansions is the count of bi-colored maps such that the white polygons are controlled by the variables $t=(t_1,\dots,t_d,0,0,\dots)$, the distinguished black polygons by the variables $Y_1,\dots,Y_n$, and all non-distinguished black polygons are two-gons.

The $x-y$ dual topological recursion obtained by the second choice produces $n$-differentials $\tilde \omega^{(g)}_n$ whose expansions in some local coordinate $w$ at $z=\infty$ are given by
\begin{align}
	\tilde \omega^{\bullet}_n = \VEV{\prod_{i=1}^n \Big(\sum_{j=1}^\infty J_j w_i^{j} \Big)\cD_{-\log(1+\theta)} e^{\sum_{i=1}^d \frac{t_iJ_i}{i\hbar}}\cD_{\log(1+\theta)}
		e^{\frac{J_{-2}}{2\hbar}}
	} \prod_{i=1}^n \frac{dw_i}{w_i},
\end{align}
and the enumerative meaning of these expansions is the count of fully simple  maps.

\subsubsection{2d Toda tau functions and weighted double Hurwitz numbers}

It is conjectured (see~\cite[Conjecture 3.13]{BCGFLS-freeprob} and also~\cite{Hock-x-y} for a genus $0$ result under some extra conditions) that the $n$-point functions or differentials for the two topological recursions related by the $x-y$ symmetry  are subject to certain universal functional relations, obtained from the operator $\cD_{\pm\log(1+\theta)}$ in the VEV formalism derived in~\cite{BDKS-fullysimple}.

However, from the point of view of weighted Hurwitz enumerative problems and hypergeometric KP or 2d Toda tau functions it is more natural to consider $\cD_\psi$ for any reasonable function $\hat\psi=\hat\psi(\theta,\hbar^2)$ (for the sake of introduction we can assume that $\hat\phi\coloneqq \exp(\hat\psi)$ is a polynomial with the constant term $1$ that does not depend on $\hbar^2$, but the actual assumptions on $\hat\psi$ are much weaker: see Definition \ref{def:spectralcurve} below). To this end, the authors conjectured in~\cite[Conjecture 4.4]{BDKS-fullysimple} the topological recursion of the following two systems of disconnected $n$-differentials:
\begin{align} \label{eq:tr-X}
 \VEV{ \prod_{i=1}^n \Big(\sum_{j=1}^\infty J_j X_i^{j} \Big) e^{\sum_{i=1}^d \frac{t_iJ_i}{i\hbar}}\cD_{\hat\psi}
		e^{\sum_{i=1}^e \frac{s_iJ_{-i}}{i\hbar}}
		\Big)
	}  \prod_{i=1}^n \frac{dX_i}{X_i}
\end{align}
and
\begin{align}  \label{eq:tr-w}
 \VEV{ \prod_{i=1}^n \Big(\sum_{j=1}^\infty J_j w_i^{j} \Big) \cD_{\hat\psi_1}  e^{\sum_{i=1}^d \frac{t_iJ_i}{i\hbar}}\cD_{\hat\psi}
		e^{\sum_{i=1}^e \frac{s_iJ_{-i}}{i\hbar}}
		\Big)
	}  \prod_{i=1}^n \frac{dw_i}{w_i},
\end{align}
with $\hat \psi_1=-\hat \psi$,
which have enumerative meaning of enumeration of particular double weighted Hurwitz numbers and/or constellations on genus $g$ surfaces. Of course, this conjecture extends to the adjoint cases
\begin{align} \label{eq:tr-Y}
\VEV{  e^{\sum_{i=1}^d \frac{t_iJ_i}{i\hbar}}\cD_{\hat\psi}
	e^{\sum_{i=1}^e \frac{s_iJ_{-i}}{i\hbar}}
	\prod_{i=1}^n \Big(\sum_{j=1}^\infty J_{-j} Y_i^{j} \Big)
}  \prod_{i=1}^n \frac{dY_i}{Y_i}
\end{align}
and
\begin{align} \label{eq:tr-v}
\VEV{e^{\sum_{i=1}^d \frac{t_iJ_i}{i\hbar}}\cD_{\hat\psi}
	e^{\sum_{i=1}^e \frac{s_iJ_{-i}}{i\hbar}}
	\cD_{\hat \psi_2}
 \prod_{i=1}^n \Big(\sum_{j=1}^\infty J_{-j} v_i^{j} \Big)
}  \prod_{i=1}^n \frac{dv_i}{v_i},
\end{align}
with $\hat \psi_2=-\hat \psi$ and the same enumerative meaning as above (up to sign and interchange $s\leftrightarrow t$). In the present paper we prove topological recursion for the cases~\eqref{eq:tr-X} and~\eqref{eq:tr-Y}, and in~\cite{ABDKS-FSTR} (using the results of the present paper as a foundation) the topological recursion is proved for the case~\eqref{eq:tr-v} (and, by extension,~\eqref{eq:tr-w} as well).  And it is not even necessary to assume that $\hat \psi_1=-\hat\psi$ and $\hat \psi_2=-\hat \psi$, these results regarding the topological recursion hold when all of these three functions can be mutually different. 

Note that in order to have all four systems of $n$-differentials non-trivial, we must have $t\not=(0,0,0,\dots)$ and $s\not=(0,0,0,\dots)$.

\subsubsection{Symplectic duality} It appears that these four instances of topological recursion (two proved in the present paper, two in~\cite{ABDKS-FSTR}
) mentioned above \eqref{eq:tr-X} -\eqref{eq:tr-v} are related by a chain of the so-called \emph{symplectic dualities} generalizing the usual $x-y$ duality discussed above.

Let us define it. Assume we have two topological recursions, with some given $x$, $y$, and $\omega^{(0)}_1 = y dX/X$, where $X=\exp(x)$ for the first one, and $\tilde x$, $\tilde y$, and $\omega^{(0)}_1 = \tilde y d\tilde X/\tilde X$, where $\tilde X=\exp(\tilde x)$.
We say that these two topological recursions are symplectic dual to each other, if
\begin{itemize}
	\item there exist a function $\Lambda(x)$ such that $y(z)=-\Lambda(z)+F(X(z))$, $\tilde y(z)=\Lambda(z)+\tilde F(\tilde X(z))$ for some $F,\tilde F$;
	\item there exist a function $\phi$ such that $X(z)\tilde X(z) \phi(\Lambda(z)) = 1$.
\end{itemize}
The term ``symplectic duality'' refers here to the fact that on a surface $S$ given in $\mathbb{C}^3$ with the coordinates $X,\tilde X,\Lambda$ by equation $X\tilde X \phi(\Lambda) = 1$, the restrictions of the differentials $(-\Lambda+F(X))dX/X$ and $(\Lambda+\tilde F(\tilde X))d\tilde X/\tilde X$ satisfy
$d(-\Lambda+F(X))dX/X=d(\Lambda+\tilde F(\tilde X))d\tilde X/\tilde X$.

Of course, this definition is ad hoc, as it is dictated by an attempt to summarize the relations between the topological recursions \eqref{eq:tr-X}-\eqref{eq:tr-v}. Let us show, however, that it indeed reduces to the $x-y$ duality in the simplest case. Let $\phi(\Lambda) = \Lambda$. Then $\omega^{(0)}_1$ can be equivalently replaced by $\Lambda(z)X(z) dX(z)^{-1} = \tilde X(z)^{-1} dX(z)^{-1}$ and $\tilde \omega^{(0)}_1$can be equivalently replaced by $-\Lambda(z)\tilde X(z) d\tilde X(z)^{-1} = X(z)^{-1} d \tilde X(z)^{-1}$. Changing $x(z)$ to $X^{-1}(z)$ and $\tilde x(z)$ to $\tilde X^{-1}(z)$, we obtain the classical $x-y$ duality between these two instances of topological recursion.

\subsubsection{Examples of symplectic dualities}
We define a global rational function $\Theta(z)$ on $\C\mathrm{P}^1$ such that $\omega^{(0)}_1$ is given in all these four cases as
\begin{align}
	\text{Case~\eqref{eq:tr-w}: } && {}^w\omega^{(0)}_1=-\Big(\Theta(z) - \sum_{k=1}^d \frac{t_k}{X(z)^k}\Big)\frac{dw^{-1}(z)}{w^{-1}(z)} ;\\ \notag
	\text{Case~\eqref{eq:tr-X}: } && {}^X\omega^{(0)}_1=\Big(\Theta(z) - \sum_{k=1}^d \frac{t_k}{X(z)^k}\Big)\frac{dX(z)}{X(z)} ;\\ \notag
	\text{Case~\eqref{eq:tr-Y}: }&& {}^Y\omega^{(0)}_1=\Big(\Theta(z) - \sum_{k=1}^e \frac{s_k}{Y(z)^k}\Big)\frac{dY(z)}{Y(z)} ;\\ \notag
	\text{Case~\eqref{eq:tr-v}: }&& {}^v\omega^{(0)}_1=-\Big(\Theta(z) - \sum_{k=1}^e \frac{s_k}{Y(z)^k}\Big)\frac{dv^{-1}(z)}{v^{-1}(z)}.
\end{align}
The relations between these functions are as follows:
\begin{align}
\label{eq:formalspectralcurve}
& w^{-1}(z)X(z) \rho_1^{-1}\Big(\Theta(z) - \sum_{k=1}^d \frac{t_k}{X(z)^k}\Big) = 1 ;
	\\ \notag
& X(z)Y(z) \phi\Big(\Theta(z)\Big) = 1;
	\\ \notag
& Y(z) v^{-1}(z) \rho_2^{-1}\Big(\Theta(z) - \sum_{k=1}^e \frac{s_k}{Y(z)^k}\Big) = 1,
\end{align}
where $\rho_1=\exp(\psi_1)$ and $\rho_2=\exp(\psi_2)$.
This gives us a chain of symplectic dualities, though, to have the right sign change in the second symplectic duality, we might redefine~\eqref{eq:tr-Y} and~\eqref{eq:tr-v} multiplying the corresponding $\omega^{(g)}_n$'s by $(-1)^n$, as we did in the example~\eqref{eq:Maps-dual} above.

As the authors remarked in~\cite[Remark 4.5]{BDKS-fullysimple}, even a proof of the topological recursion for the cases~\eqref{eq:tr-X} and~\eqref{eq:tr-Y} is not available in the literature. Only the case~\eqref{eq:tr-X} for $t=(0,0,0,\dots)$ is proved in~\cite{ACEH} for polynomial $\phi$ and in~\cite{bychkov2020topological} for much more general families of choices of $\phi$ and $s$ (and, therefore, the case~\ref{eq:tr-Y} for $s=(0,0,0,\dots)$ for some general families of choices of $\phi$ and $t$).  But the only known case for both $t$ and $s$ being nontrivial is the case of $\phi=1+\theta$ covered by the two-matrix model, see e.g.~\cite{EynardBook,DOPS}.  In this paper we prove the topological recursion for~\eqref{eq:tr-X} and~\eqref{eq:tr-Y} for any choice of parameters $t=(t_1,\dots,t_d,0,0,\dots)$ and $s=(s_1,\dots,s_e,0,0,\dots)$ and for a large family of possible choices for $\phi$.

As for the case~\eqref{eq:tr-w}, the topological recursion was known only for $\rho_1=1+\theta$ and $s=(0,1,0,0,0,\dots)$ due to~\cite{BCGF21,BDKS-fullysimple} (and, therefore, for the symmetric case~\eqref{eq:tr-v} for $\rho_2=1+\theta$ and $t=(0,1,0,0,0,\dots)$), but, based on the results of the present paper, in~\cite{ABDKS-FSTR} it has been proved for any choice of parameters $t=(t_1,\dots,t_d,0,0,\dots)$ and $s=(s_1,\dots,s_e,0,0,\dots)$ and for a large family of possible choices for $\phi$ and $\rho_1$ (and the respective statement is implied for the symmetric case~\eqref{eq:tr-v}).




\subsection{Acknowledgments} S.~S. was supported by the Netherlands Organization for Scientific Research. Research of B.B. was supported in part by the ISF grant 876/20. While working on this project B.B. benefited from support of Technion University (Spring 2022). M.K. and B.B. worked on Sections 3, 4 as part of the project supported by the International Laboratory of Cluster Geometry NRU HSE, RF Government grant, ag. № 075-15-2021-608 dated 08.06.2021. P.D.-B. worked on Section 5 under support by the Russian Science Foundation (Grant No. 20-71-10073).

We thank A.~Hock for useful discussion.

When this paper was almost ready for submission, we got a message from S.~Charbonnier that he, in collaboration with V.~Bonzom, G.~Chapuy, and E.~Garcia-Failde, obtained a big part of our results by entirely different methods. Their paper is available at~\cite{BCCG}.



\section{Recursion formula for \texorpdfstring{$(m,n)$}{(m,n)}-point functions: the case of formal series}\label{sec:formalprel}

\subsection{\texorpdfstring{$(m,n)$}{(m,n)}-point functions}\label{sec:mnpoint}
Let $t=(t_1,t_2,\dots)$, $s=(s_1,s_2,\dots)$ be two infinite sets of formal parameters and let $\hat\phi(\theta):=e^{\hat\psi(\theta)}$, where $\hat\psi(\theta)$ is an arbitrary formal power series in~$\theta$ and also in~$\hbar^2$ such that $\hat\psi(0)|_{\hbar=0}=0$; also let $\psi(\theta)=\hat{\psi(\theta)}|_{\hbar=0}$ and $\phi(\theta)=\hat{\phi(\theta)}|_{\hbar=0}=e^{\psi(\theta)}$.

The hypergeometric-type KP tau function $Z$ and the corresponding potential $F=\log Z$ are defined \cite{Kharchev,OrlovScherbin} by the following explicit expansion in the Schur functions
\begin{equation}
Z(t,s)=e^{F(t,s)}=\sum_\lambda\left(\prod_{(i,j)\in\lambda}\hat\phi(\hbar(j-i))\right)
s_\lambda(t_1/\hbar,t_2/\hbar,\dots)s_\lambda(s_1/\hbar,s_2/\hbar,\dots).
\end{equation}

The potential admits automatically the genus expansion $F=\sum_{g=0}^\infty \hbar^{2g-2} F_g$ (that is, $\hbar$ enters $F$ with even exponents greater or equal to $-2$) and the \emph{correlator functions} $H_{m,n}^{(g)}$ are defined by
\begin{multline}
H_{m,n}^{(g)}(X_1,\dots,X_m,Y_{m+1},\dots,Y_{m+n})=\\
\sum_{k_1,\dots,k_{m+n}=1}^\infty\frac{\partial^{m+n}F_g}
{\partial t_{k_1}\dots\partial t_{k_m}\partial s_{k_{m+1}}\dots\partial s_{k_{m+n}}}
\prod_{i=1}^m X_i^{k_i}\prod_{i=m+1}^n Y_i^{k_i}.
\end{multline}
More explicitly, if we treat $F_g$ as the generating function for its Taylor coefficients, the `correlators',
\begin{equation}
F_g=\sum_{m,n=1}^\infty\frac1{m!n!}\sum_{k_1,\dots,k_{m+n}=1}^\infty f_{g,(k_1,\dots,k_m),(k_{m+1},\dots,k_{m+n})}t_{k_1}\dots t_{k_m}s_{k_{m+1}}\dots s_{k_{m+n}},
\end{equation}
then the functions $H_{m,n}^{(g)}$ pack the same correlators in a different way,
\begin{multline}
H_{m,n}^{(g)}=
\sum_{m',n'=1}^\infty\frac1{m'!n'!}
\sum_{k_1,\dots,k_{m+n},l_1,\dots,l_{m'+n'}=1}^\infty 
\\
f_{g,(k_1,\dots,k_m,l_1,\dots,l_{m'}),(k_{m+1},\dots,k_{m+n},l_{m'+1},\dots,l_{m'+n'})}
\prod_{i=1}^m X_i^{k_i}
\prod_{j=m+1}^{m+n} Y_{j}^{k_j}
\prod_{i=1}^{m'} t_{l_{i}}
\prod_{j=m'+1}^{m'+n'} s_{l_{j}}.
\end{multline}

In a more conceptual way, the tau function can be represented as the vacuum expectation value
\begin{equation}
Z(t,s)=e^{F(t,s)}= \langle 0 |\;e^{\sum_{k=1}^\infty \frac{t_k J_k}{k\hbar}}\;\cD_{\hat\psi}\;e^{\sum_{k=1}^\infty \frac{s_k J_{-k}}{k\hbar}}\;|0\rangle
\end{equation}
where $J_k$ and $\cD_{\hat\psi}$ are certain operators acting on the Fock space, see above Section \ref{Sec:vaccum}; for the complete introduction to the topic see \cite{bychkov2021explicit}. Then $H^{(g)}_{m,n}$ could be represented equivalently as the corresponding \emph{connected} vacuum expectation value:
\begin{multline}\label{eq:hgmn}
\sum_{g\ge0}\hbar^{2g-2+m+n}H_{m,n}^{(g)}(X_1,\dots,X_m,Y_{m+1},\dots,Y_{m+n})=\\
\langle 0 |\;
\left(\prod_{i=1}^m \sum_{k=1}^\infty \tfrac{X_i^{k}}{k}J_k\right)
e^{\sum_{k=1}^\infty 
\frac{t_k J_{k}}{k\hbar}}\;\cD_{\hat\psi}\;e^{\sum_{k=1}^\infty 
\frac{s_k J_{-k}}{k\hbar}}
\left(\prod_{i=m+1}^{m+n} \sum_{k=1}^\infty \tfrac{Y_i^{k}}{k}J_{-k}\right)
\;|0\rangle^{\circ}.
\end{multline}
See also \eqref{eq:Hfirstmn}. The last treatment can be especially useful if we consider certain specializations of~$t$ and~$s$ variables, since for the specializations the partial derivatives are not available. Some relations on $H^{(g)}_{m,n}$ look nicer if they are reformulated in terms of $(m,n)$-point differentials defined by
\begin{equation}
	\omega^{(g)}_{m,n}=d_1\dots d_{m+n}H_{m,n}^{(g)}
\end{equation}
These forms can be regarded as yet another convenient way to pack correlators into generating series.
The functions $H_{m,n}^{(g)}$ and the forms $\omega_{m,n}^{(g)}$ are the main objects of our research. These are formal power series in the variables $X_1,\dots,X_m,Y_{m+1},\dots,Y_{m+n}$, while the quantities $t_1,t_2,\dots$, $s_1,s_2,\dots$, as well as the coefficients of the series $\hat\psi$ are regarded as parameters of the problem.

\subsection{Formulation of the recursion (preliminary version)}

Theorem below provides a formula relating the correlator functions $H^{(g)}_{m+1,n}$ and  $H^{(g)}_{m,n+1}$ modulo similar functions with smaller values of indices $(g,n,m)$ with respect to certain natural ordering. We will see that this formula can actually be used to compute both $H^{(g)}_{m+1,n}$ and $H^{(g)}_{m,n+1}$. The two functions have a common set of variables $(X_1,\dots,X_m,Y_{m+2},\dots,Y_{m+n+1})$ that we will denote by $(X_M,Y_N)$, where $M=\{1,\dots,m\}$, $N=\{m+2,\dots,m+n+1\}$. Besides, there is one extra variable $X_{m+1}$ for the function $H^{(g)}_{m+1,n}$ and the variable $Y_{m+1}$ for $H^{(g)}_{m,n+1}$, respectively. In order to simplify manipulations with these distinguished variables we set
\begin{equation}
X=X_{m+1},\quad Y=Y_{m+1}.
\end{equation}
To state the formula we introduce some new functions. Set
\begin{equation}\label{eq:S}
H_{m,n}=\sum_{g=0}^\infty\hbar^{2g}H^{(g)}_{m,n},\qquad
\cS(u)=\frac{e^{u/2}-e^{-u/2}}{u},
\end{equation}
\begin{multline}\label{eq:cT}
\cT^Y_{m,n+1}(u;X_{M};Y;Y_{N}):=
\sum_{k=1}^\infty\frac{\hbar^{2(k-1)}u^k}{k!}\left(\prod_{i=1}^k
\big\lfloor_{Y_{\bar i}=Y}\cS(u\hbar Y_{\bar i}\partial_{Y_{\bar i}})\right)\\
\Bigg(\Big(\prod_{i=1}^kY_{\bar i}\partial_{Y_{\bar i}}\Big)
H_{m,k+n}(X_{M};Y_{\{\bar1,\dots,\bar k\}},Y_{N})\\
+\delta_{(m,n,k),(0,0,1)}\sum_{k=1}^\infty\frac{s_k}{Y^k_{\bar 1}}+
\delta_{(m,n,k),(0,1,1)}\frac{Y_{2}}{Y_{\bar1}-Y_2}\Biggr),
\end{multline}
\begin{equation}\label{eq:cW}
\cW^Y_{m,n+1}(u;X_{M};Y;Y_{N}):=\\
\frac{1}{u\cS(u\,\hbar)}
\sum_{k=0}^\infty\frac{1}{k!}
\sum_{\substack{M=\sqcup_{i=1}^kI_i\\N=\sqcup_{i=1}^kJ_i}}
\prod_{i=1}^k\cT^Y_{u;|I_i|,|J_i|+1}(X_{I_i};Y;Y_{J_i}).
\end{equation}
The functions $\cT^Y_{m,n+1}$, $\cW^Y_{m,n+1}$ depend on an additional variable $u$ as a formal power series. Remark that $\cW^Y_{m,n+1}$ contains nonnegative powers of~$u$ only with an exception for $m=n=0$: the series $\cW^Y_{0,1}$ starts with $\tfrac1u$.

Next, for $p\in\Z$ we set
\begin{equation}
\hat\phi_p(\theta)=\begin{cases}
\prod_{i=-\frac{p-1}{2}}^{\frac{p+1}{2}}
\hat\phi(\theta+i\,\hbar),&p>0,\\
1,&p=0,\\
1/\hat\phi_{-p}(\theta),&p<0.
\end{cases}
\end{equation}

\begin{theorem}\label{th:formalH}
	The following relation holds for all $(g,m,n)$ with $2g-2+m+n>0$
	\begin{equation}\label{eq:mainrelserHX}
		X\partial_X H^{(g)}_{m+1,n}(X_M;X;Y_N)=[\hbar^{(2g)}]\sum_{p=-\infty}^\infty X^p\sum_{r=0}^\infty  \partial_\theta^r\hat\phi_p(\theta)\bigm|_{\theta=0}\;[Y^{-p}u^r]\;\cW^Y_{m,n+1}(u;X_M;Y;Y_N).
	\end{equation}

By the symmetry between $X$ and $Y$ variables, we also have
	\begin{equation}\label{eq:mainrelserHY}
		Y\partial_Y H^{(g)}_{m,n+1}(X_M;Y;Y_N)=[\hbar^{(2g)}]\sum_{p=-\infty}^\infty Y^p\sum_{r=0}^\infty  \partial_\theta^r\hat\phi_p(\theta)\bigm|_{\theta=0}\;[X^{-p}u^r]\;\cW^X_{m+1,n}(u;X_M;X;Y_N).
	\end{equation}
where $\cW^X,U^Y$, and also $\cT^X$ are defined in a similar way, with the exchange of the role of $X$ and $Y$ variables and also $t$ and $s$ parameters.
\end{theorem}

We will need also the following extension of the last theorem. Set
\begin{equation}
\tilde\phi_p(\theta,w)=\hat\phi_p(\theta)e^{w\,\theta}.
\end{equation}

\begin{theorem}\label{th:formalW}
	The following relation holds for all $(m,n)$ with no exceptions
	\begin{equation}\label{eq:mainrelserWX}
		\cW^X_{m+1,n}(w;X_M;X;Y_N)=\sum_{p=-\infty}^\infty X^p\sum_{r=0}^\infty  \partial_\theta^r\tilde\phi_p(\theta,w)\bigm|_{\theta=0}\;[Y^{-p}u^r]\;\cW^Y_{m,n+1}(u;X_M;Y;Y_N).
	\end{equation}
By the symmetry between $X$ and $Y$ variables, we also have
	\begin{equation}\label{eq:mainrelserWY}
		\cW^Y_{m,n+1}(u;X_M;Y;Y_N)=\sum_{p=-\infty}^\infty Y^p\sum_{r=0}^\infty  \partial_\theta^r\tilde\phi_p(\theta,u)\bigm|_{\theta=0}\;[X^{-p}w^r]\;\cW^X_{m+1,n}(w;X_M;X;Y_N).
	\end{equation}
\end{theorem}

By definition of the function~$\cW^X_{m+1,n}(w;X_M;Y;Y_N)$, its coefficient of $w^0$ is
\begin{equation}
[w^0]\cW^X_{m+1,n}=\sum_{g=0}^\infty\hbar^{2g}X\partial_X H^{(g)}_{m+1,n}+\delta_{(m,n),(0,0)}\sum_{k=1}^\infty\frac{t_k}{X^k}+
\delta_{(m,n),(1,0)}\frac{X_{1}}{X-X_1}.
\end{equation}
This shows that Theorem~\ref{th:formalH} is a specialization of Theorem~\ref{th:formalW}. It also explains why relations of Theorem~\ref{th:formalH} fail in the unstable cases and how they should be corrected. The proof of Theorems~\ref{th:formalH} and~\ref{th:formalW} are given below, in Sect.~\ref{sec:Fock}. But now we provide certain comments on the structure of the terms entering their relations.

\subsection{The structure of the formulas \texorpdfstring{\eqref{eq:mainrelserHX}--\eqref{eq:mainrelserHY}}{()-()}, \texorpdfstring{\eqref{eq:mainrelserWX}--\eqref{eq:mainrelserWY}}{()-()}, graph presentation, and genus count}
First of all, let us explain the meaning of the singular correction in~\eqref{eq:cT}. The first correction term means that the summand $Y_{\bar1}\partial_{Y_{\bar1}}H_{0,1}^{(0)}$ should be replaced by
\begin{equation}
Y_{\bar1}\partial_{Y_{\bar1}}H_{0,1}^{(0)}(Y_{\bar1})+\sum_{k=1}^\infty\frac{s_k}{Y_{\bar1}^k}
\end{equation}
which involves negative powers of $Y_{\bar1}$.
The second correction term means that the summand $Y_{\bar1}\partial_{Y_{\bar1}}H_{0,2}^{(0)}(Y_{\bar1},Y_2)$ should be replaced by
\begin{equation}
Y_{\bar1}\partial_{Y_{\bar1}}H_{0,2}^{(0)}(Y_{\bar1},Y_2)+\frac{Y_{2}}{Y_{\bar1}-Y_2}=
Y_{\bar1}\partial_{Y_{\bar1}}H_{0,2}^{(0)}(Y_{\bar1},Y_2)+\sum_{k=1}^\infty\frac{Y_2^k}{Y_{\bar1}^k}
\end{equation}
which also involves negative powers of $Y_{\bar1}$. Remark, however, that a similar term
\begin{equation}
	Y_{\bar1}Y_{\bar2}\partial_{Y_{\bar1}}\partial_{Y_{\bar2}}H_{0,2}^{(0)}(Y_{\bar1},Y_{\bar2})
\end{equation}	
	 appearing for $(m,n,k)=(0,0,2)$ and intended to be restricted to the diagonal $Y_{\bar1}=Y_{\bar2}$ is left in its original regular form. As a result, the series~$\cW^Y_{m,n+1}$ contains arbitrary large positive or negative exponents of the variable~$Y$. This series is well defined since the coefficient of any particular monomial in the remaining variables ($t_i,s_i,X_i,Y_i,u,\hbar$) is a finite Laurent polynomial in $Y$. Remark that the left hand side in~\eqref{eq:mainrelserHX} is regular in~$X$. It follows that the most essential terms of~$\cW^Y_{m,n+1}$ that contribute to the left hand side of~\eqref{eq:mainrelserHX} are exactly those with negative exponents of $Y$ while the contribution of terms with nonnegative exponents of~$Y$ cancel out.

Next, let us make a comment on the definition of~$\cW^Y_{m,n+1}$. The subsets $I_i,J_i$ participating in~\eqref{eq:cW} are allowed both being empty. By that reason the sum in~\eqref{eq:cW} is infinite. It is useful to treat the factors with $I_i=J_i=\emptyset$ separately and to represent the expression for $\cW^Y_{m,n+1}$ in the following form
\begin{equation}\label{eq:cWexp}
\cW^Y_{m,n+1}(X_M;Y;Y_N)=\frac{e^{\cT^Y_{0,1}(u;Y)}}{u\,\cS(u\hbar)}
\sum_{\substack{M\cup N=\sqcup_\alpha K_\alpha~
K_\alpha\ne\varnothing\\~I_\alpha=K_\alpha\cap M,~J_\alpha=K_\alpha\cap N}}
\prod_{\alpha}\cT^Y_{|I_\alpha|,|J_\alpha|+1}(X_{I_\alpha};Y;X_{J_\alpha})
\end{equation}
where the summation carries over the set of all partitions of the set $M\cup N$ into \emph{unordered} collection of disjoint \emph{nonempty} subsets $K_\alpha$, and where we denote $I_\alpha=K_\alpha\cap M$, $J_\alpha=K_\alpha\cap N$. This sum has finitely many terms.

In order to understand better the structure of the function $\cW^Y_{m,n+1}$, it is convenient to represent all the terms entering to its definition by graphs of special kind. These are bipartite graphs with $m+n+1$ white vertices labeled by indices $1,2,\dots,m+n+1$, and some number of black unmarked vertices. The white vertex labeled by $m+1$ is distinguished. It is connected by edges with all black vertices, moreover, multiple edges are allowed. Besides, every remaining white vertex is connected with exactly one black vertex by a single dashed edge as in the picture below.
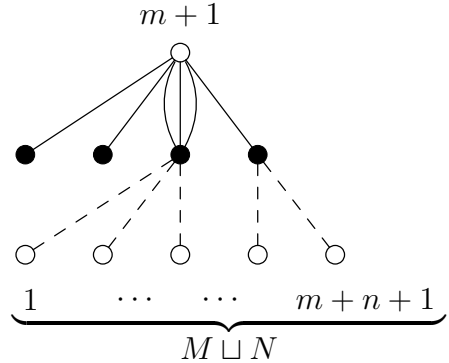
\[
\underbrace{	\vcenter{
	\xymatrix@!C=5pt@R=30pt{
		& & *+[o][F-]{{}} \save[]+<0cm,5mm>*\txt<3pc>{$m+1$} \restore  & &  \\
		*+[o][F**]{}\ar@{-}[rru] & *+[o][F**]{}\ar@{-}[ru] &  *+[o][F**]{}\ar@{-}@/^/[u]\ar@{-}[u]\ar@{-}@/_/[u] & *+[o][F**]{}\ar@{-}[lu] &\\
	*+[o][F-]{{}} \save[]+<0cm,-6mm>*\txt{\ $1$} \restore \ar@{--}[rru] &
	*+[o][F-]{{}} \ar@{--}[ru]&
	*+[o][F-]{{}} \save[]+<0cm,-6mm>*\txt{$\cdots$\phantom{$\cdots$}$\cdots$} \restore \ar@{--}[u]&
	*+[o][F-]{{}} \ar@{--}[u]&
	*+[o][F-]{{}} \save[]+<4mm,-6mm>*\txt<5pc>{$m+n+1$} \restore \ar@{--}[lu] \\
	}
}
}_{\textstyle M\sqcup N}
\]

The contribution of the graph is computed as the product of the weights of its black vertices defined in the following way. Consider one of the black vertices. Consider all white vertices connected to it by dashed edges and denote by $I$ and $J$ the sets of labels of these vertices from the ranges $M=\{1\dots m\}$ and $N=\{m+2\dots m+n+1\}$, respectively. Let $k$ be the number of edges connecting this black vertex to the white vertex labeled by $m+1$. Then the weight of this black vertex is equal to
\begin{equation}
\left(\prod_{i=1}^k\lfloor_{Y_{\bar i}=Y}u\,\cS(u\hbar Y_{\bar i}\partial_{Y_{\bar i}})Y_{\bar i}\partial_{Y_{\bar i}}\right)
H_{|I|,k+|J|}(X_{I};Y_{\{\bar1,\dots,\bar k\}},Y_{J})
\end{equation}
with an appropriate singular correction for $k=1$, $I=\varnothing$, and $|J|=0$~or~$1$.

Let $\gamma$ be a graph. Its genus $g(\gamma)$ is defined as the number of its edges (both filled and dashed) minus the number of vertices (both black and white) plus one. In other words, it is the minimal number of edges that can be removed such that the graph remains connected.

Then $\cW^Y_{m,n+1}$ is computed as the sum over all isomorphism classes of graphs of the described type; the summand corresponding to a given graph $\gamma$ is equal to the product of weights of all its black vertices multiplied by $\hbar^{2g(\gamma)}$ and divided by the order of automorphism group of the graph. Finally, the obtained sum is multiplied by an overall factor $\frac{1}{u\,\cS(u\,\hbar)}$.

One of the sources of the automorphisms of the graph are permutations of its multiple edges. These automorphisms are accounted in the factor $\frac1{k!}$ in~\eqref{eq:cT}. Another source of automorphisms are possible black leaves that is black vertices connected with the distinguished white vertex $m+1$ only. These automorphisms are accounted in the factorial factors appearing in the expansion of the exponent in~\eqref{eq:cWexp}.

\medskip
Thus, the total genus (the exponent of~$\hbar^2$) corresponding to the contribution of a particular graph to the right hand side of~\eqref{eq:mainrelserHX} or~\eqref{eq:mainrelserWX} is computed as the sum of the following nonnegative integers:
\begin{itemize}
\item the indices $g'$ corresponding to the functions $H^{(g')}_{m',k+n'}$ assigned to the black vertices;
\item the exponent of $\hbar^2$ in the differential operators $\cS(u\hbar Y_{\bar i}\partial_{Y_{\bar i}})$ applied to these functions;
\item the genus of the graph itself;
\item the exponent of $\hbar^2$ in the overall factor $\frac{1}{u\cS(u\,\hbar)}$;
\item the exponent of $\hbar^2$ in the expansion of the functions $\hat\phi_p$ and $\tilde\phi_p$ entering the right hand side of the formulae.
\end{itemize}


For example, in the case $g=0$ all the graphs that contribute to the computation of $H^{(0)}_{m+1,n}$ are trees.

\subsection{Operator formalism and the computation of \texorpdfstring{$\cW^X_{m+1,n}$}{WXm+1,n} as a series}\label{sec:Fock} The proof of Theorems~\ref{th:formalH} and~\ref{th:formalW} is based on the formalism of vertex operators and VEVs. Recall that the bosonic Fock space of charge~$0$ introduced in Sect~\ref{Sec:vaccum} can be identified with the space of infinite power series in countably many variables, $\mathcal{V}_0=\C[[q_1,q_2,\dots]]$. There is a remarkable action on this space of the Lie algebra $\widehat{\mathfrak{gl}}(\infty)$, a one-dimensional central extension of the Lie algebra of infinite matrices whose rows and columns are labeled by half-integers. Let $E_{i,j}$, $i,j\in\Z+\tfrac12$ be the matrix unit. Then, for example, the shift operator $J_k=\sum_{i\in\Z+\frac12}E_{i-k,i}$ corresponding to a matrix with one nonzero diagonal filled by units and situated on the distance $k$ from the principal one acts as $J_k=k\,\partial_{q_k}$, $J_{-k}=q_k$ (operator of multiplication by $q_k$) if $k>0$, and $J_0=0$. Another example is the operator~$\cD_{\hat\psi}$ entering the definition~\eqref{eq:Hfirstmn} of correlator functions. It corresponds to the diagonal matrices of the form
\begin{equation}
\cD_{\hat\psi}=\exp\Bigl(\sum_{k\in\Z+\tfrac12}d_k E_{k,k}\Bigr),
\end{equation}
where the entries $d_k$ are defined by
\begin{equation}
d_{k+\tfrac12}-d_{k-\tfrac12}=\hat\psi(k\,\hbar),
\end{equation}
see~\cite{bychkov2021explicit}. These equalities determine the components $d_k$ up to a common additive constant which is not important since the constant matrices from $\widehat{\mathfrak{gl}}(\infty)$ act trivially. In general, the elements of $\widehat{\mathfrak{gl}}(\infty)$ act on~$\mathcal V_0$ as differential operators in $q$-variables. This action can be described as follows. Consider the operators $\sum_{p\in\Z+\frac12}(p-\frac{k}{2})^rE_{p-k,p}$ and collect them to the following generating series
\begin{equation}
	\cE(u,z)=\sum_{p\in\Z}z^p \sum_{k\in\Z+\frac12}
	e^{u(k-\frac{p}{2})} E_{k-p,k},
\end{equation}
Then we have, explicitly (see for example \cite[Section 2]{bychkov2021explicit} for the details):
\begin{equation}
	\cE(u,z)=\frac{
		e^{\sum_{i=1}^\infty u\,\cS(u\,i) J_{-i}z^{-i}}
		e^{\sum_{i=1}^\infty u\,\cS(u\,i) J_{i}z^{i}} -1 }
	{u\,\cS(u)}.
\end{equation}

We \emph{define} now the functions $\cW^X_{m+1,n}(w;X_M;X;Y_N)$ and $\cW^Y_{m,n+1}(u;X_M;Y;Y_N)$ as the following connected vacuum expectation values
\def\bJ{{\mathbb J}}
\def\bX{{\mathbb X}}
\def\bY{{\mathbb Y}}
\begin{align}
	\bX_i&=\sum_{k=1}^\infty \tfrac{X_i^{k}}{k}J_k,
	\qquad
	\bY_j=\sum_{k=1}^\infty \tfrac{Y_i^{k}}{k}J_{-k},\\
	\label{eq:WXVEV}
	\cW^X_{m+1,n}&=\langle 0 |\;\bX_1\dots \bX_m
	e^{\sum\limits_{k=1}^\infty \frac{t_k J_{k}}{k\hbar}}\;\cE(\hbar w,X)\;\cD_{\hat\psi}\;e^{\sum\limits_{k=1}^\infty \frac{s_k J_{-k}}{k\hbar}}
	\bY_{m+2}\dots \bY_{m+n+2}
	\;|0\rangle^\circ,\\
	\label{eq:WYVEV}
	\cW^Y_{m,n+1}&=\langle 0 |\;\bX_1\dots \bX_m
	e^{\sum\limits_{k=1}^\infty \frac{t_k J_{k}}{k\hbar}}\;\cD_{\hat\psi}\;\cE(\hbar u,Y^{-1})\;e^{\sum\limits_{k=1}^\infty \frac{s_k J_{-k}}{k\hbar}}
	\bY_{m+2}\dots \bY_{m+n+2}
	\;|0\rangle^\circ.
\end{align}

Then we observe that $\cW^Y_{m,n+1}$ is given explicitly by~\eqref{eq:cT}--\eqref{eq:cW} and $\cW^X_{m+1,n}$ is given by similar expression with the exchange of the role of $X$ and $Y$ variables. Namely, the positive $J$-operators entering $\cE(\hbar u,Y^{-1})$ while commuting with $e^{\sum\limits_{k=1}^\infty \frac{s_k J_{-k}}{k\hbar}}$ and $\bY_j$ produce the singular terms in~\eqref{eq:cT}, and the negative $J$-operators entering $\cE(\hbar u,Y^{-1})$ produce the regular summands in~\eqref{eq:cT}. In fact, exactly this computation serves as the motivation for introducing singular terms in~\eqref{eq:cT}. The combinatorics behind this computation and also behind the inclusion/exclusion principle relating connected and disconnected correlators is the same as in~\cite{bychkov2021explicit,BDKS-fullysimple} and we do not reproduce the details here.

Next, we conjugate $\cE(\hbar w,X)$ in~\eqref{eq:WXVEV} by the operator $\cD_{\hat\psi}$ and obtain
\begin{equation}
	\cE(\hbar w,X)\;\cD_{\hat\psi}=\cD_{\hat\psi}\;{\mathbb E}(\hbar w,X)
\end{equation}
where
\begin{equation}
	\begin{aligned}
		{\mathbb E}(\hbar w,X)&=\cD_{\hat\psi}^{-1} \cE(\hbar w,X)\cD_{\hat\psi}=\sum_{p\in\Z}z^p \sum_{k\in\Z+\frac12}
		\frac{e^{d_k}}{e^{d_{k-p}}}e^{w\,\hbar\,(k-\frac{p}{2})} E_{k-p,k}\\
		&=\sum_{p\in\Z}z^p \sum_{k\in\Z+\frac12}
		\hat\phi_p(\hbar(k-\tfrac{p}{2}))e^{w\,\hbar\,(k-\frac{p}{2})} E_{k-p,k}\\
		&=\sum_{p\in\Z}z^p \sum_{k\in\Z+\frac12}
		\tilde\phi_p(\hbar(k-\tfrac{p}{2}),w) E_{k-p,k}\\
		&=\sum_{p=-\infty}^\infty X^p\sum_{r=0}^\infty  \partial_\theta^r\tilde\phi_p(\theta,w)\bigm|_{\theta=0}\;[Y^{-p}u^r]\;\cE(\hbar u,Y^{-1})
	\end{aligned}
\end{equation}
Comparing with~\eqref{eq:WYVEV} we obtain the desired relation of Theorem~\ref{th:formalW}. Picking the coefficient of $w^0$ on both sides of the obtained relation we get that of Theorem~\ref{th:formalH}.

\section{The formal spectral curve and Lagrange Inversion}\label{sec:formalfinal}

\subsection{Formal spectral curve}\label{sec:spectralcurveformal}

The variables $X$ and $Y$ entering relations of Theorems~\ref{th:formalH} and~\ref{th:formalW} are considered in these theorems as two independent variables having no relationship to one another. Equations~\eqref{eq:mainrelserHX} and~\eqref{eq:mainrelserWX} express the coefficients of the Laurent expansion of the left hand sides in~$X$ in terms of the coefficients of the Laurent expansion of $\cW^Y_{m,n+1}$ in~$Y$. Our next step is to relate $X$ and $Y$ by a change of variables and to interpret~\eqref{eq:mainrelserHX} and~\eqref{eq:mainrelserWX} in terms of this change. The change will involve both positive and negative powers of the variables, and in order to assign a meaning to such a change we introduce the following definition.

\begin{definition}
	We denote by $R$ the ring of regular power series in the `basic' variables $t_k$, $s_k$,
	whose coefficients are Laurent polynomials in one additional variable denoted by $z$ (or $X$ or $Y$).
\end{definition}

The whole series lying in $R$ may contain arbitrary large positive or negative powers of~$z$ but the exponents of~$z$ entering a particular monomial in $(t,s)$-variables are bounded. Note that any series in~$R$ of the form
\begin{equation}
	X(z)=z^{\pm1}+o(1),
\end{equation}
where the summand $o(1)$ belongs to                 the ideal generated by $(t,s)$-variables provides an invertible change in the ring $R$ so that any series in~$R$ rewritten in terms of the new variable $X$ also belongs to $R$.

Let $t=(t_1,t_2,\dots)$, $s=(s_1,s_2,\dots)$ be as before, and $\phi(\theta)=e^{\psi(\theta)}$ be an arbitrary power series with the constant term~$1$.

\begin{proposition}\label{prop:formalsc}
There exist series $X(z)$, $Y(z)$, $\Theta(z)$ in $R$ possessing the following properties.
\begin{itemize}
\item The three series satisfy
\begin{equation}
X(z)\;Y(z)\;\phi(\Theta(z))=1.
\end{equation}
\item The series $X(z)$ contains positive powers of~$z$ only, its coefficient of $z$ is invertible as a series in $(t,s)$, and we have
\begin{equation}\label{eq:ThetaXz}
\Theta(z)=\sum_{k=1}^\infty\frac{t_k}{X(z)^k}+O(z)
\end{equation}
where the term $O(z)$ contains positive powers of $z$ only.
\item The series $Y(z)$ contains negative powers of~$z$ only, its coefficient of $z^{-1}$ is invertible as a series in $(t,s)$, and we have
\begin{equation}\label{eq:ThetaYz}
\Theta(z)=\sum_{k=1}^\infty\frac{s_k}{Y(z)^k}+O(z^{-1})
\end{equation}
where the term $O(z^{-1})$ contains negative powers of $z$ only.
\end{itemize}
The series $X,Y,\Theta$ are determined by these requirements uniquely up to a multiplication of $z$ by a constant (an invertible series in $(t,s)$ variables).
\end{proposition}

Remark that the functions $X(z)$ and $Y(z)$ provide invertible changes in $R$, and the dependence between $X$ and $Y$ implied by these changes is independent of the freedom in a choice of the coordinate~$z$. The actual rescaling of $z$ is not so important, but it can be fixed, if needed, for example, by an additional requirement
\begin{equation}
X(z)=z+O(z^2).
\end{equation}

\begin{definition} By a \emph{formal spectral curve} associated with the data $(t,s,\phi)$ we mean the projective line $\Sigma=\C P^1$ with affine coordinate~$z$ and a triple of functions $X,Y,\Theta$ in~$R$ satisfying requirements of Proposition~\ref{prop:formalsc}.
\end{definition}

\begin{proof}[Proof of Proposition~\ref{prop:formalsc}]
Let us write
\begin{equation}
\Theta(z)=\sum_{k=-\infty}^\infty \alpha_k z^k.
\end{equation}
We will show that the coefficients of $\Theta(z)$ obey certain equation allowing one to express~$\alpha_0$ as a function in the remaining $\alpha$-parameters. So we can take the coefficients $\alpha_k$, $k\ne0$, as a new independent set of parameters of the problem (instead of $t$ and $s$). We express $t_i$, $s_j$ as functions in these parameters. Then, we apply the inverse function theorem to express $\alpha$-coordinates as functions in $t$ and $s$ parameters.

Taking the formal logarithm of $\phi$ we write
\begin{equation}
\log\phi(\Theta(z))=A(z)+B(z)
\end{equation}
where $A$ and $B$ involve nonnegative and nonpositive exponents of $z$, respectively. There is an ambiguity in a choice of the constant terms in $A$ and $B$. This ambiguity exactly corresponds to an ambiguity in a possible rescaling of $z$ coordinate in the proposition. We just require that the constant terms in $A$ and $B$ are certain series in $\alpha$-parameters such that $A|_{\alpha=0}=B|_{\alpha=0}=0$. Set
\begin{equation}
X=z e^{-A(z)},\qquad Y=z^{-1}e^{-B(z)}.
\end{equation}
Then the equation $X\,Y\,\phi(\Theta)=1$ and the requirement that $X$ and $Y$ are regular changes at $z=0$ and $z=\infty$, respectively, are satisfied. Applying these changes in~$R$ we can represent~$\Theta$ as an infinite Laurent series in the corresponding coordinate:
\begin{equation}\label{eq:ThetatXY0}
\Theta|_{z=z(X)}=\sum_{k=1}^\infty\frac{t_k}{X^k}+t_0+O(X),\qquad \Theta|_{z=z(Y)}=\sum_{k=1}^\infty\frac{s_k}{Y^k}+s_0+O(Y).
\end{equation}
The coefficients $t_k$, $s_k$ of these expansions are expressed as functions (formal power series) in $\alpha$-parameters. Expansions~\eqref{eq:ThetatXY0} differ from~\eqref{eq:ThetaXz},~\eqref{eq:ThetaYz} by the presence of the constant terms $t_0$, $s_0$. Vanishing of $t_0$ and $s_0$ provides functional relations between parameters $\alpha_k$. In fact, these two equations are equivalent to one another due to the following identity:
\begin{equation}\label{eq:t0s0}
t_0=s_0.
\end{equation}
This identity is proved below. The equation $t_0=0$ (or an equivalent one $s_0=0$) allows one to express $\alpha_0$ as a function in the remaining $\alpha$-parameters. Inverting the obtained dependence of $(t_,s)$ variables in $\alpha$ variables we resolve the equations of spectral curve.

In order to justify application of implicit function theorem one should add the computation of these equations in the liner approximation. Up to the terms of order greater than~$1$ in $\alpha$-coordinates these equations read
\begin{equation}
\alpha_k=t_{-k}\quad(k<0),\qquad \alpha_k=s_{k}\quad(k>0),\qquad \alpha_{0}=0.
\end{equation}
and so they are obviously solved in $(t,s)$-variables.

Let us prove finally the identity~\eqref{eq:t0s0}. For that we involve into consideration formal meromorphic differentials of the form $f(z)\,dz$, $f\in R$.

\begin{definition}
The residue of a formal meromorphic differential  $f(z)\,dz$, $f\in R$, is defined as its coefficient of $\tfrac{dz}{z}$ and denoted by $\Res f(z)\,dz$.
\end{definition}

The invariance of residues implies, in particular,
\begin{equation}
\Res\frac{dX(z)}{X(z)^k}=\begin{cases}1,&k=1\\0,&k\ne1,\end{cases}\qquad
\Res\frac{dY(z)}{Y(z)^k}=\begin{cases}-1,&k=1\\0,&k\ne1,\end{cases}
\end{equation}
where we denote $dX(z)=X'(z)\,dz$ and $dY(z)=Y'(z)\,dz$.

In these terms, the equations relating $\alpha$ and $(t,s)$ parameters can be written as
\begin{equation}
t_k=\Res\Theta(z) X(z)^{k-1}\,dX(z),
\quad
s_k=-\Res\Theta(z) Y(z)^{k-1}\,dY(z),\quad k\ge0.
\end{equation}

Remark that the identity $X\,Y\,\phi(\theta)=1$ implies $\tfrac{dX}{X}+\tfrac{dY}{Y}+\tfrac{\phi'(\Theta)\,d\Theta}{\phi(\Theta)}=0$. Therefore, we have
\begin{equation}
t_0-s_0=\Res\Theta\tfrac{dX}{X}+\Res\Theta\tfrac{dY}{Y}=-\Res\Theta\tfrac{\phi'(\Theta)}{\phi(\Theta)}\,d\Theta
\end{equation}
The series $\theta\tfrac{\phi'(\theta)}{\phi(\theta)}$ can be integrated as a formal power series in the variable $\theta$. Therefore, the form $\Theta\tfrac{\phi'(\Theta)}{\phi(\Theta)}\,d\Theta$ is a differential of an element of $R$, and hence, its residue is equal to zero. This completes the proof of~\eqref{eq:t0s0}, and hence that of Proposition~\ref{prop:formalsc}.
\end{proof}

\subsection{Unstable functions}

Consider the spectral curve defined above. Since $X(z)$ is a regular formal change of variables at the point $z=0$ and $Y(z)$ is a regular formal change of variables at the point $z=\infty$, we can substitute $X_i=X(z_i)$, $Y_j=Y(z_j)$ to $H^{(g)}_{m,n}$ and to treat this function as a series in $z_1,\dots,z_m,z_{m+1}^{i-1},\dots,z_{m+n}^{-1}$, that is, as a function on $\Sigma^{m+n}$ expanded as a formal power series at the point $z_1=\dots=z_m=0$, $z_{m+1}=\dots=z_{m+n}=\infty$. Our goal is to rewrite recursion of Theorem~\ref{th:formalH} it terms of these changes. The following result serves as a motivation for considering these changes.

\begin{proposition}\label{prop:unst}
The unstable correlator functions written in terms of the coordinates~$z_i$ of the (formal) spectral curve are determined explicitly by the following relations
\begin{align}\label{eq:H01}
X_1\partial_{X_1}H^{(0)}_{1,0}+\sum_{k=0}^\infty\frac{t_k}{X_1^k}=
Y_1\partial_{Y_1}H^{(0)}_{0,1}+\sum_{k=0}^\infty\frac{s_k}{Y_1^k}=\Theta(z_1),\\
\label{eq:H02}
d_1d_2H^{(0)}_{2,0}+\tfrac{dX_1\,dX_2}{(X_1-X_2)^2}=-d_1d_2H^{(0)}_{1,2}=
d_1d_2H^{(0)}_{0,2}+\tfrac{dY_1\,dY_2}{(Y_1-Y_2)^2}=\tfrac{dz_1\,dz_2}{(z_1-z_2)^2}.
\end{align}
\end{proposition}

For the proof see Sect.~\ref{sec:Lagrange} below. Remark that the expressions in~\eqref{eq:H02} can be integrated and we obtain, more explicitly,
\begin{align}
H^{(0)}_{2,0}&=\log\left(\frac{z_1^{-1}-z_2^{-1}}{X_1^{-1}-X_2^{-1}}\gamma_1^{-1}\right),\\
H^{(0)}_{1,1}&=-\log\left(1-\frac{z_1}{z_2}\right),\\
H^{(0)}_{0,2}&=\log\left(\frac{z_1-z_2}{Y_1^{-1}-Y_2^{-1}}\gamma_2^{-1}\right),
\end{align}
where the constants $\gamma_1$ and $\gamma_2$ are determined by the requirement that $H^{(0)}_{2,0}$ is vanishing for $z_1=z_2=0$ and $H^{(0)}_{0,2}$ is vanishing for $z_1=z_2=\infty$, namely,
\begin{equation}
X(z)=\gamma_1 z+O(z^2),\qquad Y(z)=\gamma_2z^{-1}+O(z^{-2}).
\end{equation}

\subsection{Formulation of the recursion (final version)}
In order to rewrite relation of Theorem~\ref{th:formalH} in terms of the change implied by the spectral curve equation we first make the following observation. Define $L_r(p,\theta)$ through the equality
\begin{equation}
\partial_\theta^r\hat\phi_p(\theta)=L_r(p,\theta)\phi(\theta)^p.
\end{equation}
Then we have, explicitly (see \cite[Section 4]{bychkov2021explicit}),
\begin{equation}\label{eq:Lr}
L_r(v,\theta)
=e^{-v\psi(\theta)}\partial_\theta^re^{v\frac{\cS(v\hbar\partial_\theta)}{\cS(\hbar\partial_\theta)}\hat\psi(\theta)}
=(\partial_\theta+v\,\psi'(\theta))^re^{v\frac{\cS(v\hbar\partial_\theta)}{\cS(\hbar\partial_\theta)}\hat\psi(\theta)-v\psi(\theta)},
\end{equation}
where $\cS$ is defined in~\eqref{eq:S}. This shows that the coefficient of any power of $\hbar^2$ in $L_r(v,\theta)$ is polynomial in $v$.

Let $f(u,Y)$ be a function which is polynomial in $u$ and a Laurent series in $Y$. Denote by $U^X$ the transformation of Theorem~\ref{th:formalH} sending $f$ to the series in~$X$ defined by
\begin{equation}
(U^Xf)(X)=\sum_{p=-\infty}^\infty X^p\sum_{r=0}^\infty  \partial_\theta^r\hat\phi_p(\theta)\bigm|_{\theta=0}\;[Y^{-p}u^r]\;
e^{u\bigl(Y\partial_Y H^{(0)}_{0,1}(Y)+\sum_{k=1}^\infty\frac{s_k}{Y^k}\bigr)}f(u,Y)
\end{equation}
(the formulation of Theorem~\ref{th:formalH} requires also to extend the transformation $U^X$ to the case when $f=\frac1{u}$).

\begin{proposition} \label{prop:Ur}
The transformation $U^X$ acts on monomials in $u$ as the following differential operator
\begin{align}
U^X(u^r f)&=-\sum_{j=0}^\infty (X\partial_X)^j\bigl(\tfrac{X}{dX}\tfrac{dY}{Y}[v^j]L_r(v,\Theta)\;f\bigr),\qquad r\ge0\\
U^X\bigl(\tfrac1u\bigr)&=\Theta+\sum_{j=1}^\infty (X\partial_X)^{j-1}\bigl([v^j]L_0(v,\Theta)\;X\partial_X\Theta\bigr).
\end{align}
where $f$ is any Laurent series in $Y$ and where we assume that $X,Y$, and $\Theta=\Theta(z)$ are related by the equation of spectral curve.
\end{proposition}

As an immediate corollary of this Proposition combined with the statement of Theorem~\ref{th:formalH} we obtain the following principal form of our recursion.

\begin{theorem}\label{th:formalHz}
With notations of Theorem~\ref{th:formalH}, the following relation holds for all $(g,m,n)$ with $2g-2+m+n>0$
\begin{equation}
	\begin{aligned}\label{eq:mainrelserHXz}
		X\partial_X H^{(g)}_{m+1,n}&=[\hbar^{(2g)}]U^X\;e^{-u\,\Theta}\cW^Y_{m,n+1}(u)\\
&=-[\hbar^{2g}]\sum_{j=0}^\infty (X\partial_X)^j \Bigg(\frac{X}{dX}\frac{dY}{Y}
\sum_{r=0}^\infty [v^j]L_r(v,\Theta)[u^r]e^{-u\,\Theta}\cW^Y_{m,n+1}(u)\\
&\qquad\qquad-\delta_{m+n,0}[v^{j+1}]L_0(v,\Theta)X\partial_X\Theta\Bigg),
	\end{aligned}
\end{equation}
where we assume that $X=X_{m+1}$, $Y=Y_{m+1}$, and $\Theta=\Theta(z_{m+1})$ are related by the equation of spectral curve.
By the symmetry between $X$ and $Y$ variables, we also have
	\begin{equation}\label{eq:mainrelserHYz}
		Y\partial_Y H^{(g)}_{m,n+1}=[\hbar^{(2g)}]\;U^Y\,e^{-u\,\Theta}\cW^X_{m+1,n}(u),
	\end{equation}
where $\cW^X$ and $U^Y$ are defined in a similar way, with the exchange of the role of $X$ and $Y$ variables and also $t$ and $s$ parameters.
\end{theorem}

The main advantage of this form of recursion comparing with that of Theorem~\ref{th:formalH} is that \emph{the right hand sides of~\eqref{eq:mainrelserHXz} and~\eqref{eq:mainrelserHYz} contain only finitely many nonzero summands for any particular triple $(g,m,n)$}. Indeed, it follows from~\eqref{eq:H01} that the genus~$0$ component of $\cT^Y_{0,1}(u;Y)$ defined by~\eqref{eq:cT} when written in $z$ coordinate coincides with $u\,\Theta(z)$, i.e.\ $\cT^Y_{0,1}(u;Y)=u\,\Theta(z)+O(\hbar^2)$. This implies (by~\eqref{eq:cWexp}, for example) that the coefficient of any particular power of $\hbar$ in~$e^{-u\,\Theta(z)}\cW^{Y}_{m,n+1}$ is polynomial in $u$ and involves a finite polynomial combination of functions of the form $H^{(g')}_{m',n'}$ and their derivatives.

A similar treatment in terms of the spectral curve equation can be applied to the equality of Theorem~\ref{th:formalW}. Define $\tilde L_r(p,\theta,w)$ through
\begin{equation}
\partial_\theta^r\tilde\phi_p(\theta,w)=\tilde L_r(p,\theta,w)\phi(\theta)^p,
\end{equation}
or, equivalently,
\begin{equation}\label{eq:tLr}
\tilde L_r(v,\theta,w)
=e^{-v\psi(\theta)}\partial_\theta^re^{v\frac{\cS(v\hbar\partial_\theta)}{\cS(\hbar\partial_\theta)}\hat\psi(\theta)+w\,\theta}
=(\partial_\theta+v\,\psi'(\theta))^re^{v\frac{\cS(v\hbar\partial_\theta)}{\cS(\hbar\partial_\theta)}\hat\psi(\theta)+w\,\theta-v\psi(\theta)}.
\end{equation}

Similarly to Proposition~\ref{prop:Ur}, we consider the transformation~$\tilde U^X$ defined by
\begin{equation}
(\tilde U^Xf)(X)=\sum_{p=-\infty}^\infty X^p\sum_{r=0}^\infty  \partial_\theta^r\tilde\phi_p(\theta,w)\bigm|_{\theta=0}\;[Y^{-p}u^r]\;
e^{u\bigl(Y\partial_Y H^{(0)}_{0,1}(Y)+\sum_{k=1}^\infty\frac{s_k}{Y^k}\bigr)}f(u,Y).
\end{equation}
Then we have, through the change implied by the spectral curve equation,
\begin{align}\label{eq:tUp}
\tilde U^X(u^r f)&=-\sum_{j=0}^\infty (X\partial_X)^j\bigl(\tfrac{X}{dX}\tfrac{dY}{Y}[v^j]\tilde L_r(v,\Theta,w)\;f\bigr),\qquad r\ge0\\
\label{eq:tU0}
\tilde U^X\bigl(\tfrac1u\bigr)&=\tfrac{e^{w\,\Theta}}{w}+\sum_{j=1}^\infty (X\partial_X)^{j-1}\bigl([v^j]\tilde L_0(v,\Theta,w)\;X\partial_X\Theta\bigr).
\end{align}
where $f$ is any Laurent series in $Y$, and we arrive at the following form of relation of~Theorem~\ref{th:formalW}.

\begin{theorem}\label{th:formalWz}
The following relation holds for all $(m,n)$
\begin{equation}
	\begin{aligned}\label{eq:mainrelserWXz}
		\cW^X_{m+1,n}(w)&=\tilde U^X\;e^{-u\,\Theta}\cW^Y_{m,n+1}(u)\\
&=-\sum_{j=0}^\infty (X\partial_X)^j \Bigg(\frac{X}{dX}\frac{dY}{Y}
\sum_{r=0}^\infty [v^j]\tilde L_r(v,\Theta,w)[u^r]e^{-u\,\Theta}\cW^Y_{m,n+1}(u)\\
&\qquad\qquad-\delta_{m+n,0}[v^{j+1}]\tilde L_0(v,\Theta,w)X\partial_X\Theta\Bigg)
+\delta_{m+n,0}\frac{e^{w\,\Theta}}{w},
	\end{aligned}
\end{equation}
where we assume that $X=X_{m+1}$, $Y=Y_{m+1}$, and $\Theta=\Theta(z_{m+1})$ are related by the equation of spectral curve.
\end{theorem}

\subsection{Lagrange inversion formula and principal identity}\label{sec:Lagrange}
This section is devoted to the proofs of relations of the previous section. Our computations are based on Lagrange Inversion Theorem and Principal Identity used in \cite{bychkov2021explicit}. The invariance of residues implies that these tools can be applied to the situation when all considered functions belong to the ring~$R$. Thus, for the proof of~\eqref{eq:tUp} we follow step by step the computations made
in~\cite[Corollary 4.4]{bychkov2021explicit}:

{\renewcommand{\arraystretch}{2}
\begin{longtable}{>{$\displaystyle}l<{$}>{$\displaystyle}l<{$}}
			\tilde U^X \bigl(u^r f\bigr)&=\sum_{p=-\infty}^\infty X^p\sum_{k=0}^\infty \partial_\theta^r\tilde\phi_p(\theta)\bigm|_{\theta=0}\;
			[Y^{-p}u^k]\;e^{u\,\Theta}\;u^r f(Y)\\
			&=\sum_{p=-\infty}^\infty X^p\sum_{k=0}^\infty  [Y^{-p}]\;\partial_\theta^{r+k}\tilde\phi_p(\theta)\bigm|_{\theta=0}\;\frac{\Theta^k}{k!}\;f(Y)\\
			&=\sum_{p=-\infty}^\infty X^p\sum_{r=0}^\infty  [Y^{-p}]\;\partial_\theta^r\tilde\phi_p(\theta)\bigm|_{\theta=\Theta}\;f(Y)\\
			&=-\sum_{p=-\infty}^\infty X^p\Res \frac{dY}{Y}\;Y^p\phi(\Theta)^p \tilde L_{r}(p,\Theta)\;f(Y)\\
			&=-\sum_{p=-\infty}^\infty X^p \Res \frac{dY}{Y}\;X^{-p} \tilde L_{r}(p,\Theta)\;f(Y)\\
			&=-\sum_{p=-\infty}^\infty X^p [X^p] \frac{dY}{Y} \frac{X}{dX} \tilde L_{r}(p,\Theta)\;f(Y)\\
			&=-\sum_{j=0}^\infty \bigl(X\partial_X\bigr)^j\Bigl(\frac{dY}{Y} \frac{X}{dX} [v^j]\tilde L_{r}(v,\Theta)\;f(Y)\Bigr).
\end{longtable}}
This proves~\eqref{eq:tUp}. In the case~\eqref{eq:tU0}, denoting
	\begin{equation}
		G(X)=\tilde U^X\bigl(\tfrac1u\bigr)=\sum_{p=-\infty}^\infty X^p\sum_{r=0}^\infty \partial_\theta^r\tilde\phi_p(\theta)\bigm|_{\theta=0}\;
		[Y^{-p}u^r]\;\frac{e^{u\,\Theta}}{u},
	\end{equation}
	we have
{\renewcommand{\arraystretch}{2}
\begin{longtable}{>{$\displaystyle}l<{$}>{$\displaystyle}l<{$}}
			X\partial_X G(X)
			&=\sum_{p=-\infty}^\infty p\,X^p\sum_{r=0}^\infty \partial_\theta^r\tilde\phi_p(\theta)\bigm|_{\theta=0}\;
			[Y^{-p}u^r]\;\frac{e^{u\,\Theta}}{u}\\
			&=-\sum_{p=-\infty}^\infty \,X^p\sum_{r=0}^\infty \partial_\theta^r\tilde\phi_p(\theta)\bigm|_{\theta=0}\;
			[Y^{-p}u^r]\;Y\partial_Y\frac{e^{u\,\Theta}}{u}\\
			&=-\sum_{p=-\infty}^\infty \,X^p\sum_{r=0}^\infty \partial_\theta^r\tilde\phi_p(\theta)\bigm|_{\theta=0}\;
			[Y^{-p}u^r]\;e^{u\,\Theta}\;Y\partial_Y\Theta\\
			&=-\tilde U^X(Y\partial_Y\Theta)\\
			&=\sum_{j=0}^\infty \bigl(X\partial_X\bigr)^j\Bigl(\frac{dY}{Y} \frac{X}{dX} [v^j]\tilde L_{0}(v,\Theta)\;Y\partial_Y\Theta\Bigr)\\
			&=e^{w\,\Theta}X\partial_X\Theta+\sum_{j=1}^\infty \bigl(X\partial_X\bigr)^j\Bigl([v^j]\tilde L_{0}(v,\Theta)\;X\partial_X\Theta\Bigr)\\
			&=X\partial_X\Biggl(\frac{e^{w\,\Theta}}{w}+\sum_{j=1}^\infty \bigl(X\partial_X\bigr)^{j-1}\Bigl([v^j]\tilde L_{0}(v,\Theta)\;X\partial_X\Theta\Bigr)\Biggr).
\end{longtable}}
	Integrating both sides we obtain~\eqref{eq:tU0} up to an additive constant. The integration constant can be obtained by taking the residues of both sides of the obtained equality multiplied by $\frac{dX}{X}$.

Theorem~\ref{th:formalWz} follows from Theorem~\ref{th:formalW} by the proven Eqs.~\eqref{eq:tUp}--\eqref{eq:tU0}. Taking the coefficient of $w^0$ in~\eqref{eq:tUp}, \eqref{eq:tU0}, and~\eqref{eq:mainrelserHXz} we obtain equalities of Proposition~\ref{prop:Ur} and Theorem~\ref{th:formalHz}, respectively.

\medskip
Let us prove finally Proposition~\ref{prop:unst}. Remark that the above computations do not use its statement, but it is needed to be assured that the right hand side of~\eqref{eq:mainrelserHXz} has finitely many nonzero summands. Let us denote
\begin{equation}
\Theta^X(X)=X\partial_X H^{(0)}_{1,0}(X)+\sum_{k=1}^\infty\frac{t_k}{X^k},\qquad
\Theta^Y(Y)=Y\partial_Y H^{(0)}_{0,1}(Y)+\sum_{k=1}^\infty\frac{s_k}{Y^k}
\end{equation}
Taking the coefficient of $w^0\hbar^0$ on both sides of~\eqref{eq:mainrelserWX} for $m=n=0$ we get
\begin{equation}
\Theta^X(X)=\sum_{p=-\infty}^\infty X^p\sum_{r=0}^\infty \partial_\theta^r\phi(\theta)^p\bigm|_{\theta=0}\;
			[Y^{-p}u^r]\;\frac{e^{u\,\Theta^Y(Y)}}{u}
\end{equation}
Repeating the computations similar to those above we obtain that the right hand side is equal to $\Theta^Y(Y)$ with the substitution inverse to the change $X(Y)=\frac1{Y\,\phi(\Theta^Y(Y))}$. In other words, the following relation holds true
\begin{equation}
\Theta^X\bigl(\tfrac1{Y\,\phi(\Theta^Y(Y))}\bigr)=\Theta^Y(Y).
\end{equation}
The change $X(Y)=\frac1{Y\,\phi(\Theta^Y(Y))}$ involves the terms with both positive and negative arbitrary large exponents of~$Y$, and in order to attain a meaning to this change we consider it as a change in the ring~$R$. We conclude that the functions $\Theta^X$ and $\Theta^Y$ are identified by this change and the inverse change is given by
$Y=\tfrac{1}{X\,\phi(\Theta^X(X))}$. Indeed,
\begin{equation}
\frac{1}{X(Y)\,\phi(\Theta^X(X(Y)))}=\frac{1}{X(Y)\,\phi(\Theta^Y(Y))}=Y.
\end{equation}

We observe now that the function $\Theta^X$ possesses the following properties.
\begin{itemize}
\item It has an expansion of the form
\begin{equation}
\Theta^X(X)=\sum_{k=1}^\infty\frac{t_k}{X^k}+O(X).
\end{equation}
\item The same function rewritten in the variable $Y$ related to $X$ by the change inverse to $Y=\tfrac{1}{X\,\Theta^X(X)}$ has an expansion of the form
\begin{equation}
\Theta^X(X(Y))=\sum_{k=1}^\infty\frac{s_k}{Y^k}+O(Y).
\end{equation}
\end{itemize}

It is easy to see that these two conditions determine the series $\Theta^X$ uniquely: this is yet another application of implicit function theorem similar to that one used in the proof of Proposition~\ref{prop:formalsc}. The function $\Theta(z)$ entering the definition of the spectral curve and rewritten in the coordinate $X$ does satisfy these conditions. Therefore, the functions $\Theta^X(X)$ and $\Theta^Y(Y)$ coincide with the function $\Theta(z)$ rewritten in the coordinates~$X$ and~$Y$ of the spectral curve, respectively. This proves Eq.~\eqref{eq:H01} of Proposition~\ref{prop:unst}.

\medskip
Eqs.~\eqref{eq:H02} of Proposition~\ref{prop:unst} corresponds to the coefficient of $w^0\hbar^0$ on both sides of~\eqref{eq:mainrelserWX} or~\eqref{eq:mainrelserHXz} for the case $m+n=1$. The only graph that contributes to the sum on the right hand side in these cases is a single tree with one black vertex connected with two white vertices. Taking into account the singular corrections entering~\eqref{eq:cT} we obtain (for the cases $(m,n)=(1,0)$ and $(0,1)$, respectively)
\begin{equation}
\begin{gathered}
X_2\partial_{X_2}H^{(0)}_{2,0}+\tfrac{X_1}{X_2-X_1}=-\tfrac{X_2}{dX_2}\tfrac{dY_2}{Y_2}Y_2\partial_{Y_2}H^{(0)}_{1,1},\\
X_1\partial_{X_1}H^{(0)}_{1,1}=-\tfrac{X_1}{dX_1}\tfrac{dY_1}{Y_1}\bigl(Y_1\partial_{Y_1}H^{(0)}_{0,2}+\tfrac{Y_2}{Y_1-Y_2}\bigr).
\end{gathered}
\end{equation}
Where we assume that $X_i$ is related to $Y_i$ by the equation of spectral curve. Differentiating the first line in $X_1$ and the second one in $X_2$ we obtain
\begin{equation}
d_1d_2H^{(0)}_{2,0}+\tfrac{dX_1dX_2}{(X_1-X_2)^2}=-d_1d_2H^{(0)}_{1,1}=d_1d_2H^{(0)}_{0,2}+\tfrac{dY_1dY_2}{(Y_1-Y_2)^2}.
\end{equation}
Let us rewrite the obtained $2$-differential in $z$ coordinates of the spectral curve. Using the fact that $X(z)$ is a regular change of coordinates at $z=0$, we get by local computations from the first equality that it has the form
\begin{equation}
\frac{dz_1dz_2}{(z_1-z_2)^2}+(\text{regular series in $z_1,z_2$})\;dz_1dz_2.
\end{equation}
On the other hand, from the last equality we get by similar local computations using the fact that $Y(z)$ is a regular change at $z=\infty$ that this $2$-differential has the form
\begin{equation}
\frac{dz_1dz_2}{(z_1-z_2)^2}+(\text{regular series in $z_1^{-1},z_2^{-1}$})\,\frac{dz_1dz_2}{z_1^2z_2^2}.
\end{equation}
This is only possible if the correction terms in both presentations are equal to zero. This proves Eq.~\eqref{eq:H02} of Proposition~\ref{prop:unst}.

\subsection{Reduced recursion, splitting of poles, and the strategy of computation of \texorpdfstring{$H^{(g)}_{m,n}$}{Hgmn}'s}

The right hand side in~\eqref{eq:mainrelserHXz} has finitely many nonzero summands only, and each summand is a polynomial combination of the functions of the form $H_{m'n'}^{(g')}$ and their derivatives. The function $H^{(g)}_{m,n+1}$ also contributes to the right hand side of~\eqref{eq:mainrelserHXz}, namely, to the summand with $j=0$: we have from~\eqref{eq:mainrelserHXz} and~\eqref{eq:Lr}
\begin{align}
[\hbar^{2g}]\cW^Y_{m,n+1}&=Y\partial_Y H^{(g)}_{m,n+1}+O(u),\\
L_0(v,\theta)&=1+O(v),\\
L_r(v,\theta)&=O(v),\quad r\ge1.
\end{align}
It follows that the summand of~\eqref{eq:mainrelserHXz} with $j=0$, disregarding the $\delta_{m+n,0}$-part, is equal to
\begin{equation}
-\frac{X}{dX}\frac{dY}{Y} Y\partial_Y H^{(g)}_{m,n+1}\frac{dY}{Y}=-X\partial_X H^{(g)}_{m,n+1}.
\end{equation}
All the summands with $j>0$ as we as the $\delta_{m+n,0}$-contribution involve only those functions $H_{m'n'}^{(g')}$ satisfying $2g'-2+m'+n'<2g-2+m+n+1$, so that we may assume that they are already computed in the previous steps of computations. It follows that the formula can be rewritten in the form
\begin{align}\label{eq:XdXHsum}
&X\partial_X H^{(g)}_{m+1,n}+X\partial_X H^{(g)}_{m,n+1}=\\ \nonumber
&\phantom{=}-[\hbar^{2g}]\sum_{j=1}^\infty (X\partial_X)^j \Bigg(\frac{X}{dX}\frac{dY}{Y}\sum_{r=0}^\infty [v^j]L_r(v,\Theta)[u^r]e^{-u\,\Theta}\cW^Y_{m,n+1}(u)\\ \nonumber
&\phantom{==}-\delta_{m+n,0}[v^{j+1}]L_0(v,\Theta)X\partial_X\Theta\Bigg)+\delta_{m+n,0}[\hbar^{2g}v^1]L_0(v,\Theta)X\partial_X\Theta.
\end{align}
Let us study the last term in the above expression:
\begin{align}\label{eq:j0term}
&[\hbar^{2g}v^1]L_0(v,\Theta)X\partial_X\Theta = [\hbar^{2g}v^1]\left.\left(e^{v\left(\frac{\cS(v\hbar\partial_\theta)}{\cS(\hbar\partial_\theta)}\hat\psi(\theta)-\psi(\theta)\right)}\right)\right|_{\theta=\Theta}X\partial_X\Theta \\ \nonumber
&\phantom{=}=[\hbar^{2g}]\left.\left(\dfrac{1}{\cS(\hbar\partial_\theta)}\hat\psi(\theta)-\psi(\theta)\right)\right|_{\theta=\Theta}X\partial_X\Theta.
\end{align}

Integrating~\eqref{eq:XdXHsum} with~\eqref{eq:j0term} substituted, we obtain our following main relation, which considerably simplifies the inductive computation of the $(m,n)$-point functions:
\begin{theorem}\label{th:mainrecursion}
For any triple $(g,m,n)$ with $2g-2+m+n\ge0$ we have
\begin{align}\label{eq:mairecX}
&H^{(g)}_{m+1,n}+H^{(g)}_{m,n+1}=\\ \nonumber
&\phantom{=}-[\hbar^{2g}]\sum_{j=1}^\infty (X\partial_X)^{j-1} \Bigg(\frac{X}{dX}\frac{dY}{Y}\sum_{r=0}^\infty [v^j]L_r(v,\Theta)[u^r]e^{-u\,\Theta}\cW^Y_{m,n+1}(u)\\ \nonumber
&\phantom{=}-\delta_{m+n,0}[v^{j+1}]L_0(v,\Theta)X\partial_X\Theta\Bigg)\\ \nonumber
&\phantom{=}+\delta_{m+n,0}\int[\hbar^{2g}]\left.\left(\dfrac{1}{\cS(\hbar\partial_\theta)}\hat\psi(\theta)-\psi(\theta)\right)d\theta\;\right|_{\theta=\Theta}+{\rm const},
\end{align}
where $X=X_{m+1}$, $Y=Y_{m+1}$, and ${\rm const}$ is certain function in $X_M,Y_N$ and independent of $X$.
\end{theorem}

Remark that there is a formally different expression for the right hand side due to the symmetry between $X$--$Y$ variables,
\begin{align}\label{eq:mairecY}
&H^{(g)}_{m+1,n}+H^{(g)}_{m,n+1}=\\ \nonumber
&\phantom{=}-[\hbar^{2g}]\sum_{j=1}^\infty (Y\partial_Y)^{j-1} \Bigg(\frac{Y}{dY}\frac{dX}{X}\sum_{r=0}^\infty [v^j]L_r(v,\Theta)[u^r]e^{-u\,\Theta}\cW^X_{m+1,n}(u)\\ \nonumber
&\phantom{=}-\delta_{m+n,0}[v^{j+1}]L_0(v,\Theta)Y\partial_Y\Theta\Bigg)\\ \nonumber
&\phantom{=}+\delta_{m+n,0}\int[\hbar^{2g}]\left.\left(\dfrac{1}{\cS(\hbar\partial_\theta)}\hat\psi(\theta)-\psi(\theta)\right)d\theta\;\right|_{\theta=\Theta}+{\rm const}.
\end{align}

Let us look more closely on the right hand side of~\eqref{eq:mairecX} at the coordinate $z=z_{m+1}$ on the spectral curve. We observe that $H^{(g)}_{m+1,n}$ involves positive powers of the variable~$z$ only, and $H^{(g)}_{m,n+1}$ involves negative powers of~$z$ (and the constant term in~$z$ vanishes: exactly this requirement provides the choice of the constant $\rm const$ on the right hand side). We conclude that~\eqref{eq:mairecX} can be regarded as a recursion relation: \emph{it allows one not only to relate  $H^{(g)}_{m+1,n}$ and  $H^{(g)}_{m,n+1}$ but also to compute both of them!} Namely, we compute
$H^{(g)}_{m+1,n}$ and  $H^{(g)}_{m,n+1}$ as the regular and the polar parts, respectively, of the right hand side in~\eqref{eq:mairecX} with respect to $z$ coordinate.

\begin{remark}
Relation of Theorem~\ref{th:mainrecursion} treats both its sides as formal power series in $(t,s)$ variables and also $X_1\dots,X_m,Y_{m+2},\dots,Y_{m+n+1}$ whose coefficients are Laurent polynomials in~$z$. In fact, one can see by induction that there is a stronger rationality assertion: \emph{$H^{(g)}_{m,n}$ can be represented in $z$-coordinates as a power series in $(t,s)$ variables whose coefficients are rational functions in $z_1,\dots,z_{m+n}$ with only possible poles at $z_i=\infty$, $z_j=0$, and on the diagonals $z_i=z_j$ for $i\in\{1,\dots,m\}$, $j\in\{m+1,\dots,m+n\}$. In other words, it can be represented as a ratio of a polynomial in $z_1,\dots,z_m,z_{m+1}^{-1},\dots,z_{m+n}^{-1}$ and a product of factors of the from $1-z_iz_j^{-1}$. The degree of the denominator is uniformly bounded for each particular $(g,m,n)$ and the degree of the numerator grows with the growth of the degree of $(t,s)$-monomial.} We use exactly this form of the functions $H^{(g)}_{m,n}$ in our numerical computer experiments.
\end{remark}

\section{Rationality of \texorpdfstring{$(m,n)$}{(m,n)}-point functions and loop equations}\label{sec:rat}

Up to this moment, in Sect.~\ref{sec:formalprel} and~\ref{sec:formalfinal}, we regarded all $n$-point functions as formal power series and no restrictions on the initial data $(t,s,\psi)$ of the problem have been assumed. We now impose certain natural analytic assumptions on the initial data ensuring rationality of all $n$-point functions on the spectral curve. The rationality is crucial for the study of topological recursion and loop equations. Without this property the very discussion of topological recursion is senseless: it analyses the behavior of the analytic extension of the functions to the points different from the point of the expansion of these functions regarded as generating series.

\subsection{Rational spectral curve}\label{sec:spectral}

We say that the formal spectral curve introduced in Sect.~\ref{sec:spectralcurveformal} is algebraic if the forms $\frac{dX}{X}$ and $\frac{dY}{Y}$ extend as global meromorphic forms on~$\Sigma$ and $\Theta(z)$ is a Laurent polynomial. Moreover, we wish that its dependence in $(t,s)$ parameters is also algebraic. It means that the coefficients of the rational forms $\frac{dX}{X}$ and $\frac{dY}{Y}$ and the Laurent polynomial $\Theta$ are defined not only as formal power series but also as true algebraic functions and their specializations at arbitrary complex numbers with sufficiently small absolute values are well defined. For the convenience of the reader we provide the definition of the spectral curve in the algebraic case which is independent of the formal case of Sect.~\ref{sec:spectralcurveformal}.

Let $t=(t_1,\dots,t_d,0,0,\dots)$, $s=(s_1,\dots,s_e,0,0,\dots)$ be two sets of complex parameters such that only a finite number of them are nonzero. Let $\phi(\theta)=e^{\psi(\theta)}=1+\sum_{k=1}^\infty c_k\theta^k$ be a power series satisfying the property that $\psi'(\theta)=\phi'(\theta)/\phi(\theta)$ is rational. A typical example is when either $\phi$ or $\psi$ is a polynomial (or a rational function).

\begin{definition}\label{def:spectralcurve}
The \emph{spectral curve associated with the data $(t,s,\phi)$} is $\Sigma=\C P^1$ with global affine coordinate~$z$ and three functions $X,Y,\Theta$ on (some open domains of) $\Sigma$ satisfying the following relations.
\begin{itemize}
\item We have
\begin{equation}\label{eq:XYTheta}
X\,Y\,\phi(\Theta)=1.
\end{equation}
\item $X$ is defined and holomorphic in a neighborhood of the disk $|z|\le 1$, has a simple zero at $z=0$ and no other zeroes in that disk. In other words, $X$ forms a global holomorphic coordinate on the disk $|z|\le1$. Similarly, $Y$ is defined and holomorphic in a neighborhood of the disk $|z|\ge 1$, has a simple zero at $z=\infty$ and no other zeroes in that disk. In other words, $Y$ forms a global holomorphic coordinate on the disk $|z|\ge1$.
\item The $1$-forms $\frac{dX}{X}$ and $\frac{dY}{Y}$ extend as global rational $1$-forms on the whole spectral curve.
\item $\Theta(z)$ is a Laurent polynomial. Moreover, its Laurent expansions at $z=0$ and $z=\infty$ in the corresponding local coordinates are given by
\begin{align}\label{eq:ThetaX}
\Theta(z)&=\sum_{k=1}^d t_kX^{-k}+O(z),\quad z\to 0,\\
\label{eq:ThetaY}
\Theta(z)&=\sum_{k=1}^e s_kY^{-k}+O(z^{-1}),\quad z\to \infty.
\end{align}
\end{itemize}
\end{definition}

\begin{proposition}\label{prop:speccurv}
For a given $\phi$, if the absolute values of $(t,s)$-parameters are small enough then the requirements on the spectral curve define it uniquely up to a multiplication of the coordinate $z$ by a nonzero constant.
\end{proposition}

A choice for a rescaling of~$z$ is not important, since it does not change the analytic dependence between $X,Y$, and $\Theta$ functions. However, it can be fixed, if needed, by an additional relation
\begin{equation}\label{eq:Xzcoef}
X(z)=z+O(z^2).
\end{equation}

\begin{proof}
We see from~\eqref{eq:ThetaX} and~\eqref{eq:ThetaY} that $\Theta$ has a pole of order~$d$ at $z=0$ and a pole of order~$e$ at $z=\infty$, i.e. it has the form
\begin{equation}\label{eq:thetaform}
\Theta(z)=\sum_{-d\le k\le e} \alpha_k z^k.
\end{equation}
We will show that the coefficients of $\Theta(z)$ obey certain polynomial equation allowing one to express $\alpha_0$ as a function in the remaining $\alpha$-parameters. So we can take the coefficients $\alpha_{-d},\dots,\alpha_{-1},\alpha_1,\dots,\alpha_e$ as an independent set of parameters of the problem. We express $t_i$, $s_j$ as (algebraic) functions in these parameters. Then, we apply the inverse function theorem to express $\alpha$-coordinates as functions in $t$ and $s$ parameters.

Having in mind the necessity to apply the inverse function theorem, we assume that the coefficients $\alpha_k$ are small enough. An explicit estimate on the absolute values of these coefficients will be clear from the arguments below.
Taking the logarithmic derivatives of the two sides of~\eqref{eq:XYTheta} we obtain the following equality of meromorphic $1$-forms on the spectral curve,
\begin{equation}\label{eq:dlogxdlogY}
\frac{dX}{X}+\frac{dY}{Y}=-\frac{d\phi(\Theta)}{\phi(\Theta)}.
\end{equation}
The two summands on the left hand side can be recovered as the contributions of the poles of the right hand side outside the unit circle and inside it, respectively. So, we \emph{define}
\begin{align}\label{eq:Xpoles}
\frac{dX(z)}{X(z)}&=\frac{dz}{z}-\sum_{a:\phi(\Theta(a))=0,|a|>1}\res_{\tilde z=a}\frac{d\phi(\Theta(\tilde z))}{\phi(\Theta(\tilde z))}\frac{dz}{\tilde z-z},\\ \label{eq:Ypoles}
\frac{dY(z)}{Y(z)}&=-\frac{dz}{z}-\sum_{a:\phi(\Theta(a))=0,|a|<1}\res_{\tilde z=a}\frac{d\phi(\Theta(\tilde z))}{\phi(\Theta(\tilde z))}\frac{dz}{\tilde z-z}.
\end{align}
Equivalently, these relations can be written as follows
\begin{equation}
-\frac{1}{2\pi i}\int_{|\tilde z|=1}\frac{d\phi(\Theta(\tilde z))}{\phi(\Theta(\tilde z))}\frac{dz}{\tilde z-z}=
\begin{cases}
\frac{dz}{z}-\frac{dX(z)}{X(z)},&|z|<1,\\
\frac{dz}{z}+\frac{dY(z)}{Y(z)},&|z|>1,
\end{cases}
\end{equation}
where the integration contour $|z|=1$ is oriented counterclockwise. The forms $\frac{dX}{X}$ and $\frac{dY}{Y}$ determine the functions $X$ and $Y$ themselves uniquely up to multiplicative constants,
\begin{equation}
X=e^{\int\frac{dX}{X}},\quad Y=e^{\int\frac{dY}{Y}}.
\end{equation}
The integration constants can be fixed by~\eqref{eq:XYTheta} and, for example,~\eqref{eq:Xzcoef}.

Next, we observe that $X$ is a local coordinate at $z=0$ and $Y$ is a local at $z=\infty$. Expanding $\Theta$ in these coordinates, we get
\begin{align}\label{eq:Thetat0}
\Theta(z)&=\sum_{k=1}^d \frac{t_k}{ X^{k}} +t_0+ O(z),\quad z\to 0,\\
\label{eq:Thetas0}
\Theta(z)&=\sum_{k=1}^e \frac{s_k}{ Y^{k}} +s_0+ O(z^{-1}),\quad z\to \infty.
\end{align}
The coefficients $t_k$, $s_k$ of these expansions are expressed as functions in $\alpha$-parameters. More explicitly, we have
\begin{equation}\label{eq:tsequatins}
t_k=\res_{z=0}\Theta(z)X(z)^{k-1}d X(z),\qquad s_k=\res_{z=\infty}\Theta(z)Y(z)^{k-1}d Y(z)
\end{equation}
Expansions~\eqref{eq:Thetat0},~\eqref{eq:Thetas0} differ from~\eqref{eq:ThetaX},~\eqref{eq:ThetaY} by the presence of the constant terms $t_0$, $s_0$. Vanishing of $t_0$ and $s_0$ provides algebraic relations between parameters $\alpha_k$. In fact, these two equations are equivalent to one another due to the following identity:
\begin{equation}
\begin{aligned}
t_0-s_0&=\res_{z=0}\Theta(z)\frac{dX(z)}{X(z)}-\res_{z=\infty}\Theta(z)\frac{dY(z)}{Y(z)}\\
&=\frac1{2\pi i}\int_{|z|=1}\Theta(z)\frac{dX(z)}{X(z)}+\frac1{2\pi i}\int_{|z|=1}\Theta(z)\frac{dY(z)}{Y(z)}\\
&=-\frac1{2\pi i}\int_{|z|=1}\Theta(z)\frac{d\phi(\Theta(z))}{\phi(\Theta(z))}\\
&=-\frac1{2\pi i}\int_{\Gamma}\theta\frac{d\phi(\theta)}{\phi(\theta)}=0,
\end{aligned}
\end{equation}
where the integration contour $\Gamma$ is the image of the unit circle $|z|=1$ under the map~$\Theta$. Since the coefficients of the Laurent polynomial $\Theta(z)$ are small, we may assume that the contour $\Gamma$ belongs a small disk centered at the origin in the $\theta$-plane. Moreover, since $\phi(0)=1$, we may assume that this disk is small enough such that $\theta\frac{d\phi(\theta)}{\phi(\theta)}$ is holomorphic inside the disk and hence its integral along any closed contour in the disk vanishes.

Thus, regarding~\eqref{eq:tsequatins} as implicit algebraic equations on the $\alpha$-parameters along with the equation $t_0=0$ (or an equivalent one $s_0=0$) we express by implicit function theorem $\alpha$-parameters as holomorphic functions in $(t,s)$-parameters. This resolves the equation of spectral curve.
\end{proof}

\begin{remark}
It follows from the above arguments that if $(t,s)$ parameters tend to zero, all the poles of $\frac{dX}{X}$ except that at $z=0$ converge to the point $z=\infty$, and all poles of $\frac{dY}{Y}$ except that at $z=\infty$ converge to the point $z=0$.
\end{remark}

\begin{remark}
If $\phi$ is rational then $X$ and $Y$ are also rational functions. Namely, if we represent $\phi(\Theta(z))$ as the product of linear factors of the form $(z-a_i)^{\pm1}$ multiplied by a monomial in $z$, then $X$ and $Y$ in the product $X\,Y=1/\phi(\Theta)$ absorb those factors with $|a_i|>1$ and $|a_i|<1$, respectively. In the general case, however, $\frac{dX}{X}$ and $\frac{dY}{Y}$ might have nonzero residues and the holomorphic extension of $X$ and $Y$ functions may meet logarithmic singularities.
\end{remark}

\subsection{Rationality of \texorpdfstring{$(m,n)$}{(m,n)}-point functions}\label{sec:separ}
Let the data $(t,s,\psi)$ of the spectral curve satisfy the analytic properties of the previous section, namely, $t=(t_1,\dots,t_d,0,0,\dots)$, $s=(s_1,\dots,s_e,0,0,\dots)$, and $\psi'(\theta)$ is rational. Assume that an $\hbar^2$-deformation $\hat\psi$ of $\psi$ is chosen such that the coefficient of any positive power of $\hbar$ in $\hat\psi$ is a derivative of a rational function. This implies, in particular, that the last summands of~\eqref{eq:mairecX} and~\eqref{eq:mairecY} are rational. We refer everywhere below to the assumptions made as the \emph{natural analytic assumptions} on the data $(t,s,\hat\psi)$ of the problem.



Let us treat $X_i=X(z_i)$ and $Y_i=Y(z_i)$ as the local coordinates at the corresponding points $z_i=0$ or $z_i=\infty$, respectively, on the $i$th copy of the spectral curve.
In that way we regard  $H_{m,n}^{(g)}$ as a function on $\Sigma^{m+n}$ expanded as a formal power series at the point $z_1=\dots=z_m=0$, $z_{m+1}=\dots=z_{m+n}=\infty$. The following theorem describes the properties of analytic extension of $H_{m,n}^{(g)}$ to the spectral curve.

\begin{theorem}\label{th:rat}
Assume that the parameters $t_i,s_j$ are chosen small enough. Then for any triple $(g,m,n)$ satisfying $2g-2+m+n>0$ the function $H_{m,n}^{(g)}$ extends as a global rational function to $\Sigma^{m+n}$ and Equations~\eqref{eq:mairecX} and~\eqref{eq:mairecY} hold true as equalities of rational functions.

Moreover, $H^{(g)}_{m,n}$ is holomorphic in the multidisk $|z_i|<1$, $i=1,\dots,m$, $|z_j|>1$, $j=m+1,\dots,m+n$. In fact, if we regard $H_{m,n}^{(g)}$ as a rational function in $z_i$, $i=1,\dots,m$, then it might have poles on the diagonals $z_i=z_j$, $j=m+1,\dots,m+n$ (this does not contradict to the above assertion). All the other poles converge to $\infty$ as the parameters $t_k,s_k$ tend to zero. Similarly, if we regard $H_{m,n}^{(g)}$ as a rational function in $z_j$, $j=m+1,\dots,m+n$, then it might have poles on the diagonals $z_j=z_i$, $i=1,\dots,m$, and all the other poles converge to zero as the parameters $t_k,s_k$ tend to zero.

In the unstable cases $2g-2+m+n\le0$ the correlator functions are determined explicitly by~\eqref{eq:H01}--\eqref{eq:H02}.
\end{theorem}

\begin{proof}

We argue by induction in $g$ and $m+n$. Consider Eq.~\eqref{eq:mairecX} as an equality in the ring~$R$. By induction hypothesis, every term on the right hand side is rational in~$z$ coordinates. Since the right hand side contains finitely many terms, we conclude that the whole right hand side is rational as a function in $z_1,\dots,z_{m+n+1}$.

Let us look  more closely at the dependence of all the terms on the right hand side of~\eqref{eq:mairecX} in each particular variable $z_i$. Every term is holomorphic in~$z_i$ in the domain $|z_i|\le1$ for $i\in M$, and holomorphic in the domain $|z_i|\ge1$ for $i\in N$. Hence, the same holds true for the whole right hand side. The dependence in $z_{m+1}$ is more complicated. The terms might have poles both for $|z_{m+1}|<1$ and for $|z_{m+1}|>1$.

Let us denote by $F(z)$ the right hand side of~\eqref{eq:mairecX} regarded as a rational function in $z=z_{m+1}$ and represent it as $F(z)=F^+(z)+F^-(z)+c$ where $F^+$ is holomorphic in $|z|\le 1$ and $F^-$ is holomorphic in $|z|\ge1$, with the normalization $F^+(0)=F^-(\infty)=0$, and $c$ is a constant. More explicitly, we have
\begin{equation}
	\begin{aligned}
		F_+(z)&=\frac1{2\pi i}\int_{|\tilde z|=1}f(\tilde z)\frac{d\tilde z}{\tilde z-z}\frac{z}{\tilde z},\qquad |z|<1,\\
		F_-(z)&=-\frac1{2\pi i}\int_{|\tilde z|=1}f(\tilde z)\frac{d\tilde z}{\tilde z-z},\qquad |z|>1.
	\end{aligned}
\end{equation}
This integral representation shows that both $F^+$ and $F^-$ depend regularly in $t$ and $s$ as these parameters tend to zero. This implies that the Laurent expansion of $F^+$ ($F^-$) in the ring $R$ contain only positive (respectively, negative) powers of $z$. This implies the equalities $F^+=H^{(g)}_{m+1,n}$ and  $F^-=H^{(g)}_{m,n+1}$ as it is explained at the discussion after Theorem~\ref{th:mainrecursion}. This proves Theorem~\ref{th:rat}.
\end{proof}

Remark that the arguments above provide not only the proof of rationality of the functions $H^{(g)}_{m+1,n}$ and  $H^{(g)}_{m,n+1}$ but also an explicit inductive procedure for their computations: $H^{(g)}_{m+1,n}$ regarded as a rational function in $z=z_{m+1}$ absorbs the principal parts of the poles of the right hand side of~\eqref{eq:mairecX} situated in the domain $|z|>1$ (including those at $z_i$ for $i\in N$), while  $H^{(g)}_{m,n+1}$ absorbs the principal parts of the poles situated in the domain $|z|<1$ (including those at $z_i$ for $i\in M$).

It is also useful to note that the contour integral used in the proof above is the analytic analogue of the operator $\Res$ applied in the formal case of the ring $R$ in Sect~\ref{sec:formalfinal}.

\subsection{Possible poles and linear loop equations}

It is sometimes convenient to represent the equalities of Theorem~\ref{th:formalHz} as equalities between meromorphic differential forms rather than functions.
Consider the following operators acting in the space of meromorphic $1$-forms on~$\Sigma$:
\begin{align}
D_X:\omega\mapsto d\left(\frac{\omega}{dX/X}\right),\qquad
D_Y:\omega\mapsto d\left(\frac{\omega}{dY/Y}\right).
\end{align}
They are the counterparts of the corresponding operators $X\partial_X$ and $Y\partial_Y$ acting in the space of functions. Then~\eqref{eq:mainrelserHXz} can be rewritten in yet another equivalent form
\begin{multline}\label{eq:dXXPsum}
d H^{(g)}_{m+1,n}=-[\hbar^{2g}]\sum_{j=0}^\infty D_X^j
\Bigg(\sum_{r=0}^\infty [v^j]L_r(v,\Theta)[u^r]e^{-u\,\Theta}\cW^Y_{m,n+1}(u)\frac{dY}{Y}\\
-\delta_{m+n,0}[v^{j+1}]L_0(v,\Theta)d\Theta\Bigg).
\end{multline}
where the differential on the left hand side is taken with respect to the variable $X=X_{m+1}$.

The operator $D_X$ acting in the space of meromorphic differentials on $\Sigma$ has poles at zeroes of the form $dX/X$, that is, at the critical points of $X$.

\begin{definition} We denote by $\Xi^X$ the space of meromorphic differentials defined in a neighborhood of the zero locus of $dX/X$ on $\Sigma$ and spanned by the differentials of the form $D_X^k\omega$ where $k=0,1,2,\dots$ and $\omega$ is holomorphic.  We denote by $\Xi^Y$ the space of meromorphic differentials defined in a neighborhood of the zero locus of $dY/Y$ and spanned by the differentials of the form $D_Y^k\omega$ where $k=0,1,2,\dots$ and $\omega$ is holomorphic.
\end{definition}

The condition $\alpha\in\Xi^X$ implies restriction on the principal part of the poles of~$\alpha$. For example, assume that the $dX/X$ has a simple zero at the given point. Then we may choose a local holomorphic coordinate $\zeta$ at this point such that $dX/X=\zeta\,d\zeta$. Then $\Xi^X$ is spanned by the forms holomorphic at $\zeta=0$ and the forms $\frac{d\zeta}{\zeta^{2k}}$, $k>0$. In other words, for any form from the space $\Xi^X$, the principal part of its pole at a simple zero of $dX/X$ should be odd with respect to the deck transformation for the function $X$ regarded locally as ramified covering with the ramification of order two at the considered point.

The very form of~\eqref{eq:dXXPsum} along with the symmetry with respect to the $X$-$Y$-variables implies

\begin{corollary}\label{cor:LLE}
The differential of $H^{(g)}_{m,n}$ with respect to any $X$-variable belongs to $\Xi^X$ and its differential with respect to any $Y$-variable belongs to $\Xi^Y$:
\begin{equation}\label{eq:LLE}
\begin{aligned}
d_iH^{(g)}_{m,n}&\in\Xi^{X_i},\quad i=1,\dots,m,\\
d_jH^{(g)}_{m,n}&\in\Xi^{Y_j},\quad j=m+1,\dots,m+m.
\end{aligned}
\end{equation}
\end{corollary}

The relations of the corollary are called (or are equivalent to what is known as) the \emph{linear loop equations}.

Consider the differential $dH^{(g)}_{m+1,n}=d_{m+1}H^{(g)}_{m+1,n}$ as a $1$-form in $X=X_{m+1}$. We know a priori that it might have poles on the diagonals  $z=z_j$, $j=m+2,\dots,m+n+1$, where $z=z_{m+1}$, and all remaining poles belong to the domain $|z|>1$. Among those poles there are the zeroes of $dX/X$. We consider them as \emph{expected poles}. Eq.~\eqref{eq:dXXPsum} implies that there could be some other poles that we call \emph{unwanted poles}. These are the poles of the function $\psi'(\Theta(z))$, that is, preimages of the poles of $\psi'$ under the mapping~$\Theta$. To be exact, we consider those preimages only that belong to the domain $|z|>1$. In fact, if the parameters $(t,s)$ are chosen small enough then the image of the circle $|z|=1$ under the map $\Theta$ belongs to a small neighborhood of the origin in the $\theta$ plane and we may assume that $\psi$ is holomorphic in this neighborhood. It follows that the unit circle separates indeed the poles of $\psi'(\Theta(z))$. Moreover, as the $(t,s)$ parameters tend to zero, both expected and unwanted poles of $d_{m+1}H^{(g)}_{m+1,n}$ converge to infinity.

The linear loop equations at the expected poles are satisfied independently of the presence or absence of the unwanted poles. We list certain cases in Section \ref{sec:Ex} when the unwanted poles cancel out. But, in general, they are present and we have no control at the moment on their principal parts.

\subsection{Higher loop equations}

The following result is an immediate consequence of Theorem~\ref{th:formalWz}:

\begin{corollary}\label{cor:HLE}
For any $g$ and $r$ the coefficient of $\hbar^{2g}w^{r-1}$ of $\cW_{m+1,n}^X(w)\frac{dX}{X}$ considered as a differential form in $X=X_{m+1}$, belongs to the space $\Xi^X$,
\begin{equation}
[\hbar^{2g}w^{r-1}]\cW_{m+1,n}^X(w)\frac{dX}{X}\in\Xi^X
\end{equation}
\end{corollary}

The relation of this corollary is called the \emph{$r$-loop equation} for the correlator functions $H^{(g)}_{n,m}$. In order to write explicitly the loop equations for small~$r$, let us introduce the following notations. For two finite sets of indices $I$, $J$ we denote
\begin{equation}
DH^{X,(g)}_{I,k,J}(X_{\bar 1},\dots,X_{\bar k})=
d_{\bar1}\dots d_{\bar k}H^{(g)}_{|I|+k,|J|}(X_I,X_{\bar 1},\dots,X_{\bar k};Y_J)
\end{equation}
with the exceptions
\begin{align}
DH^{X,(0)}_{\varnothing,1,\varnothing}(X)&=\Theta\frac{dX}{X},\\
DH^{X,(0)}_{\{i\},1,\varnothing}(X)&=\frac{z_i}{z-z_i}\frac{d z}{z}.
\end{align}
With this notation, the \emph{linear}, \emph{quadratic}, and \emph{cubic} loop equations corresponding to $r=1$, $2$, and $3$, respectively, are
\begin{equation}
DH^{X,(g)}_{M,1,N}(X)\in\Xi^X
\end{equation}
\begin{equation}
\frac{X}{dX}\Bigg(DH^{X,(g-1)}_{M,2,N}(X,X)+
\sum_{\substack{g_1+g_2=g\\
I_1\sqcup I_2=M,~J_1\sqcup J_2=N}}
DH^{X,(g_1)}_{I_1,1,J_1}(X)\;DH^{X,(g_2)}_{I_2,1,J_2}(X)
\Bigg)\in\Xi^X
\end{equation}
\begin{multline}
\Bigl(\frac{X}{dX}\Bigr)^2\Bigg(DH^{X,(g-2)}_{M,3,N}(X,X,X)\\+
3\sum_{\substack{g_1+g_2=g-1\\
I               _1\sqcup I_2=M,~J_1\sqcup J_2=N}}
DH^{X,(g_1)}_{I_1,1,J_1}(X)\;DH^{X,(g_2)}_{I_2,2,J_2}(X,X)\\+
\sum_{\substack{g_1+g_2+g_3=g\\
I_1\sqcup I_2\sqcup I_3=M,~J_1\sqcup J_2\sqcup J_3=N}}
DH^{X,(g_1)}_{I_1,1,J_1}(X)\;DH^{X,(g_2)}_{I_2,1,J_2}(X)\;DH^{X,(g_3)}_{I_3,1,J_3}(X)
\Bigg)\\
+\frac14\bigl(D_X^2-1\bigr)DH^{X,(g-1)}_{M,1,N}
\in\Xi^X
\end{multline}
The last summand of the last expression itself belongs to $\Xi^X$. However, we include it to the cubic loop equation just in the way it appears in~$[\hbar^{2g}w^{2}]\cW_{m+1,n}^X(w)\frac{dX}{X}$. 

Note that the linear and quadratic loop equations imply the blobbed topological recursion for, separately, the $H^{g}_{m,0}$ functions, and the $H^{g}_{0,n}$ functions (under the assumptions of meromorphy and generality).

\section{Projection property and topological recursion}\label{sec:projprop}
\subsection{Projection property}
In this section we assume that
\begin{align} \label{eq:phi}
\phi(\theta)&=e^{\psi(\theta)}=\hat\phi(\theta)\big|_{\hbar=0}=R(\theta)e^{P(\theta)},\\ \label{eq:phihat}
\hat\phi(\theta)&=e^{\hat\psi(\theta)}=R(\theta)e^{\cS(\hbar\partial_{\theta})P(\theta)},
\end{align}
where $R$ is a rational function and $P$ is a polynomial. We also further assume the \emph{generality condition}, i.e. that all zeroes and poles of $R(\theta)$ and all zeroes of $P(\theta)$ are simple. 

We prove here the following theorem
\begin{theorem}\label{th:projpr}
	For $\hat{\phi}$ as in~\eqref{eq:phihat}, under the generality condition,	
	\begin{equation}
		H^{(g)}_{m,n}(X(z_1),\dots,X(z_m),Y(z_{m+1}),\dots,Y(z_{m+n}))
	\end{equation} has no poles in $z_i$ for $1\leq i\leq m$ apart from the poles at the zeroes of $dX(z_i)$ and apart from the diagonal poles at $z_i=z_j$ where $m+1\leq j\leq m+n$; and it has no poles in $z_j$ for $m+1\leq j\leq m+n$ apart from the poles at the zeroes of $dY(z_j)$ and apart from the diagonal poles at $z_j=z_i$ where $1\leq i\leq m$.
\end{theorem}
The proof will go by induction in $g$ and $m+n$. Let us assume the induction hypothesis (we refer to it as IH below), i.e. that this statement holds for all $H^{g'}_{m',n'}$ such that either $g'<g$ or simultaneously $g'=g$ and $m'+n'<m+n$.

\begin{proposition}\label{prop:sumnopoles}
	Under the above assumption, the expression
	\begin{align}
		&H^{(g)}_{m+1,n}((X(z_1),\dots,X(z_{m+1}),Y(z_{m+2}),\dots,Y(z_{m+n+1})))\\\nonumber
		+&H^{(g)}_{m,n+1}((X(z_1),\dots,X(z_m),Y(z_{m+1}),\dots,Y(z_{m+n+1})))
	\end{align}
has no poles in $z=z_{m+1}$ apart from the poles at the zeroes of $dX(z)$ or the zeroes of $dY(z)$, or the diagonal poles at $z=z_i$, $i\neq m+1$.
\end{proposition}
In order to proceed with the proof of Proposition~\ref{prop:sumnopoles}, let us prove several technical statements first.

Recall equation \eqref{eq:mairecX}:
\begin{align}
	&H^{(g)}_{m+1,n}+H^{(g)}_{m,n+1}=\\ \nonumber
	&\phantom{=}-[\hbar^{2g}]\sum_{j=1}^\infty (Y\partial_Y)^{j-1} \Bigg(\frac{Y}{dY}\frac{dX}{X}\sum_{r=0}^\infty [v^j]L_r(v,\Theta)[u^r]e^{-u\,\Theta}\cW^X_{m+1,n}(u)\\ \nonumber
	&\phantom{=}-\delta_{m+n,0}[v^{j+1}]L_0(v,\Theta)Y\partial_Y\Theta\Bigg)\\ \nonumber
	&\phantom{=}+\delta_{m+n,0}\int[\hbar^{2g}]\left.\left(\dfrac{1}{\cS(\hbar\partial_\theta)}\hat\psi(\theta)-\psi(\theta)\right)\right|_{\theta=\Theta}\partial_Y\Theta\, dY+{\rm const}.
\end{align}
where $X=X(z)$ and $Y=Y(z)$.

Recall formula \eqref{eq:Lr} and substitute the expressions \eqref{eq:phi} and \eqref{eq:phihat} for $\psi$ and $\hat{\psi}$ respectively:
\begin{align}\label{eq:Lrsubs}
	L_r(v,\theta)
	&=(\partial_\theta+v\,\psi'(\theta))^re^{v\frac{\cS(v\hbar\partial_\theta)}{\cS(\hbar\partial_\theta)}\hat\psi(\theta)-v\psi(\theta)}\\ \nonumber
	&=(\partial_\theta+v\,\psi'(\theta))^re^{v\left(\frac{\cS(v\hbar\partial_\theta)}{\cS(\hbar\partial_\theta)}-1\right)\log R(\theta)}e^{v\left(\cS(v\hbar\partial_\theta)-1\right)P(\theta)}.
\end{align}

Denote
\begin{align}\label{eq:taugmndef}
	\sigma^{(g)}_{m,n} &:= [\hbar^{2g}]\sum_{j=1}^\infty (X\partial_X)^{j-1} \frac{X}{dX}\frac{dY}{Y}\sum_{r=0}^\infty [v^j]L_r(v,\Theta)[u^r]e^{-u\,\Theta}\cW^Y_{m,n+1}(u)\\ \nonumber
	&=[\hbar^{2g}]\sum_{\substack{j\geq 1\\r \geq 0}} (X\partial_X)^{j-1} \frac{d\log Y}{d \log X}\\ \nonumber
	&\phantom{=}\times
	[v^j]\left((\partial_\theta+v\,\psi'(\theta))^re^{v\left(\frac{\cS(v\hbar\partial_\theta)}{\cS(\hbar\partial_\theta)}-1\right)\log R(\theta)}e^{v\left(\cS(v\hbar\partial_\theta)-1\right)P(\theta)}\right)\Bigg|_{\theta=\Theta(z)}[u^r]e^{-u\,\Theta}\cW^Y_{m,n+1}(u)
\end{align}
and
\begin{align}\label{eq:taugmndef2}
	\check\sigma^{(g)}_0 &:= [\hbar^{2g}]\sum_{j=1}^\infty (X\partial_X)^{j-1} [v^{j+1}]L_0(v,\Theta)X\partial_X\Theta\\ \nonumber
	&=[\hbar^{2g}]\sum_{j=1}^\infty (X\partial_X)^{j-1}
	[v^{j+1}]\left(e^{v\left(\frac{\cS(v\hbar\partial_\theta)}{\cS(\hbar\partial_\theta)}-1\right)\log R(\theta)}e^{v\left(\cS(v\hbar\partial_\theta)-1\right)P(\theta)}\right)\Bigg|_{\theta=\Theta(z)}X\partial_X\Theta
\end{align}
and
\begin{align}
	\tilde\sigma^{(g)}_{00} &:=\int[\hbar^{2g}]\left.\left(\dfrac{1}{\cS(\hbar\partial_\theta)}\hat\psi(\theta)-\psi(\theta)\right)\right|_{\theta=\Theta}\partial_Y\Theta\, dY\\ \nonumber
	&=\int[\hbar^{2g}]\left.\left(\left(\dfrac{1}{\cS(\hbar\partial_\theta)}-1\right)\log R(\theta)\right)\right|_{\theta=\Theta}\partial_Y\Theta\, dY\\ \nonumber
	&=\left([u^{2g}]\dfrac{1}{\cS(u)}\right)\left.\left(\partial_\theta^{2g-1}\log R(\theta)\right)\right|_{\theta=\Theta(z)}.
\end{align}
Then we have
\begin{equation}
	H^{(g)}_{m+1,n}+H^{(g)}_{m,n+1}=\sigma^{(g)}_{m,n}-\delta_{m+n,0}\left(\check\sigma^{(g)}_0+\tilde\sigma^{(g)}_{00}\right)+{\rm const}.
\end{equation}


Since $\Theta(z)$ has the form \eqref{eq:thetaform}, we see that, apart from the ``allowed'' poles, $\sigma^{(g)}_{m,n}$ and $\check\sigma^{(g)}_{0}$ (and thus $H^{(g)}_{m+1,n}+H^{(g)}_{m,n+1}$) can only possibly have poles at $z=0$, $z=\infty$ or at the zeroes or poles of $R(\Theta(z))$.

\begin{lemma} \label{lem:Fam1PfSecondLemma} $\sigma^{(g)}_{m,n}$ is regular at the zeroes and poles of $R(\Theta(z))$ which are not equal to $\infty$ and lie outside the unit circle on the $z$-plane.
\end{lemma}

\begin{proof} The proof is similar to the proof of \cite[Lemma~4.4]{bychkov2020topological}.
	
	Let $B$ be a zero of $R(\Theta(z))$, $B\notin \{0,\infty\}$, and furthermore let $|B|>1$. The proof for the case of a pole (not coinciding with $0$ and $\infty$) of $R(\Theta(z))$ is analogous.
	
	 Note that in this case we can write $\sigma^{(g)}_{m,n}$ 
	as
	\begin{align}\label{eq:WgnIreg}
		&\sigma^{(g)}_{m,n} = [\hbar^{2g}]\sum_{\substack{j\geq 1\\r \geq 0}} \left(X\partial_X\right)^{j-1}\Bigg(\frac{d\log Y}{d \log X}\\ \nonumber
		&\phantom{=}\times[v^j]\left((\partial_\theta+v\,\psi'(\theta))^re^{v\left(\frac{\cS(v\hbar\partial_\theta)}{\cS(\hbar\partial_\theta)}-1\right)\log R(\theta)}e^{v\left(\cS(v\hbar\partial_\theta)-1\right)P(\theta)}\right)\Bigg|_{\theta=\Theta(z)} \cdot
		\mathrm{reg}_r\Bigg),
	\end{align}
	where 
	$\mathrm{reg}_r$ is some expression regular in $z$ at $z=B$.
	
	From 
	the conditions of generality 
	there exists exactly one root $A$ of the numerator of $R(\theta)$ 
	such that $B$ is a root of equation $\Theta(z)=A$. Then note that
	\begin{equation}
		e^{v\left(\frac{\cS(v\hbar\partial_\theta)}{\cS(\hbar\partial_\theta)}-1\right)\log R(\theta)}e^{v\left(\cS(v\hbar\partial_\theta)-1\right)P(\theta)} = e^{v\left(\frac{\cS(v\,\hbar\,\partial_\theta)}{\cS(\hbar\,\partial_\theta)}-1\right)\log\left(\theta-A\right)} \cdot \mathrm{reg},
	\end{equation}
	where $\mathrm{reg}$ is regular in $\theta$ at $A$, and the pole in $z$ at $B$ in the whole expression for 
	$\sigma^{(g)}_{m,n}$
	can only come from the $e^{v\left(\frac{\cS(v\,\hbar\,\partial_\theta)}{\cS(\hbar\,\partial_\theta)}-1\right)\log\left(\theta-A\right)}$ part (after the substitution $\theta=\Theta(z)$). According to \cite[Lemma~4.1]{bychkov2020topological}, we can then rewrite
	\begin{align}\label{eq:Lv1}	
		& (\partial_\theta+v\,\psi'(\theta))^re^{v\left(\frac{\cS(v\hbar\partial_\theta)}{\cS(\hbar\partial_\theta)}-1\right)\log R(\theta)}e^{v\left(\cS(v\hbar\partial_\theta)-1\right)P(\theta)} \\ \notag
		& = (\partial_\theta+v\,\psi'(\theta))^r\left(\sum_{k=1}^\infty\hbar^{2k}v(v-1)(v-2)\cdots(v-2k+1) \dfrac{p_{k}(v)}{(\theta-A)^{2k}}+\widetilde{\mathrm{reg}}\right),
	\end{align}
	where $p_k(v)$ are some polynomials in $v$ and $\widetilde{\mathrm{reg}}$ is regular in $\theta$ at $A$. Note that
	\begin{equation}
		\psi'(\theta) = \dfrac{1}{\theta-A}+\mathrm{reg}^\psi_A,
	\end{equation}
	where $\mathrm{reg}^\psi_A$ is regular in $\theta$ at $A$. Thus
	\begin{equation}
		(\partial_\theta+v\,\psi'(\theta))\dfrac{1}{(\theta-A)^{2k}}= (v-2k) \dfrac{1}{(\theta-A)^{2k+1}} + O\left(\dfrac{1}{(\theta-A)^{2k}}\right).
	\end{equation}
	This means that we can rewrite \eqref{eq:Lv1} as
	\begin{equation}
		\sum_{k=1}^\infty\hbar^{2k}\sum_{l=1}^\infty v(v-1)(v-2)\cdots(v-l+1) \dfrac{\tilde{p}_{k,l}(v)}{(\theta-A)^{l}}+\widetilde{\widetilde{\mathrm{reg}}},
	\end{equation}
	where $\tilde{p}_{k,l}(v)$ are some polynomials in $v$ and $\widetilde{\widetilde{\mathrm{reg}}}$ is regular in $\theta$ at $A$. Now let us plug $\theta=\Theta(z)$ into this expression and substitute it into \eqref{eq:WgnIreg}. Note that for $z\rightarrow B$
	\begin{equation}
		\dfrac{1}{(\Theta(z)-A)^l} = \dfrac{C}{(z_1-B)^l}+O\left(\dfrac{1}{(z_1-B)^{l-1}}\right)
	\end{equation}
	for some constant $C$.
	
	Since $|B|>1$, $d\log X(z)$ has a simple pole at $B$, while $d \log Y(z)$ is regular at $B$, according to \eqref{eq:Xpoles} and \eqref{eq:Ypoles} respectively. Thus, $d \log Y/d\log X$ has a simple zero at $z_1=B$.
	
	This means 
	that Equation~\eqref{eq:WgnIreg} can be rewritten as
	\begin{align}
		&\sigma^{(g)}_{m,n}
		= 
		[\hbar^{2g}]\sum_{r=0}^\infty\sum_{j=1}^\infty \left(X\partial_X\right)^{j-1}
		[v^j]\, \left(v\sum_{k=2}^\infty\hbar^{k}\sum_{l=1}^\infty (v-1)(v-2)\cdots(v-l+1) \dfrac{q_{r,k,l}(v)}{(z-B)^{l-1}}+\widetilde{\mathrm{reg}}_r\right),
	\end{align}
	where $q_{r,k,l}$ are some expressions polynomial in $v$
	and $\widetilde{\mathrm{reg}}_r$ is regular in $z$ at $B$.
	Taking the sum over $j$ we can rewrite this expression as
	\begin{equation}\label{eq:WgnID}
		\sigma^{(g)}_{m,n} =
		[\hbar^{2g}]\sum_{r=0}^\infty \left(\sum_{k=2}^\infty\hbar^{k}\sum_{l=1}^\infty q_{r,k,l}(X\partial_X)\cdot (X\partial_X-1)(X\partial_X-2)\cdots(X\partial_X-l+1) \dfrac{1}{(z-B)^{l-1}}+\widetilde{\mathrm{reg}}_r\right).
	\end{equation}
	%
	Now note that, since $X\partial_X=(d\log X/dz)^{-1}\partial_z$ and $|B|>1$, then, again according to \eqref{eq:Xpoles},
	\begin{align}\label{eq:XDXser}
		X\partial_X = \left(-(z-B) + O\left((z-B)^2\right)\right)	 \partial_{z}
	\end{align}
	at $z\rightarrow B$, and, therefore, by an easy inductive argument,
	\begin{align}
		(X\partial_X-1)(X\partial_X-2)\cdots(X\partial_X-l+1) \dfrac{1}{(z-B)^{l-1}}
	\end{align}
	is regular at $z\rightarrow B$ for any $l\geq 1$. This implies that \eqref{eq:WgnID} is regular at $z\rightarrow B$.
\end{proof}

\begin{lemma} \label{lem:Fam1PfSecondLemma2} $\check\sigma^{(g)}_{0}+\tilde\sigma^{(g)}_{00}$ is regular at the zeroes and poles of $R(\Theta(z))$ which are not equal to $\infty$ and lie outside the unit circle on the $z$-plane.
\end{lemma}
\begin{proof}	
	As in the proof of~Lemma~\ref{lem:Fam1PfSecondLemma}, let $B$ be a zero of $R(\Theta(z))$, $B\notin \{0,\infty\}$, and furthermore let $|B|>1$. The proof for the case of a pole (not coinciding with $0$ and $\infty$) of $R(\Theta(z))$ is analogous.
	
	Note that the expression $\check\sigma^{(g)}_{0}+\tilde\sigma^{(g)}_{00}$ has a pole at $z=B$ if and only if $X\partial_X(\check\sigma^{(g)}_{0}+\tilde\sigma^{(g)}_{00})$ has a pole there. Indeed, recall~\eqref{eq:XDXser}:
	\begin{align}
		X\partial_X = \left(-(z-B) + O\left((z-B)^2\right)\right)	 \partial_{z}\qquad \text{for } z\rightarrow B,
	\end{align}
	which directly implies  that the operator $X\partial_X$
	 preserves the degree of the pole at $B$ for any function.  We have:
	\begin{align}
		X\partial_X(\check\sigma^{(g)}_{0}+\tilde\sigma^{(g)}_{00}) &= [\hbar^{2g}]\sum_{j=0}^\infty (X\partial_X)^{j} [v^{j+1}]L_0(v,\Theta)X\partial_X\Theta\\ \nonumber
		&=[\hbar^{2g}]\sum_{j=0}^\infty (X\partial_X)^{j}
		[v^{j+1}]\left(e^{v\left(\frac{\cS(v\hbar\partial_\theta)}{\cS(\hbar\partial_\theta)}-1\right)\log R(\theta)}e^{v\left(\cS(v\hbar\partial_\theta)-1\right)P(\theta)}\right)\Bigg|_{\theta=\Theta(z)}X\partial_X\Theta.
	\end{align}
	
	From that point the proof becomes analogous to the proof of Lemma~\ref{lem:Fam1PfSecondLemma}.
	Just note that while the term $d\log Y/d \log X$ (which had a simple zero at $z\rightarrow B$) is absent, we have an extra $X\partial_X$ which accounts for the same effect.
\end{proof}

\begin{lemma} \label{lem:checksigma-infty}
	$\check\sigma^{(g)}_{0}+\tilde\sigma^{(g)}_{00}$ is regular at $z\rightarrow\infty$.
\end{lemma}

\begin{proof}
	It is clear that $\tilde\sigma^{(g)}_{00}$ is regular at $z\rightarrow\infty$. Let us prove that $\check\sigma^{(g)}_{0}$ is regular there as well.
	
	We recall from \eqref{eq:taugmndef} that $\check\sigma^{(g)}_{0}$ is equal to the coefficient of $\hbar^{2g}$ in the expansion of
	\begin{align}\label{eq:sigmacheck2}
		\sum_{j=1}^\infty (X\partial_X)^{j-1}
		[v^{j}]\left(\dfrac{1}{v}\,e^{v\left(\frac{\cS(v\hbar\partial_\theta)}{\cS(\hbar\partial_\theta)}-1\right)\log R(\theta)}e^{v\left(\cS(v\hbar\partial_\theta)-1\right)P(\theta)}\right)\Bigg|_{\theta=\Theta(z)}X\partial_X\Theta
	\end{align}
	and count the order of pole of this expression at $z\to \infty$.
	
	We begin with a few observations:
	\begin{itemize}
		\item The factor $e^{v\left(\frac{\cS(v\hbar\partial_\theta)}{\cS(\hbar\partial_\theta)}-1\right)\log R(\theta)}|_{\theta=\Theta(z)}$ has no pole at $z\to \infty$. 
		\item Note that $d\log\phi/dz=R'(\Theta(z))\Theta'(z)/R(\Theta(z)) + P'(\Theta(z))\Theta'(z)$ has a pole of order $(e\deg P -1)$ at $z\rightarrow\infty$. Thus, taking into account \eqref{eq:dlogxdlogY} and the fact that $Y(z)$ has a simple zero at $z\rightarrow \infty$, we can conclude that each application of $X\partial_X=(d\log X/dz)^{-1}\partial_{z}$ decreases the degree of the pole at $z\to \infty$ by $e\deg P$. The total effect of the operator $(X\partial_X)^{j-1}$ and of the separate occurrence of $X\partial_X$ is the decrease of the order of pole by $je\deg P$.
	\end{itemize}
	Hence, the order of pole of~\eqref{eq:sigmacheck2} at $z\to \infty$ is equal to the order of pole at $z\to \infty$ of the following expression:
	\begin{align} \label{eq:infinity-expr2-orderpole-B}
		& \Bigg(\dfrac{1}{v}
		 e^{v\left(\cS(v\,\hbar\,\partial_\theta)-1\right)P(\theta)
		}\Big|_{\theta=\Theta(z)}\Theta\Bigg)\Bigg|'_{v=z^{-e\deg P}},
	\end{align}	
	where by $|'$ we mean that we only select the terms with $\deg v \geq 2$ before the substitution $v=z^{-e\deg P}$.
	Note also that
	\begin{itemize}
		\item Since $\Theta(z)$ has a pole of order $e$ at $z\to\infty$, each application of $\partial_\theta$ decreases the order of pole in the resulting expression by $e$ (until the total number of $\partial_\theta$ exceeds $\deg P$).
	\end{itemize}	
	
	Note that $\cS(z)$ is an even series starting with $1$, and the first terms of the expansion of $e^{v\left(\cS(v\,\hbar\,\partial_\theta)-1\right)P(\theta)}$ in $v$ have the following form:
	\begin{equation}\label{eq:expexp}
		1+\dfrac{1}{24}v^3\hbar^2\partial_\theta^2 P(\theta),
	\end{equation}
	with all the other terms having only lower degree of the pole at $z\rightarrow \infty$. Since  $\Theta(z)$ has a pole of order $e$ at $z\to\infty$, after the substitution $v=z^{-e\deg P}$ the term
	\begin{equation}
		\dfrac{1}{24}v^2\hbar^2 P''(\Theta)\, \Theta
	\end{equation}
becomes regular at $z\rightarrow \infty$, and all the other terms as well (we disregard  the term $1$ from~\eqref{eq:expexp} due to the $\deg v \geq 2$ requirement).
		

	So,  \eqref{eq:infinity-expr2-orderpole-B} is regular at $z\to\infty$, and, therefore, $\check\sigma^{(g)}_{0}$ is regular at $z\to\infty$ as well.
\end{proof}

\begin{lemma} \label{lem:Fam1-tau1-infty}
	$\sigma^{(g)}_{m,n}$ is regular at $z\rightarrow\infty$.
\end{lemma}

\begin{proof}
	
	The proof is similar to the proof of Lemma~\ref{lem:checksigma-infty}. 
	
	We recall from \eqref{eq:taugmndef} that $\sigma^{(g)}_{m,n}$ is equal to the coefficient of $\hbar^{2g}$ in the expansion of
	\begin{align}\label{eq:taugmn2}
		&\sum_{\substack{j\geq 1\\r \geq 0}} (X\partial_X)^{j-1} \frac{d\log Y}{d \log X}\\ \nonumber
		&\phantom{=}\times
		[v^j]\left((\partial_\theta+v\,\psi'(\theta))^re^{v\left(\frac{\cS(v\hbar\partial_\theta)}{\cS(\hbar\partial_\theta)}-1\right)\log R(\theta)}e^{v\left(\cS(v\hbar\partial_\theta)-1\right)P(\theta)}\right)\Bigg|_{\theta=\Theta(z)}[u^r]e^{-u\,\Theta}\cW^Y_{m,n+1}(u)
	\end{align}
	and count the order of pole of this expression at $z\to \infty$.
	
	We begin with a few observations:
	\begin{itemize}
		\item Similar to what happened in the previous Lemma, the factor $e^{-u\,\Theta}\cW^Y_{m,n+1}(u)$ has no pole at $z\rightarrow \infty$.
		\item The factor $e^{v\left(\frac{\cS(v\hbar\partial_\theta)}{\cS(\hbar\partial_\theta)}-1\right)\log R(\theta)}|_{\theta=\Theta(z)}$ has no pole at $z\to \infty$ and cannot acquire one when acted upon by further operators $\partial_\theta$. This follows from the generality condition.
		\item Note that $d\log\phi/dz=R'(\Theta(z))\Theta'(z)/R(\Theta(z)) + P'(\Theta(z))\Theta'(z)$ has a pole of order $(e\deg P -1)$ at $z\rightarrow\infty$. Thus, taking into account \eqref{eq:dlogxdlogY} and the fact that $Y(z)$ has a simple zero at $z\rightarrow \infty$, we can conclude that the factor $d \log Y/d \log X$ has zero of order $e\deg P$ at $z\to \infty$ and each application of $X\partial_X=(d\log X/dz)^{-1}\partial_{z}$ decreases the degree of the pole at $z\to \infty$ by $e\deg P$. The total effect of the factor $d \log Y/d \log X$ and of $(X\partial_X)^{j-1}$ is the decrease of the order of pole by $je\deg P$.
	\end{itemize}
Hence, the order of pole of~\eqref{eq:taugmn2} at $z\to \infty$ is equal to the order of pole at $z\to \infty$ of the following expression:
\begin{align} \label{eq:infinity-expr2-orderpole-B2}
& \sum_{r=0}^\infty\Bigg(
(\partial_\theta+v\,\psi'(\theta))^r e^{v\left(\cS(v\,\hbar\,\partial_\theta)-1\right)P(\theta)
}\Big|_{\theta=\Theta(z)}\Bigg)\Bigg|'_{v=z^{-e\deg P}},
\end{align}	
where by $|'$ we mean that we only select the terms with $\deg v \geq 1$ before the substitution $v=z^{-e\deg P}$.
Note also that
\begin{itemize}
\item Since $\Theta(z)$ has a pole of order $e$ at $z\to\infty$, each $\partial_\theta$ decreases the order of pole in the resulting expression by $e$.
\item Multiplication by $\psi'(\theta)$ increases the order of pole by $e(\deg P-1)$.
\end{itemize}	
Taking into account these two observations and that each $v$ factor decreases the order of pole by $e\deg P$, we see that each application of the operator $\partial_\theta + v\psi'(\theta)$ decreases the order of pole in the resulting expression by $e$.

Since $P(\theta)$ comes together with at least one power of $v$, and any further application of $\partial_\theta$ or $\partial_\theta + v\psi'(\theta)$ only decreases the order of the pole, we see that this expression is regular at $z\rightarrow \infty$.



\end{proof}	


Let us proceed to the proof of Proposition~\ref{prop:sumnopoles}.

\begin{proof}[Proof of Proposition~\ref{prop:sumnopoles}]
	Lemmas~\ref{lem:Fam1PfSecondLemma}--\ref{lem:Fam1-tau1-infty}, 
	proved with the help of Equation~\eqref{eq:mairecX}, imply that the expression
	\begin{equation}
		H^{(g)}_{m+1,n}(z_1,\dots,z_{m+n+1})+H^{(g)}_{m,n+1}(z_1,\dots,z_{m+n+1})
	\end{equation}
has no ``unwanted'' poles outside the unit disk.

In a completely analogous way, using Equation~\eqref{eq:mairecY} instead of Equation~\eqref{eq:mairecX}, one can prove that there are no ``unwanted'' poles \emph{inside} the unit disk.

 This completes the proof of the proposition.	
\end{proof}

Finally we are ready to prove Theorem~\ref{th:projpr}.

\begin{proof}[Proof of Theorem~\ref{th:projpr}]
	The reasoning given in Section~\ref{sec:separ} implies that if (as given by Proposition~\ref{prop:sumnopoles})
	\begin{align}
		&H^{(g)}_{m+1,n}(X(z_1),\dots,X(z_{m+1}),Y(z_{m+2}),\dots,Y(z_{m+n+1}))\\\nonumber
		+&H^{(g)}_{m,n+1}(X(z_1),\dots,X(z_{m}),Y(z_{m+1}),\dots,Y(z_{m+n+1}))
	\end{align}
has no ``unwanted'' poles in $z=z_{m+1}$ and can only have diagonal poles and poles at the zeroes of $dX(z)$ and $dY(z)$, then $H^{(g)}_{m+1,n}$ can only have poles at the zeroes of $dX(z)$ and diagonal poles at $z=z_j$ for $j\geq m+2$; similarly for $H^{(g)}_{m,n+1}$. Due to the symmetry between the  $X$-arguments and (separately) the $Y$-arguments of $H^g_{m+1,n}$ (and $H^g_{m+1,n}$ as well), we have thus proved the induction step of Theorem~\ref{th:projpr}.


This completes the proof of the theorem.	
\end{proof}
\subsection{Topological recursion}
Projection property together with the linear and quadratic loop equations imply the \emph{spectral curve topological recursion} for, separately, the set of functions $H^{(g)}_{m,0}$ and the set of functions $H^{(g)}_{0,n}$: 

\begin{theorem}
	For $\hat{\phi}$ as in~\eqref{eq:phihat}, under the generality condition, the functions $H^{(g)}_{m,0}$ satisfy the spectral curve topological recursion on the spectral curve
	\begin{align}		
	x(z)=1/X(z),\quad y(z)=-X(z)\,\Theta(z),
\end{align}
i.e. with $\omega^{(0)}_{1}=y\, dx=\Theta\tfrac{dX}{X}$ and the functions $H^{(g)}_{0,n}$ satisfy the spectral curve topological recursion on the spectral curve
\begin{align}	
x(z)=1/Y(z),\quad y(z)=-Y(z)\,\Theta(z)
\end{align}
i.e. with $\omega^{(0)}_{1}=y\, dx=\Theta\tfrac{dY}{Y}$.
\end{theorem}
\begin{remark}
	The generality condition can actually be lifted, but one has to replace the topological recursion with the \emph{Bouchard-Eynard recursion} of~\cite{BE13}. This can be done in a completely analogous way to how it is discussed in \cite[Section~5.1]{bychkov2020topological}.
\end{remark}

\section{Specializations and examples}
\label{sec:Ex}

\subsection{Specialization \texorpdfstring{$t=0$}{t=0}}\label{sec:t0}
The above computations are applied for arbitrary choices of $(t,s,\hat\psi)$ (both in the formal and analytic settings). We wish now to compare our general results with some special cases known in the literature. In the case $t=0$ the equation of the spectral curve reduces to the following:
\begin{equation}
	Y=z^{-1},\qquad \Theta(z)=\sum_{i=1}^e s_i z^i,\qquad X(z) \phi(\Theta(z))=z.
\end{equation}
This is exactly the spectral curve that was considered in \cite{ACEH} in the case when $\phi$ is a polynomial and in~\cite{bychkov2021explicit,bychkov2020topological} for more general~$\phi$.

Remarkably, the correlator functions $H_{0,n}^{(g)}$ in this case are known explicitly: they are all identically equal to zero! The only nonzero terms that contribute to the inductive formula~\eqref{eq:mairecX} for $m=0$ are the singular corrections in~\eqref{eq:cT}. Applying the induction starting from the known $H_{0,n}^{(g)}=0$, we obtain subsequently $H_{1,n-1}^{(g)}$, $H_{2,n-2}^{(g)}$, and so on, and finally $H_{n,0}^{(g)}$. This computation of $H_{n,0}^{(g)}$ with all steps collected into a single expression coincides with that of~\cite{bychkov2021explicit} where it is formulated as a summation over the set of all connected graphs with $n$ labeled vertices. We just haven't mentioned in~\cite{bychkov2021explicit} that these `intermediate' functions obtained in the course of computations have the meaning of correlator functions $H^{(g)}_{m,n}$ with both $m$ and~$n$ positive.

Remark that even for this case $t=0$ the relation of Theorem~\ref{th:mainrecursion} provides an alternative (and even probably more efficient) way of computing these correlator functions. Namely, let us apply just one step of induction in the inverse direction, in order to express $H^{(g)}_{n,0}+H^{(g)}_{n-1,1}$ by~\eqref{eq:mairecY} in terms of the functions~$H^{(g')}_{n',0}$ with smaller values of $(g',n')$. Then, applying the procedure of separating poles described at the end of Sect.~\ref{sec:separ} we find $H^{(g)}_{n-1,1}$ and also $H^{(g)}_{n,0}$ as a byproduct.

Let us provide an explicit formula for (the derivative of) $H^{(g)}_{m,n}$ for this $t=0$ case. We have

\begin{proposition} \label{prop:t0} Consider the case $t=0$. Recall that in this case
\begin{equation}
		Y=z^{-1},\qquad \Theta(z)=\sum_{i=1}^e s_i z^i,\qquad X(z) \phi(\Theta(z))=z.
	\end{equation}

	Let
	\begin{equation}
	W^{(g)}_{m,n}:=\prod_{i=1}^m \left(X_i\partial_{X_i}\right)\prod_{j=m+1}^n \left(Y_j\partial_{Y_j}\right)H^{(g)}_{m,n}.	
\end{equation}
	
	Then for $2g-2+m+n>0,\,m+n>1$ we have
	\begin{align}\label{eq:mainprop}
		&W^{(g)}_{m,n}(X(z_1),\dots,X(z_m),Y(z_{m+1}),\dots,Y(z_{m+n}))\\ \nonumber &\phantom{==}=[\hbar^{2g
		}]
		U_m\dots U_1
		\sum_{\gamma \in \Gamma_{m,n}}\prod_{\{v_k,v_\ell\}\in E_\gamma} w_{k,\ell}\prod_{\{v_i,\widetilde{v}_j\}\in \widetilde{E}_\gamma} \widetilde{w}_{i,j},
	\end{align}
	where the sum is over all connected simple graphs $\gamma$ on $m$ labeled vertices $v_1,\dots,v_m$ with $n$ additional leaves (i.e. 1-valent vertices) $\widetilde{v}_1,\dots,\widetilde{v}_n$; $E_\gamma$ is the set of normal edges (i.e. connecting the vertices $v_1,\dots,v_m$), $\widetilde{E}_\gamma$ is the set of edges for which one of the endpoints is an additional leaf,
	\begin{align}\label{eq:wkl}
		w_{k,\ell}&:=e^{
			u_ku_\ell\cS(u_k\hbar\,z_k\partial_{z_k})\cS(u_\ell\hbar\,z_\ell\partial_{z_\ell})
			\frac{z_k z_\ell}{(z_k-z_\ell)^2}}-1,\\
		\widetilde{w}_{i,j}&:=u_i\cS(u_i\hbar\,z_i\partial_{z_i})	\frac{z_i z_j}{(z_i-z_j)^2},
	\end{align}
	and $U_{i}$ is the operator acting on a function~$f$ in~$u_i$ and $z_i$ by
	\begin{align}\label{eq:Uihbar}
		U_{i} f&=
		\sum_{j,r=0}^\infty \left(X_i\partial_{X_i}\right)^j\left(\frac{L^j_{r,i}
		}{Q_i}
		[u_i^r] \frac{e^{u_i(\cS(u_i\,\hbar\,z_i\partial_{z_i})-1)\Theta(z_i)}}{u_i
			\,\cS(u_i\,\hbar)}f(u_i,z_i)\right),
	\end{align}
	where
	\begin{align}
		\label{eq:Lr2}
		L^j_{r,i} &:= \left.\left([v^j]e^{-v\,\psi(\theta)}\partial_\theta^r e^{v\frac{\cS(v\,\hbar\,\partial_\theta)}{\cS(\hbar\,\partial_\theta)}\hat \psi(\theta)}\right)\right|_{\theta=\Theta(z_i)}\\ \nonumber
		&\phantom{:}=\left.\left([v^j]\left(\partial_\theta+v\psi'(\theta)\right)^r e^{v\left(\frac{\cS(v\hbar\partial_\theta)}{\cS(\hbar\partial_\theta)}\hat{\psi}(\theta)-\psi(\theta)\right)}\right)\right|_{\theta=\Theta(z_i)}
	\end{align}
	and $Q_i=Q(z_i)$, where
	\begin{equation}
		Q(z):=-\frac{Y(z)}{dY(z)}\frac{dX(z)}{X(z)}=1-z\, \Theta'(z)\,\psi'(\Theta(z)).
	\end{equation}
\end{proposition}
For $m=1,\, n=0$ there is an additional term in the formula, but we do not write this case here for brevity, as it is completely covered by \cite{bychkov2021explicit,bychkov2020topological}.
It is also straightforward to write an explicit formula for the $H^{(g)}_{m,n}$ function itself (rather than $W^{(g)}_{m,n}$), but it is bulkier and we do not list it here for brevity. 

It is also useful to consider an even deeper specialization $t=0$, $s=0$. In this case $\Theta=0$ and the above formulas get simplified. The spectral curve becomes $Y=z^{-1}$, $X=z$. All $H_{0,n}^{(g)}$ and $H_{m,0}^{(g)}$ vanish, but $H_{m,n}^{(g)}$ where $m$ and $n$ are both nonzero are interesting, since they correspond to the weighted double Hurwitz numbers with all preimages of $0$ and $\infty$ being marked (which is equivalent to absence of any markings at all), and this way we get a rather simple formula for their generating functions.

\subsection{Specialization \texorpdfstring{$\phi(\theta)=1+\theta$}{phi(theta)=1+theta}}
Specialization $\phi(\theta)=1+\theta$ or $\psi(\theta)=\log(\phi(\theta))=\log(1+\theta)$ corresponds to enumeration of maps/bicolored maps, see discussion in Sect.~\ref{sec:separ}. In the case of just maps we have further specialization $s=(0,1,0,\dots)$ i.e. $s_k=\delta_{k,2}$. One of the ways to write the equation of spectral curve for this case that can be found in the literature (see, for example,~\cite{CEO,EynardBook}) is
\begin{equation}
	\begin{aligned}
		x(z)&=a+c\Bigl(z+\frac1z\Bigr),\\
		y(z)&=\left[x(z)-\sum_{k=1}^dt_kx(z)^{k-1}\right]_{\le0},
	\end{aligned}
\end{equation}
where $\left[\;\cdot\;\right]_{\le0}$ denotes the polar part of the corresponding Laurent polynomial in~$z$. The constants~$a$ and~$c$ are functions in~$t$ parameters determined by two implicit algebraic equations that can be written in the form
\begin{equation}
	[z^0]\,y(z)=0,\qquad [z^{-1}]\,y(z)=\frac1{c},\qquad a|_{t=0}=0,\quad c_{t=0}=1.
\end{equation}

The generating series enumerating maps is the power expansion of the corresponding correlator functions in the local coordinate $X=x^{-1}$ at the point $z=0$ on the spectral curve. In this case, that is for $\phi(\theta)=1+\theta$, $s_k=\delta_{2,k}$ the above equations produce the same spectral curve as the one introduced in Sect.~\ref{sec:spectral} with the identifications
\begin{equation}
	X=\frac1{x},\qquad Y=\frac1{x-y},\qquad \Theta=\frac{1}{X Y}-1=x\,(x-y)-1.
\end{equation}
It is an exercise to check that all requirements of Definition~\ref{def:spectralcurve} are satisfied. Remark that with this identification we have
\begin{equation}
	\Theta\frac{dX}{X}=y\,dx+\Bigl(\frac1x-x\Bigr) dx=dH^{(0)}_{1,0}-\sum_{k-1}^d t_kx^{k-1}dx.
\end{equation}
We conclude that $\Theta\frac{dX}{X}$, $y\,dx$, and $dH^{(0)}_{1,0}$ are three possible equivalent choices for the form~$\omega^{(0)}_1$ of topological recursion, since they differ by meromorphic differentials in~$x$.

\medskip
In the case of bicolored maps the parameters $s_k$ are chosen arbitrarily. The equation of the spectral curve for this case can be found in~\cite{EynardBook}. With a small change of notation it can be formulated by saying that $x=\frac1X$ and $y=\frac1Y$ are Laurent polynomials of the form
\begin{equation}
	\begin{aligned}
		x(z)&=\gamma z^{-1}+\sum_{k=0}^{e-1}\alpha_k z^k,\\
		y(z)&=\gamma z+\sum_{k=1}^{d-1}\beta_k z^{-k}
	\end{aligned}
\end{equation}
uniquely determined by the requirements
\begin{equation}
	\begin{aligned}
		-\sum_{k=1}^{d}t_k x(z)^k&=-x(z)\,y(z)+1+O(z),\\
		-\sum_{k=1}^{e}s_k y(z)^k&=-x(z)\,y(z)+1+O(z^{-1}).
	\end{aligned}
\end{equation}
We see that these conditions become equivalent to those of Definition~\ref{def:spectralcurve} if we set $x\,y=\Theta-1$ which exactly corresponds to the case $\phi(\theta)=1+\theta$. Remark that our choice of the form $\omega^{(0)}_{1}=\Theta\frac{dX}{X}$ of topological recursion becomes in this case
\begin{equation}
	\Theta\frac{dX}{X}=-(1+x\,y)\frac{dx}{x}=-y\,dx-\frac{dx}{x}.
\end{equation}
The second summand does nor affect the recursion and the first summand differs from the conventional form $y\,dx$ by a sign. Therefore, the recursions for both choices for $\omega^{(0)}_{1}$ produce the same differentials, up to a sign convention.

\begin{remark}
	In most of the sources devoted to topological recursion for (bicolored) maps the authors make an additional restriction $t_1=t_2=0$ forbidding $1$-gons and $2$-gons for the maps. This restriction comes from the matrix models where these terms are usually absorbed via a suitable change of variables in the matrix integrals.  Nevertheless, the combinatorial problem itself is well posed with the presence of $1$ and $2$-gons, the equation of the spectral curve is resolved independently of the values of $t_1$ and $t_2$ and the topological recursion holds. In fact, with the approach of the present paper based on the VEVs and operator formalism the role of $t_1$ and $t_2$ variables does not differ from the role of any other $t$-variable.
\end{remark}

\subsection{Computations for small \texorpdfstring{$g$}{g} and \texorpdfstring{$n$}{n}}

\subsubsection{\texorpdfstring{$(0,3)$}{(0,3)}-case}

In the case $g=0$, $m+n+1=3$ both~\eqref{eq:mairecX} and~\eqref{eq:mairecY} applied to  functions $H^{(0)}_{3,0}(X(z_1),X(z_2),X(z_3))$, \dots, $H^{(0)}_{0,3}(Y(z_1),Y(z_2),Y(z_3))$ give
\begin{align*}
H^{(0)}_{0,3}+H^{(0)}_{1,2}&=f(z_1)\tfrac{z_1 z_1}{(z_2-z_1) (z_3-z_1)}+c_1(z_1,z_3),\\
H^{(0)}_{1,2}+H^{(0)}_{2,1}&=-f(z_2)\tfrac{z_1 z_2}{(z_2-z_1) (z_2-z_3)}+c_2(z_1,z_3),\\
H^{(0)}_{2,1}+H^{(0)}_{3,0}&=f(z_3)\tfrac{z_1 z_2}{(z_3-z_1) (z_3-z_2)}+c_3(z_1,z_2),
\end{align*}
where
\begin{equation}
f(z)=-\psi'(\Theta(z))\tfrac{X(z)}{z\partial_zX(z)}\tfrac{Y(z)}{z\partial_zY(z)}.
\end{equation}
Let us represent
\begin{equation}
f(z)=f^+(z)+f^-(z)+c
\end{equation}
where $f^+$ is holomorphic in $|z|<1$, $f^-$ is holomorphic in $|z|>1$, and $f^+(0)=f^-(\infty)=0$. Then, applying the procedure of splitting of poles to either of the above three formulas we obtain, explicitly,
\begin{align*}
H^{(0)}_{3,0}&=
      f^+(z_1)\tfrac{z_2 z_3}{(z_1-z_2) (z_1-z_3)}
     +f^+(z_2)\tfrac{z_3 z_1}{(z_2-z_3) (z_2-z_1)}
     +f^+(z_3)\tfrac{z_1 z_2}{(z_3-z_1) (z_3-z_2)},\\[1ex]
H^{(0)}_{2,1}&=
     -f^+(z_1)\tfrac{z_1 z_2}{(z_1-z_2) (z_1-z_3)}
     -f^+(z_2)\tfrac{z_1 z_2}{(z_2-z_3) (z_2-z_1)}
 +(f^-(z_3)+c)\tfrac{z_1 z_2}{(z_3-z_1) (z_3-z_2)},\\[1ex]
H^{(0)}_{1,2}&=
  (f^++c)(z_1)\tfrac{z_1^2  }{(z_1-z_2) (z_1-z_3)}
     -f^-(z_2)\tfrac{z_1 z_2}{(z_2-z_3) (z_2-z_1)}
     -f^-(z_3)\tfrac{z_1 z_3}{(z_3-z_1) (z_3-z_2)},\\[1ex]
H^{(0)}_{0,3}&=
      f^-(z_1)\tfrac{z_1^2  }{(z_1-z_2) (z_1-z_3)}
     +f^-(z_2)\tfrac{z_2^2  }{(z_2-z_3) (z_2-z_1)}
     +f^-(z_3)\tfrac{z_3^2  }{(z_3-z_1) (z_3-z_2)}.
\end{align*}

\subsubsection{\texorpdfstring{$(1,1)$}{(1,1)}-case}
Let $\hat\psi(\theta)=\psi(\theta)+\psi_1(\theta)\hbar^2+O(\hbar^4)$. Denote
\begin{multline}
W^{X,(0)}_{2,0}(X,X)=X'\partial_{X'}(X\partial_X H^{(0}_{2,0})\Bigm|_{X'=X}=\\
\left(\tfrac{X(z)}{dX(z)}\right)^2\lim_{z'\to z}\left(\tfrac{dz\,dz'}{(z-z')^2}-\tfrac{dX(z)\,dX(z')}{(X(z)-X(z'))^2}\right)
=\tfrac{X(z)^2 \left(3 X''(z)^2-2 X^{(3)}(z) X'(z)\right)}{12 X'(z)^4}.
\end{multline}
Then, in the notations of Theorem~\ref{th:mainrecursion}, we have
\begin{equation}
e^{-u\,\Theta}\cW^X_{1,0}(u,X)=\frac1{u}+\left(\frac{u}{2}W^{X,(0)}_{0,2}(X,X)-\frac{u}{24}+\frac{u^2}{24}(X\partial_X)^2\Theta\right)\hbar^2+O(\hbar^4).
\end{equation}
We have also
\begin{equation}
\begin{aligned}
L_0(v,\theta)&=1+v\Bigl(\frac{v^2-1}{24}\psi''(\theta)+\psi_1(\theta)\Bigr)\,\hbar^2+O(\hbar^4),\\
L_1(v,\theta)&=v\,\psi'(\theta)+O(\hbar^2),\\
L_2(v,\theta)&=v\,\psi''(\theta)+v^2\psi'(\theta)^2+O(\hbar^2).
\end{aligned}
\end{equation}

Therefore, applying~\eqref{eq:mairecY}, for example, we obtain in this case
\begin{multline}\label{eq:omega11}
H^{(0)}_{1,0}+H^{(0)}_{0,1}=
\sum_{j=0}^2 (Y\partial_Y)^{j}[v^j] \biggl(\tfrac{Y}{dY}\tfrac{dX}{X}\Bigl(
v \psi '(\Theta ) \bigl(\tfrac1{24}-\tfrac1{2}W^{X,(0)}_{2,0}(X,X)\bigr)\\
-\tfrac1{24}\bigl(v^2 \psi '(\Theta )^2+v \psi ''(\Theta )\bigr)(X\partial_X)^2\Theta\Bigr)
+  \tfrac{v^2-1}{24} \psi '(\Theta)+\int\psi_1(\theta)d\theta\Bigm|_{\theta=\Theta}\biggr)
\end{multline}
As usual, the two summands on the left hand side are identified by the separating poles procedure.

\subsection{One more example}

Along with the cases $t=0$ and $\phi(\theta)=1+\theta$ known in the literature and discussed above let us consider just the next simplest (apparently new) case when $\phi$ is quadratic and only $t_1$ and $s_1$ are different from zero. Namely, let
\begin{equation}
\begin{aligned}
{\mathbf t}&=(t,0,0,\dots),\\
{\mathbf s}&=(s,0,0,\dots),\\
\phi(\theta)&=1+c\,\theta+\theta^2.
\end{aligned}
\end{equation}
Then we obtain, explicitly,
\begin{equation}
\begin{aligned}
\frac1X&=\tfrac{1}{1-s t}(z^{-1}+c s+s^2 z),\\
\frac1Y&=\tfrac{1}{1-s t}(z+c t+t^2z^{-1}),\\
\Theta&=\tfrac{1}{1-s t}(s\,z+c s t+t\,z^{-1})
\end{aligned}
\end{equation}
(our choice for the normalization of the coordinate $z$ differs here from~\eqref{eq:Xzcoef}). Indeed, one checks by direct substitution that with this choice of $\Theta$ the Laurent polynomial $\phi(\Theta)$ factorizes as~\eqref{eq:XYTheta}, and, moreover,
\begin{equation}
\Theta=t X^{-1}+s z=s Y^{-1}+t z^{-1}
\end{equation}
which agrees with~\eqref{eq:ThetaX}--\eqref{eq:ThetaY}. We have
\begin{equation}
\begin{aligned}
X\partial_X&=\frac{1+c s z+s^2 z^2}{(1-s z) (1+s z)}z\partial_z,\\
Y\partial_Y&=-\frac{z^2+c t z+t^2}{(z-t) (t+z)}z\partial_z.
\end{aligned}
\end{equation}
The first operator has poles at $z=\pm s^{-1}$. These poles are zeroes of the form $\frac{dX}{X}$ and they converge to infinity as the $(s,t)$-parameters tend to zero.
The second operator has poles at $z=\pm t$. These poles are zeroes of the form $\frac{dY}{Y}$ and they converge to zero as the $(s,t)$-parameters tend to zero. This observation allows one to single out the summands on the left hand sides of~\eqref{eq:mairecX} and~\eqref{eq:mairecY}. Applying computations for the cases $(g,m+n)=(0,3)$ and $(1,1)$ above we find explicitly, for example,
\begin{align}
H^{(0)}_{3,0}&
=\frac{(1+c/2) s^3 z_1 z_2 z_3}{(1-s z_1) (1-s z_2) (1-s z_3)}
+\frac{(1-c/2) s^3 z_1 z_2 z_3}{(1+s z_1) (1+s z_2) (1+s z_3)},\\
H^{(0)}_{0,3}&
=\frac{(1+c/2) t^3}{(z_1-t) (z_2-t) (z_3-t)}
+\frac{(1-c/2) t^3}{(t+z_1) (t+z_2) (t+z_3)},\\
H^{(1)}_{1,0}(X)&
=\frac{1}{24} \left(X\partial_X\Bigl(\frac{s z}{1-s z}+\frac{s z}{1+s z}\Bigr)
-\Bigl(\frac{(1+c) s z}{1-s z}+\frac{(1-c) s z}{1+s z}\Bigr)\right),\\
H^{(1)}_{0,1}(Y)&
=\frac{1}{24} \left(Y\partial_Y\Bigl(\frac{t}{z-t}+\frac{t}{z+t}\Bigr)
-\Bigl(\frac{(1+c) t}{z-t}+\frac{(1-c) t}{z+t}\Bigr)\right).
\end{align}
Remark that the topological recursion is applicable as well and leads to the same answers.

\bibliographystyle{alphaurl}
\bibliography{FullySimple}

\newcommand{\etalchar}[1]{$^{#1}$}
\begin{thebibliography}{BCGF{\etalchar{+}}21b}

\bibitem[ABDB{\etalchar{+}}23]{ABDKS-FSTR}
Alexander Alexandrov, Boris Bychkov, Petr Dunin-Barkowski, Maxim Kazarian, and
  Sergey Shadrin.
\newblock Topological recursion, symplectic duality, and generalized fully
  simple maps, 2023.
\newblock URL: \url{https://arxiv.org/abs/2304.11687}, \href
  {https://doi.org/10.48550/arXiv.2304.11687}
  {\path{doi:10.48550/arXiv.2304.11687}}.

\bibitem[ACEH20]{ACEH}
A.~Alexandrov, G.~Chapuy, B.~Eynard, and J.~Harnad.
\newblock Weighted {H}urwitz numbers and topological recursion.
\newblock {\em Comm. Math. Phys.}, 375(1):237--305, 2020.
\newblock \href {https://doi.org/10.1007/s00220-020-03717-0}
  {\path{doi:10.1007/s00220-020-03717-0}}.

\bibitem[ALS16]{ALS}
A.~Alexandrov, D.~Lewanski, and S.~Shadrin.
\newblock Ramifications of {H}urwitz theory, {KP} integrability and quantum
  curves.
\newblock {\em J. High Energy Phys.}, (5):124, front matter+30, 2016.
\newblock \href {https://doi.org/10.1007/JHEP05(2016)124}
  {\path{doi:10.1007/JHEP05(2016)124}}.

\bibitem[BCCGF22]{BCCG}
Valentin Bonzom, Guillaume Chapuy, Séverin Charbonnier, and Elba
  Garcia-Failde.
\newblock Topological recursion for {O}rlov-{S}cherbin tau functions, and
  constellations with internal faces, 2022.
\newblock URL: \url{https://arxiv.org/abs/2206.14768}, \href
  {https://doi.org/10.48550/ARXIV.2206.14768}
  {\path{doi:10.48550/ARXIV.2206.14768}}.

\bibitem[BCGF21a]{BCGF21}
Ga\"{e}tan Borot, S\'{e}verin Charbonnier, and Elba Garcia-Failde.
\newblock Topological recursion for fully simple maps from ciliated maps, 2021.
\newblock \href {https://arxiv.org/abs/2106.09002} {\path{arXiv:2106.09002}}.

\bibitem[BCGF{\etalchar{+}}21b]{BCGFLS-freeprob}
Gaëtan Borot, Séverin Charbonnier, Elba Garcia-Failde, Felix Leid, and Sergey
  Shadrin.
\newblock Analytic theory of higher order free cumulants, 2021.
\newblock URL: \url{https://arxiv.org/abs/2112.12184}, \href
  {https://doi.org/10.48550/ARXIV.2112.12184}
  {\path{doi:10.48550/ARXIV.2112.12184}}.

\bibitem[BDBKS20]{bychkov2020topological}
Boris Bychkov, Petr Dunin-Barkowski, Maxim Kazarian, and Sergey Shadrin.
\newblock Topological recursion for {K}adomtsev-{P}etviashvili tau functions of
  hypergeometric type, 2020.
\newblock URL: \url{https://arxiv.org/abs/2012.14723}, \href
  {https://doi.org/10.48550/arXiv.2012.14723}
  {\path{doi:10.48550/arXiv.2012.14723}}.

\bibitem[BDBKS21]{BDKS-fullysimple}
Boris Bychkov, Petr Dunin-Barkowski, Maxim Kazarian, and Sergey Shadrin.
\newblock Generalised ordinary vs fully simple duality for $n$-point functions
  and a proof of the {B}orot--{G}arcia-{F}ailde conjecture, 2021.
\newblock To appear in Comm. Math. Phys.
\newblock URL: \url{https://arxiv.org/abs/2106.08368}, \href
  {https://doi.org/10.48550/ARXIV.2106.08368}
  {\path{doi:10.48550/ARXIV.2106.08368}}.

\bibitem[BDBKS22]{bychkov2021explicit}
Boris Bychkov, Petr Dunin-Barkowski, Maxim Kazarian, and Sergey Shadrin.
\newblock Explicit closed algebraic formulas for {O}rlov-{S}cherbin {$n$}-point
  functions.
\newblock {\em J. \'{E}c. polytech. Math.}, 9:1121--1158, 2022.
\newblock \href {https://doi.org/10.5802/jep.202} {\path{doi:10.5802/jep.202}}.

\bibitem[BE13]{BE13}
Vincent Bouchard and Bertrand Eynard.
\newblock Think globally, compute locally.
\newblock {\em J. High Energy Phys.}, (2):143, front matter + 34, 2013.
\newblock \href {https://doi.org/10.1007/JHEP02(2013)143}
  {\path{doi:10.1007/JHEP02(2013)143}}.

\bibitem[BEO15]{BEO}
Ga\"{e}tan Borot, Bertrand Eynard, and Nicolas Orantin.
\newblock Abstract loop equations, topological recursion and new applications.
\newblock {\em Commun. Number Theory Phys.}, 9(1):51--187, 2015.
\newblock \href {https://doi.org/10.4310/CNTP.2015.v9.n1.a2}
  {\path{doi:10.4310/CNTP.2015.v9.n1.a2}}.

\bibitem[BGF20]{borot2018simple}
Ga\"{e}tan Borot and Elba Garcia-Failde.
\newblock Simple maps, {H}urwitz numbers, and topological recursion.
\newblock {\em Comm. Math. Phys.}, 380(2):581--654, 2020.
\newblock \href {https://doi.org/10.1007/s00220-020-03867-1}
  {\path{doi:10.1007/s00220-020-03867-1}}.

\bibitem[BS17]{BorotSha-Blobbed}
Ga\"{e}tan Borot and Sergey Shadrin.
\newblock Blobbed topological recursion: properties and applications.
\newblock {\em Math. Proc. Cambridge Philos. Soc.}, 162(1):39--87, 2017.
\newblock \href {https://doi.org/10.1017/S0305004116000323}
  {\path{doi:10.1017/S0305004116000323}}.

\bibitem[CE06]{ChekhovEynard}
Leonid Chekhov and Bertrand Eynard.
\newblock Hermitian matrix model free energy: {F}eynman graph technique for all
  genera.
\newblock {\em J. High Energy Phys.}, (3):014, 18, 2006.
\newblock \href {https://doi.org/10.1088/1126-6708/2006/03/014}
  {\path{doi:10.1088/1126-6708/2006/03/014}}.

\bibitem[CEO06]{CEO}
Leonid Chekhov, Bertrand Eynard, and Nicolas Orantin.
\newblock Free energy topological expansion for the 2-matrix model.
\newblock {\em J. High Energy Phys.}, (12):053, 31, 2006.
\newblock \href {https://doi.org/10.1088/1126-6708/2006/12/053}
  {\path{doi:10.1088/1126-6708/2006/12/053}}.

\bibitem[DBKPS19]{DKPS}
Petr Dunin-Barkowski, Reinier Kramer, Alexandr Popolitov, and Sergey Shadrin.
\newblock Loop equations and a proof of {Z}vonkine's $qr$-{ELSV} formula, 2019.
\newblock URL: \url{https://arxiv.org/abs/1905.04524}, \href
  {https://doi.org/10.48550/ARXIV.1905.04524}
  {\path{doi:10.48550/ARXIV.1905.04524}}.

\bibitem[DBOPS18]{DOPS}
Petr Dunin-Barkowski, Nicolas Orantin, Aleksandr Popolitov, and Sergey Shadrin.
\newblock Combinatorics of loop equations for branched covers of sphere.
\newblock {\em Int. Math. Res. Not. IMRN}, (18):5638--5662, 2018.
\newblock \href {https://doi.org/10.1093/imrn/rnx047}
  {\path{doi:10.1093/imrn/rnx047}}.

\bibitem[EO07]{EynardOra-TopoRec}
B.~Eynard and N.~Orantin.
\newblock Invariants of algebraic curves and topological expansion.
\newblock {\em Commun. Number Theory Phys.}, 1(2):347--452, 2007.
\newblock \href {https://doi.org/10.4310/CNTP.2007.v1.n2.a4}
  {\path{doi:10.4310/CNTP.2007.v1.n2.a4}}.

\bibitem[EO08]{EO-x-y-1st}
B.~Eynard and N.~Orantin.
\newblock Topological expansion of mixed correlations in the {H}ermitian
  2-matrix model and {$x$}-{$y$} symmetry of the {$F_g$} algebraic invariants.
\newblock {\em J. Phys. A}, 41(1):015203, 28, 2008.
\newblock \href {https://doi.org/10.1088/1751-8113/41/1/015203}
  {\path{doi:10.1088/1751-8113/41/1/015203}}.

\bibitem[EO13]{eynard2013xy}
B.~Eynard and N.~Orantin.
\newblock About the $x-y$ symmetry of the $f_g$ algebraic invariants, 2013.
\newblock \href {https://arxiv.org/abs/1311.4993} {\path{arXiv:1311.4993}}.

\bibitem[Eyn16]{EynardBook}
Bertrand Eynard.
\newblock {\em Counting surfaces}, volume~70 of {\em Progress in Mathematical
  Physics}.
\newblock Birkh\"{a}user/Springer, [Cham], 2016.
\newblock CRM Aisenstadt chair lectures.
\newblock \href {https://doi.org/10.1007/978-3-7643-8797-6}
  {\path{doi:10.1007/978-3-7643-8797-6}}.

\bibitem[GJ08]{goulden2008kp}
I.~P. Goulden and D.~M. Jackson.
\newblock The {KP} hierarchy, branched covers, and triangulations.
\newblock {\em Adv. Math.}, 219(3):932--951, 2008.
\newblock \href {https://doi.org/10.1016/j.aim.2008.06.013}
  {\path{doi:10.1016/j.aim.2008.06.013}}.

\bibitem[GPH15]{guaypaquet20152d}
Mathieu Guay-Paquet and J.~Harnad.
\newblock 2{D} {T}oda {$\tau$}-functions as combinatorial generating functions.
\newblock {\em Lett. Math. Phys.}, 105(6):827--852, 2015.
\newblock \href {https://doi.org/10.1007/s11005-015-0756-z}
  {\path{doi:10.1007/s11005-015-0756-z}}.

\bibitem[HO15]{HarnadOrlov}
J.~Harnad and A.~Yu. Orlov.
\newblock Hypergeometric {$\tau$}-functions, {H}urwitz numbers and enumeration
  of paths.
\newblock {\em Comm. Math. Phys.}, 338(1):267--284, 2015.
\newblock \href {https://doi.org/10.1007/s00220-015-2329-5}
  {\path{doi:10.1007/s00220-015-2329-5}}.

\bibitem[Hoc22]{Hock-x-y}
Alexander Hock.
\newblock On the $x$-$y$ symmetry of correlators in topological recursion via
  loop insertion operator, 2022.
\newblock URL: \url{https://arxiv.org/abs/2201.05357}, \href
  {https://doi.org/10.48550/ARXIV.2201.05357}
  {\path{doi:10.48550/ARXIV.2201.05357}}.

\bibitem[KL15]{Kazarian_2015}
M~E Kazarian and S~K Lando.
\newblock Combinatorial solutions to integrable hierarchies.
\newblock {\em Russian Mathematical Surveys}, 70(3):453--482, jun 2015.
\newblock URL: \url{https://doi.org/10.1070%2Frm2015v070n03abeh004952}, \href
  {https://doi.org/10.1070/rm2015v070n03abeh004952}
  {\path{doi:10.1070/rm2015v070n03abeh004952}}.

\bibitem[KMMM95]{Kharchev}
S.~Kharchev, A.~Marshakov, A.~Mironov, and A.~Morozov.
\newblock Generalized {K}azakov-{M}igdal-{K}ontsevich model: group theory
  aspects.
\newblock {\em Internat. J. Modern Phys. A}, 10(14):2015--2051, 1995.
\newblock \href {https://doi.org/10.1142/S0217751X9500098X}
  {\path{doi:10.1142/S0217751X9500098X}}.

\bibitem[OS01]{OrlovScherbin}
A.~Yu. Orlov and D.~M. Shcherbin.
\newblock Hypergeometric solutions of soliton equations.
\newblock {\em Teoret. Mat. Fiz.}, 128(1):84--108, 2001.
\newblock \href {https://doi.org/10.1023/A:1010402200567}
  {\path{doi:10.1023/A:1010402200567}}.

\end{thebibliography}
\end{document}